\documentclass[10pt,conference,letterpaper]{IEEEtran}
\usepackage{graphicx, amsmath, epsfig, ifthen}
\usepackage{tikz}
\usetikzlibrary{automata,positioning}
\usepackage{dsfont}
\usepackage{amsmath}
\usepackage{multirow}
\usepackage{caption}
\usepackage{subfig}
\usepackage[ruled,vlined]{algorithm2e}
\usepackage{amsfonts}
\usepackage{datetime}
\usepackage{balance}
\usepackage{url}

\newboolean{short}
\setboolean{short}{false}

\newcommand{\scream}[1]{{\bf * #1 *}{\typeout{#1}}}
\newcommand{\eat}[1]{}

\newtheorem{theorem}{Theorem}
\newtheorem{lemma}{Lemma}[section]

\newtheorem{definition}{Definition}
\newtheorem{example}{Example}[section]

\def \production#1#2{#1 \rightarrow #2}

\def \derives#1#2{#1 \Rightarrow^* #2}
\def \productionGraph#1{{\cal P}(#1)}

\def \M{${\cal M}=(Q,\delta, q_0,F)~$}

\makeatletter \let\@copyrightspace\relax \makeatother
\begin{document}

\title{Answering Regular Path Queries on Workflow Provenance}%\\ \today}
\newcommand{\superscript}[1]{\ensuremath{^{\textrm{#1}}}}
\def\wa{\superscript{1}}
\def\wb{\superscript{2}}
\def\wc{\superscript{3}}
\def\waa{\superscript{4}}
\def\waaa{\superscript{5}}
\def\sharedaffiliation{\end{tabular}\newline\begin{tabular}{c}}

\author{
 Xiaocheng Huang\wa~ Zhuowei Bao\wb~ Susan B. Davidson\wc~ Tova
 Milo\waa~ Xiaojie Yuan\waaa~\\
\begin{small}
\begin{tabular}{ccccc}
{\wa}Genome Institute of Singapore&{\wb} Facebook Inc.&{\wc}University of Pennsylvania &{\waa} Tel Aviv
  University&{\waaa} Nankai University\\
Singapore & Menlo Park, USA&   Philadelphia,USA&Tel Aviv, Israel&Tianjin, China\\
  huangxc@gis.a-star.edu.sg &zhuowei@cis.upenn.edu&susan@cis.upenn.edu&milo@cs.tau.ac.il &yuanxj@nankai.edu.cn
\end{tabular}
\end{small}
}

\eat{\numberofauthors{1}
\newcommand{\superscript}[1]{\ensuremath{^{\textrm{#1}}}}
\def\wa{\superscript{1}}
\def\wb{\superscript{2}}
\def\wc{\superscript{3}}
\def\waa{\superscript{4}}
\def\sharedaffiliation{\end{tabular}\newline\begin{tabular}{c}}
\author{
  \alignauthor{  Xiaocheng Huang\wa \eat{\thanks{The work of these authors
      is partly supported by the NSFC under grant 61170184.}}~ Zhuowei
    Bao\wb~ Susan B. Davidson\wc~ Tova Milo\waa~ Xiaojie Yuan\wa~\eat{\footnotemark[1]}}
    \sharedaffiliation
\begin{small}
  \begin{tabular}{cccc}
  \affaddr{{\wa}Nankai University}&\affaddr{{\wb} Facebook Inc.}
  &\affaddr{{\wc}University of Pennsylvania} &\affaddr{{\waa} Tel Aviv
  University}\\
 \affaddr{Tianjin, China}& \affaddr{Menlo Park, USA}&   \affaddr{Philadelphia,
   USA} &\affaddr{Tel Aviv, Israel}\\
   \affaddr{huangx@seas.upenn.edu} &\affaddr{zhuowei@cis.upenn.edu}& \affaddr{susan@cis.upenn.edu} &  \affaddr{milo@cs.tau.ac.il} \\
       \affaddr{yuanxj@nankai.edu.cn}&&&
  \end{tabular}
\end{small}
}}

\maketitle

\begin{abstract}
This paper proposes a novel approach for efficiently
evaluating regular path queries over provenance graphs of workflows that may include recursion. 
The approach assumes that an execution $g$ of a workflow  
$G$ is labeled with {\em query-agnostic} reachability
labels using an existing technique.  
At query time, given $g$, $G$ and a regular path query $R$, the
approach decomposes $R$ into a set of subqueries $R_1$, ..., $R_k$ that are  {\em safe} for 
$G$.   For each safe subquery $R_i$, $G$ is rewritten so that, using the reachability labels of nodes in $g$, whether or not there is a path which matches $R_i$ between
two nodes can be decided in constant time.   The results of each safe subquery are then
composed, possibly with some small unsafe remainder, to produce an answer to $R$.    The approach results in an algorithm that significantly reduces the number of subqueries $k$ over existing techniques by increasing their size and complexity, 
and that evaluates each subquery in time bounded by its input and output size.
Experimental results demonstrate the benefit of this approach.
\end{abstract}

%\category{H.2.8}{Database Management}{Database Applications}[scientific databases]
%\terms{Algorithms, Performance, Theory}
%\keywords{Dynamic Labeling Scheme, Reachability Query, Recursive Workflow}

\vspace{-2mm}
\section{Introduction}
\label{sec:intro}
Capturing and querying workflow provenance is
increasingly important for scientific as well as business
applications. By maintaining information about the sequence of
module executions used to produce data, as well as the
parameter settings and intermediate data passed between module
executions, the validity and reproducibility of data can be
enhanced.

A series of ``provenance challenges"~\footnote{{\small http://twiki.ipaw.info/bin/view/Challenge/}} was held
between 2006 and 2010 to compare the expressiveness of various
provenance systems. Many of the sample queries given in these
challenges were simple {\em reachability queries} that check
the existence of an (arbitrary) execution path between workflow
nodes, e.g. ``Identify the data sources that contributed some data
leading to the production of publication $p$".
However, others were
more complex, requiring the path between nodes to have a
certain  shape.

Such constraints on the path structure can naturally be captured by
regular expressions. For example, the query ``Find all publications
$p$ that resulted from starting with data of type $x$, then
performing a repeated analysis using either technique $a_1$ or
technique $a_2$, terminated by producing a result of type $s$, and
eventually ending by publishing $p$." can be captured as the {\em
regular path query}  $R=\; x.(a_1|a_2)^+.s.\_^*.p$.

As users become familiar with the power of provenance,
such complex, regular path queries will become even more common.  
In particular, they are necessary to find workflows that exhibit certain
types of behaviors within shared repositories of workflows and their
executions, a topic of increasing interest within the scientific
community~\cite{Cohen-Boulakia:2011:SAR:2034863.2034865,DBLP:conf/ssdbm/StarlingerBL12}.

\eat{what is the state of the art:
1.  if query is simple (single label, reachability) know how to do
for reachability, solved: people use reachability labels
  2. complicated quereis, not easy
- example of a complicated path query  - how do people typically
traverse graph
cut into smaller queries and compose using joins which are slow

what we do is an alternative of new labeling that can handle in one shot the whole regular expression
turns out can't do for every query, but for a large class of "safe" queries it is possible, show to extend to general queries
build on previous work on reachability labels:  labelling has two parts run and spec, run labels are done on the fly
novel trick is to rewrite the spec based on query, then have reduced problem to previous labeling (fine-granined)
not trivial, prove that by this rewriting we have correctly reduced
even for qureeis that are not safe, we propose a decomposition algorithm.  unlike traditional decomp that breaks into atomic, we break into big things and save a lot of joins
}

Answering regular path queries over graphs (in particular,
XML trees) has  been extensively studied~\cite{DBLP:conf/icde/Al-KhalifaJPWKS02,DBLP:conf/sigmod/BrunoKS02,DBLP:conf/vldb/ChenGK05,DBLP:conf/sigmod/TrisslL07,DBLP:journals/www/WangLWL08,DBLP:journals/pvldb/ZengH12}.  The
typical approach used is to cut the query into smaller subqueries
(e.g. reachability queries), and traverse the graph to answer each
subquery.  The results of the subqueries are then joined together
to answer the original query.  The
problem with this approach is the large number and size of 
intermediate results, and the subsequent cost of joins.
 In this paper, we show that since workflow executions are {\em
not arbitrary graphs}, but rather graphs that originate from a {\em
given specification},  regular path queries can be processed much more efficiently.
Specifically, we show that a regular path query {\em does not need to be
decomposed} when it is {\em safe} for a given workflow
specification.   Safe queries are quite general, and go well
beyond reachability queries.

Before discussing our solution, note that regular
path queries, such as the one presented above, cannot be answered
simply by looking at a workflow {\em specification}. This is because
1) the queries may involve run-time data; and 2) if the workflow
specification contains alternatives then the exact paths
between data may not be known in advance. For example, if a workflow
specification $G$, which takes something of type $x$ as input,
involves a choice of either executing $a_1$ repeatedly followed by
$s$ and terminating with $p$ (which matches $R$), or executing $a_3$
repeatedly followed by $s$ and terminating with $p$ (which doesn't
match $R$), to answer the query one needs to examine which option
was actually taken at run time. Nevertheless, we will see that the
specification can still be used to speed up query processing.

\begin{figure}
\vspace{-2mm} \hspace{-10mm}
\includegraphics[width=1.2\columnwidth]{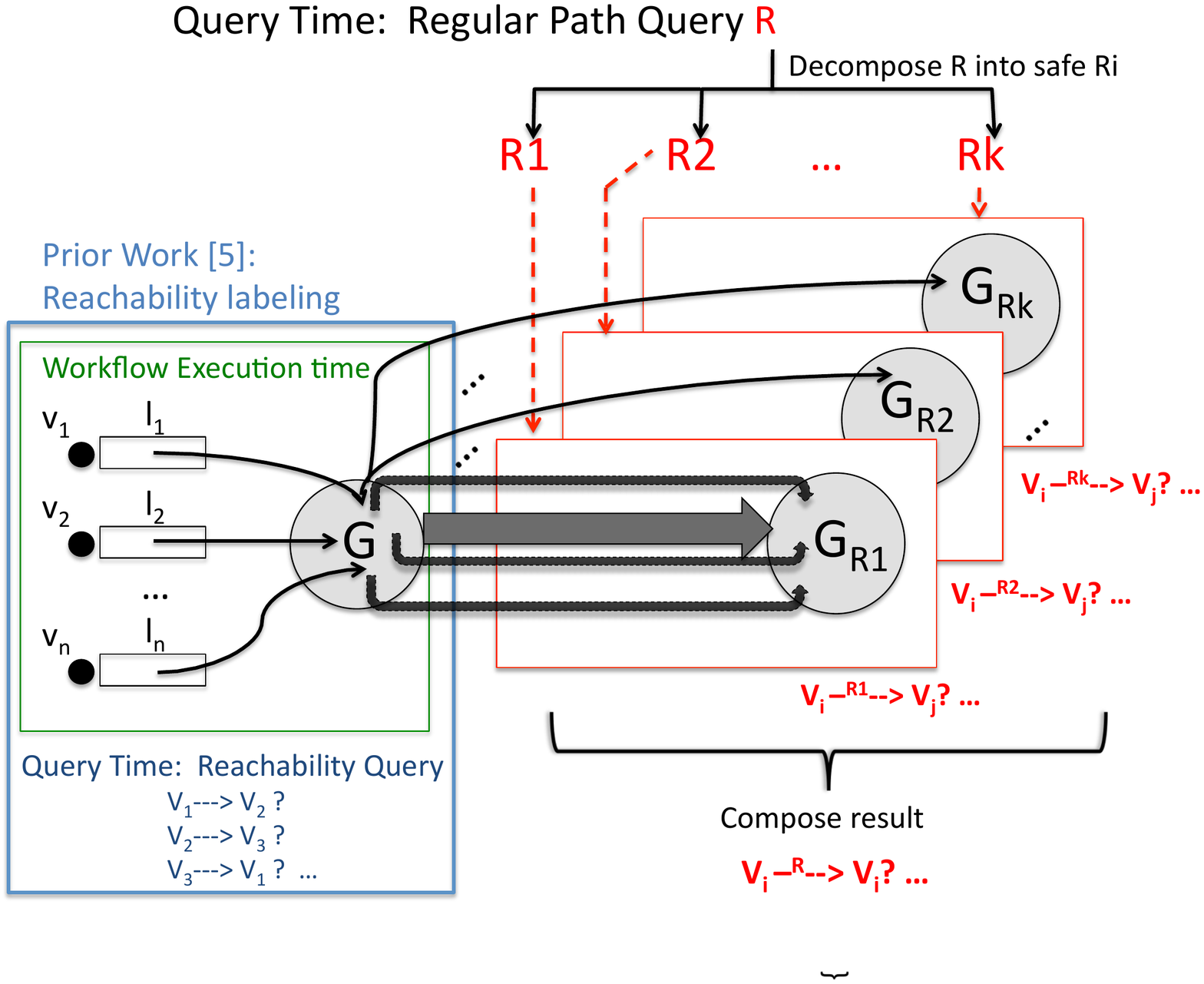}
\vspace{-12mm}
\caption{Overview of our approach} \label{fig:overview} \vspace{-6mm}
\end{figure}

The approach we present in this paper, illustrated in Fig.~\ref{fig:overview},  assumes that an execution $g$ of a  workflow specification 
$G$ is labeled with query-agnostic,  {\em reachability}
labels using an existing technique~\cite{DBLP:conf/sigmod/BaoDM12}
(left portion of Fig.~\ref{fig:overview}; each $v_i$ is a node in
$g$, and its reachability label $l_i$ references the specification $G$).
At query time, given  $g$, $G$ and a regular path query $R$, the
approach decomposes $R$ into a set of subqueries $R_1$, ..., $R_k$ which are  {\em safe} for  the specification
$G$ .    For each safe subquery $R_i$, $G$ is rewritten so that, using the reachability labels of nodes in $g$, 
whether or not there is a path which matches $R_i$ between two nodes $u$
and $v$ can be decided in constant time ($G_{R_1}$,...$G_{R_k}$ in Fig.~\ref{fig:overview}).   The results of each safe subquery are then composed, possibly with some small unsafe remainder, 
 to produce an answer to $R$.%  (bottom right portion of Fig.~\ref{fig:overview}).

The reason this solution is possible is because the labeling in \cite{DBLP:conf/sigmod/BaoDM12} is parameterized by the specification.  The novelty in this paper is to rewrite the specification so that reachability labeling can be used to evaluate regular path expressions.

The benefit  of this approach is 1) the number of subqueries is much smaller than previous approaches; 2) subqueries are larger and more complex than the atomic subqueries used previously, reducing the number and size of intermediate results; and 3) there are therefore fewer, potentially less expensive joins, hence overall a significant
speedup is achieved.   The type of executions (provenance graphs) 
for which this can be done are produced by a very general
class of workflows expressed as {\em context-free graph
grammars}~\cite{DBLP:conf/vldb/BeeriEKM06}, and which handle
recursion in a reasonable way ({\em strictly linear recursive}~\cite{DBLP:conf/sigmod/BaoDM12}).

\eat{The approach that we use is to label nodes in the provenance graph
{\em dynamically}, i.e. as they are produced in an execution, where
the labels encode the sequence of grammar productions used to
produce the object. The labeling is {\em query agnostic}, and builds
on a technique for answering reachability queries presented
in~\cite{DBLP:conf/sigmod/BaoDM11}. At query time, the
workflow specification from which the execution was generated is
rewritten using the regular path query $R$, so that given the labels
of two nodes $u$ and $v$, and the rewritten specification, one is
able to answer {\em pairwise} queries efficiently, e.g. ``Is there a
path which matches $R$ between $u$ and $v$?" The labels assigned to
the nodes are {\em compact}, i.e. logarithmic in the size of the
workflow execution, and {\em efficient}, i.e. pairwise queries can
be answered in constant time given two labels. We then show how to
use this to answer queries over sets of workflow nodes $U$ and $V$
({\em all-pairs } queries) to find all pairs $(u,v) \in U \times V$
having an $R$-shaped path between them. In particular, the cost of
rewriting the specification is now amortized over all the tested
pairs.}

{\bf Related Work.}
The model of workflow provenance that we adopt in this paper is based on \cite{davidson@deb2007,opm2011}.  
The problem of labeling workflow runs to answer reachability queries was studied in \cite{DBLP:conf/sigmod/HeinisA08}, where they adapted the interval-based labeling for trees to work for DAGs. 
\eat{The basic idea is to transform a DAG into a tree by replicating the vertices which have more than one incoming edges; a similar idea was proposed in GRIPP~\cite{DBLP:conf/sigmod/TrisslL07}.}
The problem with this approach is that the size of the transformed tree can be exponential in the size of the original DAG, which leads to linear-size interval labels.  Furthermore, it does not extend to regular path queries since it is not parameterized by the specification.
\eat{ It is worth mentioning that [46] observed that large provenance graphs originate from a small specification through loop (sequential) and fork (parallel) executions of sub-workflows. Based on that, they discussed several optimization techniques to reduce the size of the transformed tree, however, their approach does not fully take advantage of this information.}

Reachability queries have also been extensively studied for XML trees and graphs, and  a common approach is to use labeling.
%more recently for workflow provenance graphs.  
%The most common approach is to use labeling.
% (see~\cite{DBLP:conf/sigmod/BaoDM11} for a discussion). 
An algorithm for\eat{answering} all-pairs reachability queries over {\em trees} 
is given in~\cite{DBLP:conf/icde/Al-KhalifaJPWKS02}, which executes in time linear in the input and output size and is therefore optimal.  
\cite{DBLP:conf/sigmod/BrunoKS02} gives an optimal algorithm for XML pattern matching. 
However,  existing work on all-pairs reachability queries on DAGs/graphs
cannot achieve linear time complexity \eat{in the input and output size}~\cite{DBLP:conf/vldb/ChenGK05,DBLP:conf/sigmod/TrisslL07,DBLP:journals/www/WangLWL08,DBLP:journals/pvldb/ZengH12}. 
\eat{Since in this paper we consider
DAGs that are executions of strictly linear recursive workflows, we are able to develop an optimal algorithm for all-pairs
reachability queries.}

%\scream{We should also mention the work of provenance graph labeling Heinis  Thomas Heinis and Gustavo Alonso. Efficient lineage tracking for scientific workflows. In
%SIGMOD Conference, pages 1007Ð1018, 2008.  [74] Silke Tri§l and Ulf Leser. Fast and practical indexing and querying of very large graphs. In
%SIGMOD Conference, pages 845Ð856, 2007.}

\eat{Answering (pairwise) regular path queries is known to be NP-complete when (i) the graphs have cycles;
and (ii) the matching paths are required to be
simple~\cite{DBLP:conf/vldb/MendelzonW89}. } Pairwise
regular path queries on DAGs can be answered in time linear in graph size~\cite{DBLP:conf/vldb/MendelzonW89}.
\eat{Since
provenance graphs are acyclic, using depth-first search we can
answer pairwise regular path queries in time linear in the size of the
graph.}
Two optimization techniques (query pruning and query rewriting) which use graph schemas~\cite{DBLP:conf/icdt/BunemanDFS97} are proposed in~\cite{DBLP:conf/icde/FernandezS98}.\eat{:  {\em Query pruning} (which restricts path search to a fragment of the graph) and {\em query rewriting} (which rewrites the regular expression
to an equivalent but simpler one).} \cite{DBLP:conf/vldb/LiM01} proposes
to decompose regular expressions into concatenation/union/Kleene star 
subexpressions, and then uses reachability labeling to perform
joins. 
Recently \cite{DBLP:conf/ssdbm/KoschmiederL12} proposes to use rare labels to decompose queries to smaller subqueries and perform a breadth-first search in parallel. \cite{Dey:2013:IPR:2457317.2457353} proposes multiple regular query variants and represents queries as datalog.
\cite{DBLP:conf/icdt/LibkinV12} considers querying both data
and the topology of graphs. 
 \cite{DBLP:conf/pods/LosemannM12} considers regular expressions with numerical occurrence
indicators. Regular expressions of special forms have also been
recently studied~\cite{\eat{DBLP:journals/corr/abs-1203-2886,}DBLP:conf/sigmod/JinHWRX10}.\eat{:
Label constrained reachability (LCR) queries are studied in
\cite{DBLP:conf/sigmod/JinHWRX10}, and label order constrained
reachability queries (LOCR) are studied in
\cite{DBLP:journals/corr/abs-1203-2886}. }
Languages for path queries over graph-structured data are surveyed in
\cite{DBLP:journals/tods/BarceloLLW12}; among them,
\cite{DBLP:conf/vldb/LiM01,DBLP:conf/vldb/MendelzonW89} {\color{red} \cite{DBLP:conf/ssdbm/KoschmiederL12}} can be extended to our setting. We
will show a comparison to this in the experiments.

Most relevant for this paper are the dynamic reachability labeling
techniques of
\cite{DBLP:conf/sigmod/BaoDM11, DBLP:conf/sigmod/BaoDM12,DBLP:conf/sigmod/BaoDKR10}
for workflow provenance graphs, which address reachability queries between
a single pair of nodes.
In contrast, this paper addresses considerably more complex queries, regular path queries, between sets of nodes. 
To do this, we harness in a non-trivial way the labeling techniques in \cite{DBLP:conf/sigmod/BaoDM12},
and employ them for processing general queries over workflow provenance.
\eat{However, there is no existing work on regular path labeling
for the reason that it is query-dependent. In this paper, we extend
the reachability labeling scheme in \cite{DBLP:conf/sigmod/BaoDM12} to
regular path labeling, which is also compact.}

\eat{{\noindent \bf Mary Fernandez and Dan Suciu, ICDE'1998~\cite{DBLP:conf/icde/FernandezS98}:\\ Optimizing Regular Path Expressions Using Graph Schemas}

This paper considers the problem of optimizing {\em regular path queries} over semi-structured data. The optimization techniques utilize the underlying {\em graph schema}~\cite{DBLP:conf/icdt/BunemanDFS97} which describes the partial knowledge of the structure of data graphs. More precisely, the correspondence between graph schema and data graphs is defined via {\em graph simulation}. This is similar to our notion of specification and executions in the context of workflows. However, we define such graph simulation in terms of the derivation of a {\em context-free graph grammar} (rather than {\em graph homomorphism}). Their application scenario is to query data graphs that represent the structure of intranet web sites whose nodes denote pages and whose edges denote hyperlinks. The administrator who is familiar to the structure of intranet web sites provides the graph schema that specifies a high-level summary of the entire data graph and possible expansions.

Two families of optimization techniques are proposed in this paper. The first is {\em query pruning} which restricts path search to a fragment of the graph. This is done by pruning the given regular path expression to a more restricted but equivalent form by exploiting the knowledge of graph schema. The other is {\em query rewriting using state extents}. The idea is that, by rewriting the given regular path query to equivalent logical plans, we will be able to start the search deeper in the graph. In some cases, we may even avoid any graph traversal in the query evaluation. While the algorithm runs in exponential time to find all possible logical plans, for special forms of regular path queries, the algorithm runs in polynomial time. This paper also gives an efficient approximation algorithm that works for all regular path queries.}

\eat{Other related work~\cite{DBLP:conf/icdt/Abiteboul97, DBLP:conf/vldb/GoldmanW97} studied the problem of regular path expressions in semi-structured databases, and focused  on deriving and using schema information to rewrite queries and guide the search. \cite{DBLP:conf/icdt/MiloS99} considers indexing.}

\eat{
\medskip

{\noindent \bf Quanzhong Li, Bongki Moon, VLDB'2001~\cite{DBLP:conf/vldb/LiM01}:\\ Indexing and Querying XML Data for Regular Path Expressions}

This paper studies the problem of indexing XML data to support efficient regular path queries. The idea is to first use a labeling scheme (similar to {\em interval encoding}, but reserve the numbers for future insertions) to enable efficient ancestor-descendant queries, and then decompose a regular path query to a combination of three types of subexpressions: (1) element joins element; (2) element joins attribute and (3) Kleene closure of a subexpression. For each of them, this paper proposes path-join algorithms, namely, {\em EE-Join}, {\em EA-Join} and {\em KC-Join}. In particular, {\em EE-Join} use the labeling to quickly decide the ancestor-descendant relationship between elements.}

%\medskip
%\textbf{regular path queries in graphs}
%Regular expressions as a language for querying graphs have
%been studied in the database literature for more than a decade,
%sometimes under the name of regular path queries or general path queries \cite{DBLP:journals/jodl/AbiteboulQMWW97,DBLP:conf/sigmod/BunemanDHS96, ...} 

%\cite{DBLP:conf/webdb/GubichevN11}

{\bf Contributions.} In contrast to previous work, we  answer
regular path queries over graphs using labeling, by leveraging the
fact that the graphs represent executions generated from a given
workflow specification. Specifically:

\vspace{0.75mm} \noindent $\bullet$ We identify a core  property,
{\em safe query}, that is defined for a query relative to a workflow specification,
and that enables the use of reachability labels for processing regular
path queries.   We show that safety of a query can be detected in polynomial
time in the size of the query and specification.

\vspace{0.75mm} \noindent $\bullet$ {\em Pairwise safe queries.}  We
show how to rewrite a specification using a safe query $R$, and use
the rewritten specification together with the reachability labels of two input
nodes $u$ and $v$ to answer whether there exists a path between $u$
and $v$ which conforms to $R$ , $u\stackrel{R}\leadsto v$, in constant time.

\vspace{0.75mm} \noindent $\bullet$ {\em All-pairs safe queries.} We
extend the pairwise query technique to answer all-pairs safe queries,
which ask whether $u\stackrel{R}\leadsto v$ for node pairs $(u,v)
\in U \times V$, and give an algorithm for answering all-pairs
queries that runs in time linear in $|U|,|V|,N$ and polynomial in the size of the
specification,  where $N$ is the number of reachable
nodes in $U\times V$. As a side effect, we answer
all-pairs reachability queries in  linear time in the input and output size,
which is optimal.

\eat{We show that the brute-force (nested loop) approach can  be
significantly improved using a prefiltering step which tests
reachability. We also give }

\vspace{0.75mm} \noindent $\bullet$ {\em All-pairs general queries.}
Finally, we present our approach for answering general regular path queries.  We give a top-down
algorithm for decomposing a general query into a small set of safe
subqueries, and show how to compose results of the safe subqueries, possibly with some small unsafe remainder,
to answer the original query.

\vspace{0.75mm} \noindent $\bullet$ Experimental studies demonstrate the significant speedup that is achieved by our approach.
% mentioned in the Related Work and the significant speedup that is achieved.

\vspace{2mm}
{\bf Outline.}  %The rest of the paper is organized as follows: 
Section~\ref{sec:prelims} presents the workflow model and  reachability labeling of~\cite{DBLP:conf/sigmod/BaoDM12}.  
We formally define regular path queries, and discuss  pairwise safe queries in Section~\ref{sec:pairwise}.  In particular, we show
how to  transform a regular path query $R$ to a reachability query by rewriting the workflow specification and decoding the labels of nodes using the rewritten workflow; 
we also discuss conditions under which this can be done (safe query).
Section~\ref{sec:allpairs}  shows how to answer all-pairs
safe queries, and discusses how to decompose a general query into a small
set of safe subqueries to answer general all-pairs queries.
Experimental results are given in Section~\ref{sec:experiments}.  
%We conclude in Section~\ref{sec:conclu}.

%\input{related-work}
\section{Prior Work}
%\section{Preliminaries}
\label{sec:prelims}

In this section, we summarize the workflow model and labeling scheme of~\cite{DBLP:conf/sigmod/BaoDM12}.
Although the labeling scheme was designed to answer reachability queries, we will extend it in Section~\ref{sec:pairwise} to answer pairwise regular path queries.  %Readers familiar with the results in ~\cite{DBLP:conf/sigmod/BaoDM12} can skip Section~\ref{ and go directly to Section~\ref{sec:pairwise}.}

\subsection{Workflow model~\protect\cite{DBLP:conf/sigmod/BaoDM12}}
\label{sec:model}

A  workflow {\em specification} is modeled as a {\em context-free graph grammar} (CFGG), which describes the design of the workflow 
and whose  language  corresponds to the set of all possible executions ({\em  runs}).    The model that we use is similar to \cite{DBLP:conf/sigmod/BaoDM11, DBLP:conf/vldb/BeeriEKM06}.  
Nonterminals in a CFGG $G$ correspond to {\em composite modules} and terminals to {\em atomic modules}; edges in graphs in $G$ correspond to dataflow between modules.  
%Both nodes and edges are labeled.  
More formally, we start by defining simple workflows and build up to workflows using productions.

%\vspace{-1mm}
\begin{definition}
\label{def:simple-workflow}
{\bf (Simple Workflow)}
A {\em simple workflow} is $W = (V, E)$, where $V$ is a set of {\em modules} and $E$ is a set of {\em data edges} between modules. 
Each node $v$ has a name drawn from a finite set of 
symbols $\Sigma$, denoted $name(v)$.  Each edge $e$ is tagged with an element of a finite set of 
symbols, $\Gamma$, which represents the name of the data flowing over the edge, denoted $\tau_E(e)$.  
There may be multiple parallel edges between two nodes, each with a different tag.  
\end{definition}
%\vspace{-1mm}

Simple workflows are reused as  composite modules to build more complex workflows. This is modeled using {\em workflow productions}.

%\vspace{-2mm}
\begin{definition}
\label{def:production}
{\bf (Workflow Production)}
A {\em workflow production} is of form $\production{M}{W}$, where $M$ is a composite module and $W$ is a simple workflow. 
\end{definition}
%\vspace{-3mm}

%\vspace{-2mm}
\begin{definition}
\label{def:grammar}
{\bf (Workflow Specification)}
A  {\em workflow specification} is a CFGG $G = (\Sigma, \Delta, S,
P)$, where $\Sigma$ is a finite set of modules, $\Delta \subseteq
\Sigma$ is a set of {\em composite modules} (then $\Sigma \setminus
\Delta$ is the set of {\em atomic modules}), $S \in \Sigma$ is a {\em
  start module}, and $P = \{ \production{M}{W} \ | \ M \in \Delta, W
\in  \Sigma^* \}$ is a finite set of {\em workflow productions}
(i.e. $W$ is a simple workflow whose nodes are modules in $\Sigma$).\eat{The {\em language} of $G$ is $L(G) = \{ R \in (\Sigma \setminus
  \Delta)^* \ | \ \derives{S}{R} \}$, representing the set of all {\em
    executions} ({\em runs}) of $G$.} We will frequently  refer to workflow specifications as workflows.
\end{definition}
%\vspace{-3mm}

%\vspace{-2mm}
\begin{definition}
{\bf (Workflow Derivation and Execution)} A given workflow {\em execution}
is {\em derived} by a series of {\em node replacements} or {\em
  derivation steps} corresponding to the productions in the
specification. We start with a graph $g_0$ consisting of a node named $S$. At the $i^{th}$ step of the derivation, a new
graph $g_{i}$ is obtained by replacing (executing) some composite node
$v$ of the current graph $g_{i-1}$ with a simple workflow $W$, where
$p: name(v)\rightarrow W$ is a production of the grammar. 
If $u$ is a node in $W$, then we say that $v$ {\em derives} $u$ ($u$
{\em is derived by} $v$) and extend this transitively. We denote by
$(v,p)$ a node replacement. The {\em
  language} of  a workflow is the set of all executions.
\end{definition}
%\vspace{-2mm}

\begin{figure}[t]
%\begin{minipage}{0.5\linewidth}
%\begin{center}
\centering
\subfloat[Workflow specification $G$]{\label{fig:grammar}\includegraphics[scale=0.2]{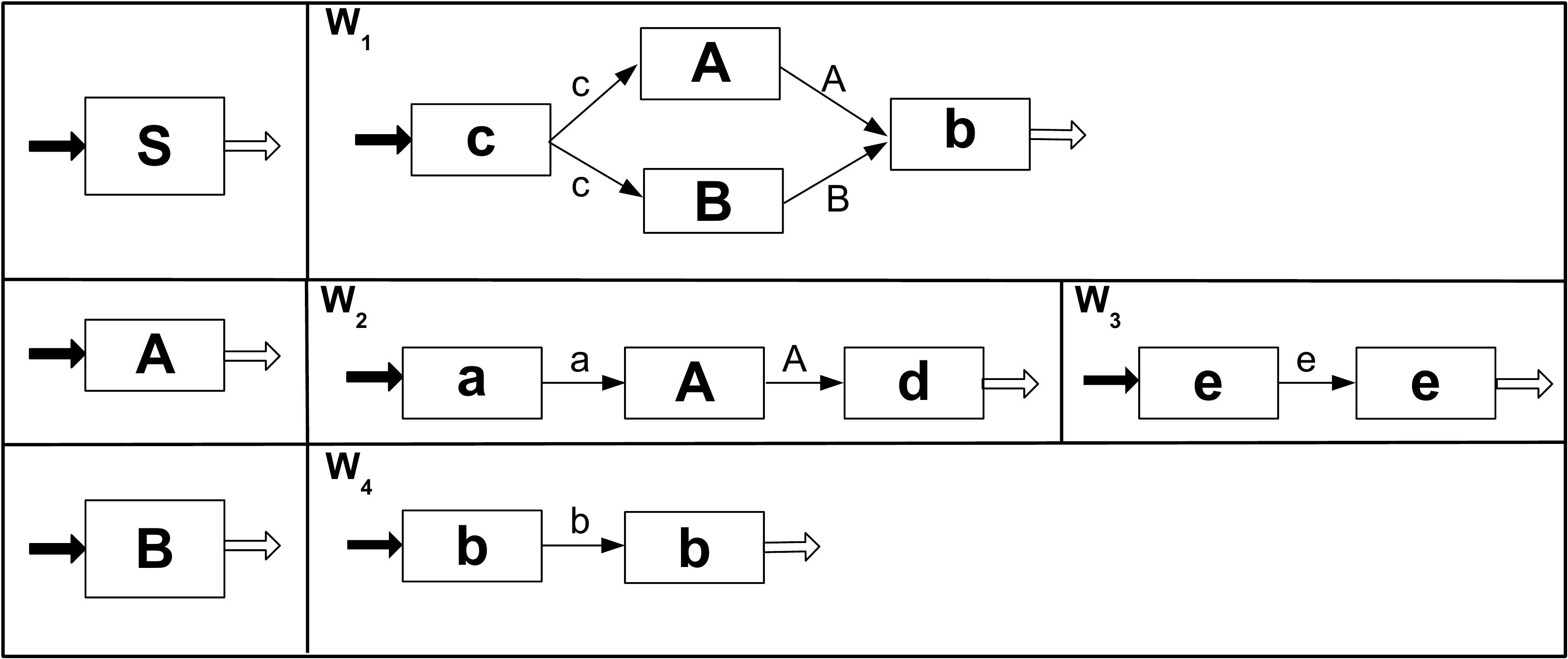}}
\\
\subfloat[A run of $G$]{\label{fig:run}\includegraphics[scale=0.2]{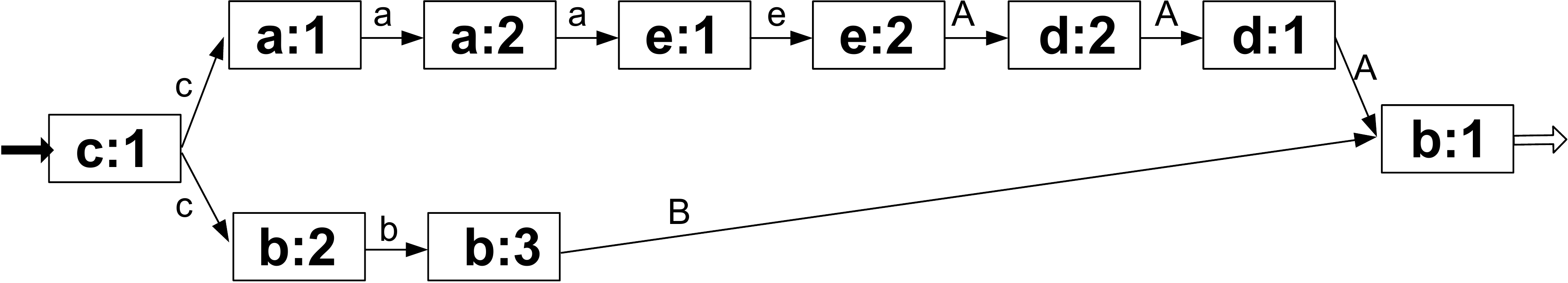}}
\\
\subfloat[Partial derivation graph]{\label{fig:derv}\includegraphics[scale=0.25]{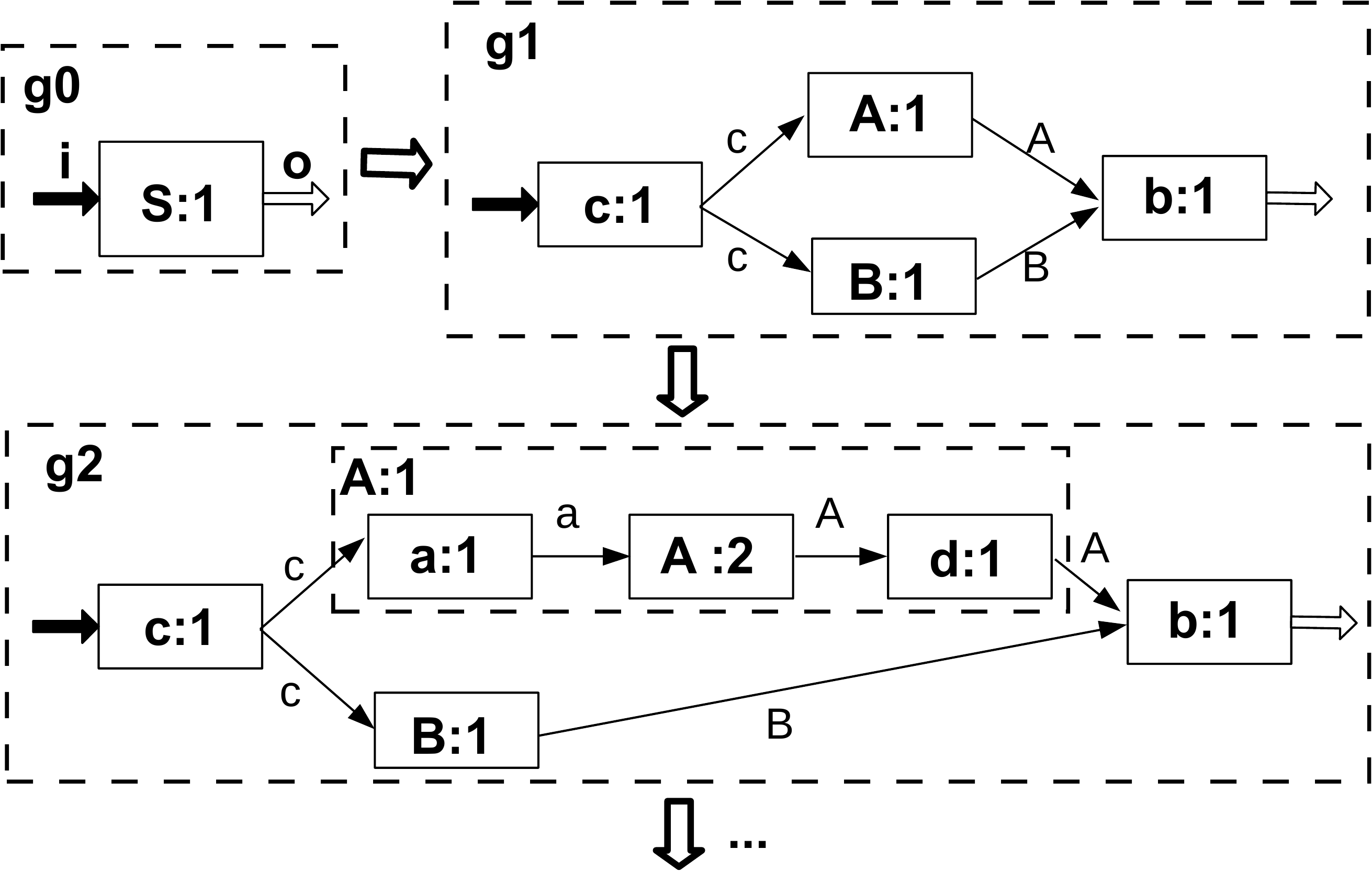}}
%\subfloat[Production graph $\productionGraph{G}$]{\label{fig:prod}\includegraphics[scale=0.6]{docs/ProdGraph.pdf}}
\eat{\subfloat[Production graph $\productionGraph{G}$]{\label{fig:prod}\includegraphics[scale=0.3] {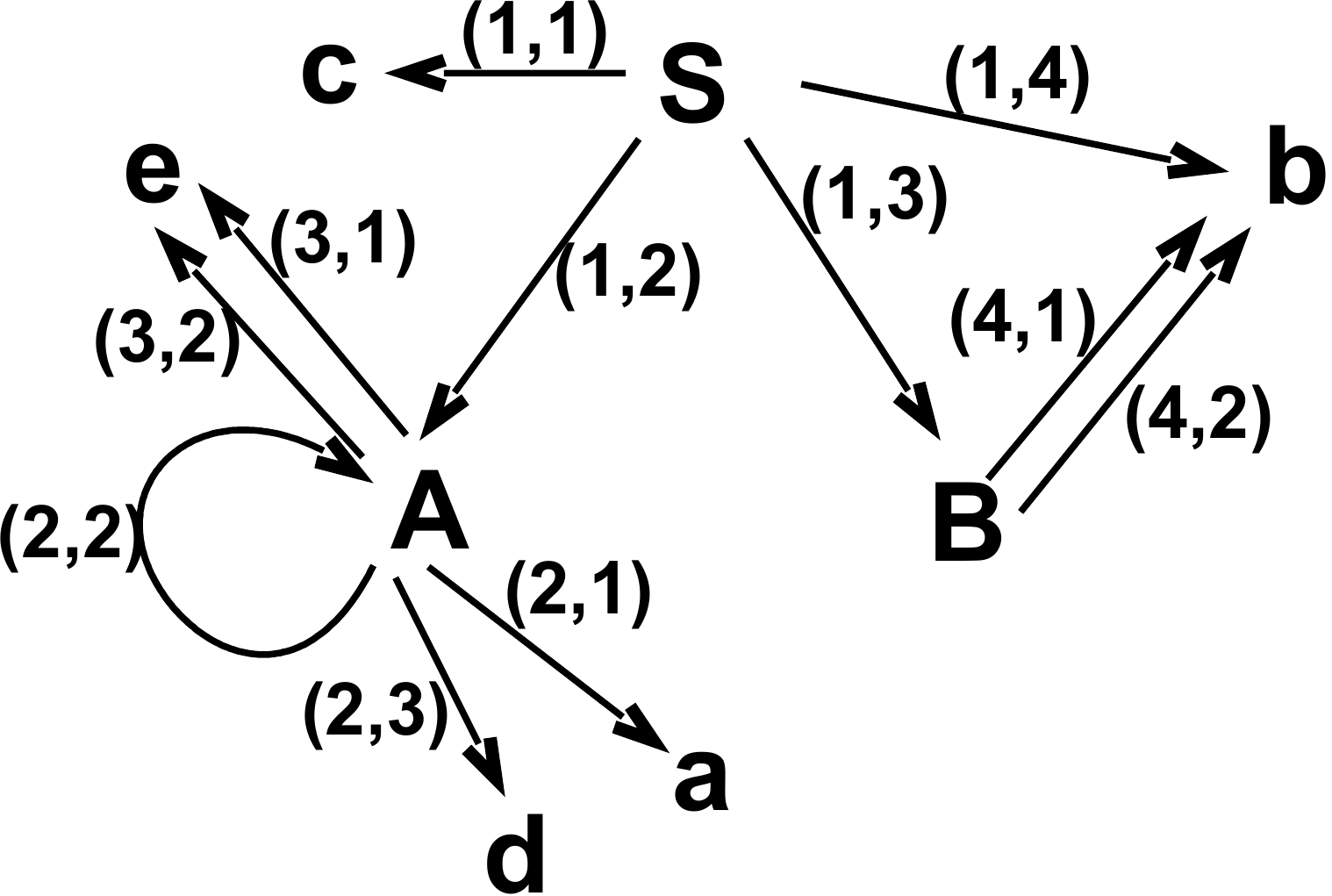}}
\\}
%\end{center}
%\end{minipage}
\vspace{-2mm}
\caption{Sample Workflow}
\vspace{-0.7cm}
\end{figure}
%\vspace{-0.2mm}

To simplify, in the examples of specifications throughout this paper the tags on edges are the same as the name of the modules at their head. 

%\vspace{-2mm}
\begin{example}
An example of a workflow specification is shown in Fig.~\ref{fig:grammar}, and one of its runs in Fig.~\ref{fig:run}.   Upper case module names ($S$, $A$, $B$) correspond to composite modules, and lowercase to atomic modules ($a$, $b$, $c$, $d$, $e$).  
The specification contains a choice of implementations for $A$, i.e. either $W_2$ or $W_3$.
The run inherits node (module) names and edge tags from its specification;  to disambiguate multiple occurrences of the same module  an occurrence number is appended to module names to  form a unique node id.   A partial sequence of derivation steps that would arrive at the run\eat{ in Fig.~\ref{fig:run}} is shown in Fig.~\ref{fig:derv}.
We start with a graph $g_0$ consisting of a single node $S:1$.%, which would be immediately labeled.  
In the first step, we replace $S:1$  with $W_2$; %, and immediately label $c:1$, $A:1$, $B:1$ and $b:1$.  
$S:1$ therefore derives $c:1$, $A:1$, $B:1$ and $b:1$ since they are nodes in $W_2$.
\eat{The end-node of the incoming edge to $S:1$ (tagged $i$)  would become  $c:1$ and be fixed as such since $c$ is atomic.  The start-node of the outgoing edge from $S:1$ (tagged $o$) would become  $b:1$ and again would be fixed since $b$ is atomic.}
Similarly, in the second step, we replace $A:1$ by $W_2$; %and immediately label $a:1$, $A:2$ and $d:1$.
$A:1$ therefore derives  $a:1$, $A:2$ and $d:1$.  By transitivity, $S:1$  also derives  $a:1$, $A:2$ and $d:1$.
\eat{In the second step, we could replace $A:1$ by $W_2$. The end-node of the incoming edge to $A:1$ (tagged $c$) would become $a:1$, and the start-node of the outgoing edge from $A:1$
(tagged $A$) would become $d:1$.}
\end{example}
%\vspace{-2mm}

In a {\em fine-grained} workflow model,  each module $M$ may have
multiple input/output {\em ports}, each representing a different data
item\eat{ with edges connecting an output port of one module to an input
port of another module}.  \eat{Using this model, e}Explicit dependencies
between the output of an atomic module and a subset of its inputs can
be captured as {\em internal} module edges.  
For example, if an atomic $m$ has two inputs, $x$ and $y$, and produces as output $x$ and $x+y$, $m$ would have two input ports and two output ports (see Fig.~\ref{fig:fine-grained}).  
The first output port, representing $x$, would be connected only to
the first input port, representing $x$, while the second output port, representing
$x+y$, would be connected to both input ports. The dependency
between input and output ports of a composite module may vary according
to its execution (more details can be found in Section~\ref{sec:labeling}). 
Executions of fine-grained workflows are also fine-grained.

\begin{figure}
\centering
\includegraphics[scale=0.3]{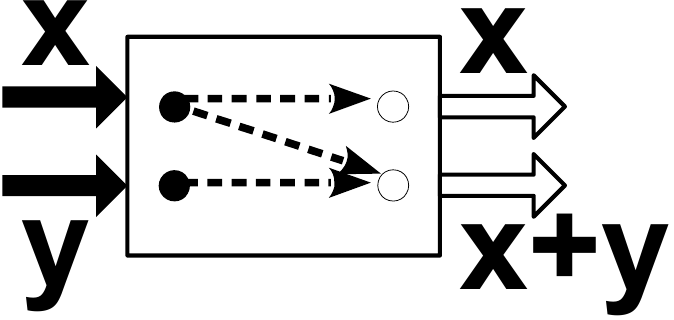}
\vspace{-2mm}
\caption{Fine-grained workflow}
\vspace{-0.7cm}
\label{fig:fine-grained}
\end{figure}

\eat{Moved to next section
To simplify the presentation in this paper, we consider
regular path queries on  {\em  coarse-grained} workflows in which each module has a single-input and
single-output, and the output is assumed to depend on the input; we
also assume  that simple workflows are {\em acyclic}. \eat{ Neither of
  these conditions limit the expressive power of our model; in
  particular loops can be captured by recursive productions. The
labeling technique proposed in this paper can be naturally extended to
the fine-grained model.} We introduce fine-grained workflows here since we will reduce regular path queries on
coarse-grained workflows to reachability queries on fine-grained
workflows (Section~\ref{sec:pairwise}).

The queries we consider are over runs (a.k.a. {\em provenance graphs} \cite{DBLP:journals/debu/DavidsonBELMBAF07}). 
Queries are {\em regular expressions} over edge tags, defined using concatenation, alternation and Kleene star:
$$e~:=~c~|~e_1e_2~|~e_1+e_2~|~e_1^*~|~e_1^+$$ 
where $c~:=~\epsilon~|~\_~|~a$ is a constant regular expression
($\epsilon$ is the empty string and $~\_~$ is the wildcard symbol that
matches any single symbol);
$e_1e_2$ denotes the concatenation of two sub-expressions; $e_1+e_2$ denotes alternation; and $e_1^*$ ($e_1^+$) denotes the 
set of all strings that can be obtained by concatenating zero (one) or
more strings chosen from $e_1$. Given a regular expression $R$, we
denote by $L(R)$ the set of strings that conform to $R$.

%\vspace{-1mm}
\begin{definition}
\label{def:regular-query}
{\bf (Regular Path Query)}
Let $G$ be a workflow specification and $g \in L(G)$ be a run.  Given a path 
$p = v_0 \xrightarrow{e_1} v_1 \xrightarrow{e_2} v_2 \ldots v_{n-1} \xrightarrow{e_n} v_n$ in $g$, we define $\tau_P(p)$ to be the concatenation of all edge tags on this
path, that is, $\tau_P(p) = \tau_E(e_1) \tau_E(e_2) \ldots \tau_E(e_n) \in  \Gamma^*$. A {\em regular path query} $R$ over $g$ is a {\em regular expression}  over
$\Gamma$. The result of $R$ on $g$  is defined as the set of node pairs $(u, v)$ in $g$ such that there is a path $p$ in $g$ from $u$ to
$v$  where $\tau_P(p) \in L(R)$.
\end{definition}
%\vspace{-2mm}

In this paper, we study two related sub-problems of answering regular path queries over workflow runs, {\em pairwise} queries and {\em all-pairs} queries.
%\vspace{-1mm}
\begin{definition}{\bf (Pairwise  Query)} 
Given two nodes $u,v$ from an edge-tagged graph $g$,  a {\em pairwise query} $R$ asks if there exists a path $p$ from $u$ to $v$ in $g$
such that $\tau_P(p) \in L(R)$, denoted by $u\stackrel{R}\leadsto v$.  
 \end{definition}
 %\vspace{-1mm}
 
The answer to a pairwise query is either true or false; reachability is a special case ($R= \; \_^*$).

%\vspace{-1mm}
 \begin{definition}{\bf (All-Pairs  Query)} 
Given two lists of nodes $l_1,l_2$ from an edge-tagged graph $g$,  {\em all-pairs query} $R$ asks for all node pairs $(u,v)\in
l_1\times l_2$ such that $u\stackrel{R}\leadsto v$. 
 \end{definition}
%\vspace{-2mm}

%\vspace{-1mm}
\begin{example}
Let $R_1= \;A^+$ and $R_2= \; A$.  
Revisiting the run in Fig.~\ref{fig:run}, the pairwise query result of $R_1$ for $(d:2, b:1)$ is true, but is false for $R_2$.  
The all-pairs query result of $R_1$ for $l_1= \{d:1, d:2, e:2\}$, $l_2=\{b:1, b:2\}$ is $\{(d:1,b:1), (d:2,b:1),(e:2,b:1)\}$.
The all-pairs  query result of $R_2$ for $l_1$, $l_2$ is $\{(d:1,b:1)\}$.
%\end{example}
}
%%%%%%%%%
%\subsection{\hspace{-3mm} Reachability labeling \cite{DBLP:conf/sigmod/BaoDMFi}}
\subsection{Reachability labeling \protect\cite{DBLP:conf/sigmod/BaoDM12}}
\label{sec:labeling}

The  labeling scheme in \cite{DBLP:conf/sigmod/BaoDM12},  called {\em dynamic, derivation based labeling}, was designed to answer reachability queries over views of workflows.
A {\em reachability query} is one which, given two nodes $u$, $v$ in a run $g$, returns ``yes" iff there is a path from $u$ to $v$ in $g$ (written $u\leadsto v$).
The  labeling scheme is based on the fine-grained  workflow model; however, it labels a run as if the workflow were coarse-grained, encoding only the sequence of productions used to arrive at each node (hence the name {\em derivation-based}).  
Reachability queries over views are then answered by decoding the labels 
using the fine-grained workflow specification intersected with the view definition.  
In a similar way, to answer regular path queries we label a run as if it were coarse-grained; however,
to decode labels we will use the query intersected workflow specification $G_R$,  which is fine-grained.
%In this section, we summarize the approach in \cite{DBLP:conf/sigmod/BaoDM12}.
Readers familiar with the results in ~\cite{DBLP:conf/sigmod/BaoDM12} can go directly to Section~\ref{sec:pairwise}.

{\bf Constraints.}
A labeling scheme is optimal (or compact) if 1) labels are
logarithmic in the size of the run, and 2)  labels can be decoded %(i.e. a query can be answered) 
in constant time, assuming that any operation on two
words ($log\;n$ bits) can be done in constant time.
For compact  reachability labeling to be achievable for fine-grained workflows, two corresponding constraints must be met:  1) the workflow must be {\em strictly-linear recursive}; and
2) the workflow must be {\em safe}. The first condition is
essential for logarithmic-size labeling and the second for
efficient decoding. %We now explain these two constraints.

To define the first constraint, we use the notion of  a {\em production graph}.

%\vspace{-2mm}
\begin{definition}
\label{def:production-graph}
{\bf (Production Graph\eat{~\cite{DBLP:conf/sigmod/BaoDM12}})}
Given a workflow  $G = (\Sigma, \Delta, S, P)$, the {\em production graph} of $G$ is a directed multigraph $\productionGraph{G}$ in which each vertex denotes a unique module in $\Sigma$. For each production $\production{M}{W}$ in $P$ and each module $M'$ in $W$, there is an edge from $M$ to $M'$ in $\productionGraph{G}$. Note that if $W$ has multiple instances of a module $M'$, then $\productionGraph{G}$ has multiple parallel edges from $M$ to $M'$.
\end{definition}
%\vspace{-2mm}

%\vspace{-1mm}
\begin{definition}
\label{def:strictly-linear-grammar}
{\bf (Strictly Linear-Recursive Workflow\eat{~\cite{DBLP:conf/sigmod/BaoDM12}})} 
A workflow $G$ is {\em recursive} if $\productionGraph{G}$ is cyclic, and a module in $G$ is  recursive if it belongs to a cycle in $\productionGraph{G}$.
$G$ is  {\em strictly linear-recursive} iff all cycles in $\productionGraph{G}$ are vertex-disjoint.
\end{definition}
%\vspace{-1mm}

%\vspace{-2mm}
\begin{example}
The production graph for the grammar $G$  in Fig.~\ref{fig:grammar} is shown in Fig.~\ref{fig:prod} (ignore for now the pair of numbers on edges).  
$G$ is recursive since there is a cycle in $\productionGraph{G}$; it is strictly linear-recursive since $\productionGraph{G}$ contains only one cycle. 
The only recursive node in $G$ is $A$.  

 In contrast, the hypothetical (and unlabeled) production graph shown in 
Fig.~\ref{fig:nonSLR} contains two cycles which share a node, $S$, and therefore the workflow that it represents is not  strictly linear-recursive. 
\end{example}
%\vspace{-2mm}

\begin{figure}
\centering
\begin{minipage}{0.6\linewidth}
\includegraphics[scale=0.3]{production-crop.pdf}
\caption{Production graph $\productionGraph{G}$}
\vspace{-7mm}
\label{fig:prod}
\end{minipage}
\begin{minipage}{0.3\linewidth}
\includegraphics[scale=0.3]{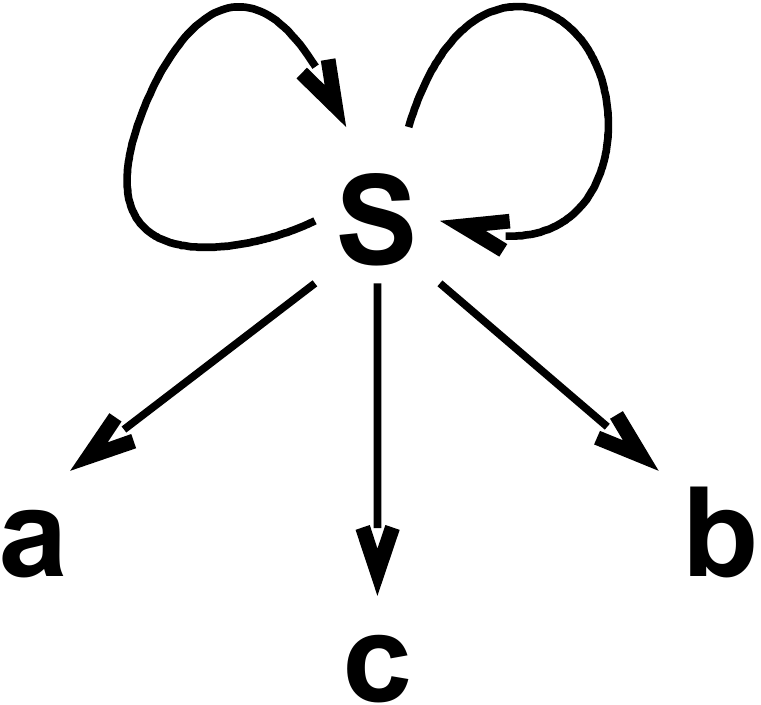}
\caption{Synthetic production graph}
\vspace{-7mm}
\label{fig:nonSLR}
\end{minipage}
\end{figure}

As argued in \cite{DBLP:conf/sigmod/BaoDM12}, strict linear recursion is able to capture
common recursive patterns found in repositories of scientific workflows, 
\eat{e.g. myExperiment~\cite{DBLP:journals/fgcs/RoureGS09},}
in particular looping and forked executions. 

Now we turn to the second constraint. Recall that a fine-grained
workflow is such that each  atomic module has  one or more
input/output ports whose dependency is explicitly specified by module
internal edges (see Fig.~\ref{fig:fine-grained}). The dependency between
the input and output ports of a composite module is determined by 
the executions of the module.  Intuitively, if a workflow is safe, we can draw unambiguous internal edges for all composite modules. %Details can be found in \cite{DBLP:conf/sigmod/BaoDM12}.

%\vspace{-2mm}
\begin{definition}
\label{def:safe-workflow}
{\bf (Safe Workflow\eat{~\cite{DBLP:conf/sigmod/BaoDM12}})} 
A workflow $G$ is {\em safe} iff for each
composite module, the dependency between its input and output ports
is deterministic w.r.t. all its executions. 
\end{definition}
%\vspace{-2mm}

%\vspace{-2mm}
\begin{example}
Consider the fine-grained workflow below 
and two of its executions
$ex_1$ and $ex_2$. % (Fig.~\ref{fig:unsafeSpec}). 
In $ex_1$, the second output port of $S$ solely
depends on the second input port of $S$. However, in $ex_2$, the
second output port of $S$ depends on both input ports of $S$. Thus the
dependency between the input and output ports of $S$ is not
deterministic. Therefore the workflow is not safe. An example of safe
workflow is given in Fig.~\ref{fig:specQ}, where the dependency for
composite modules are illustrated by internal module edges.
\begin{figure}[h]
\centering
\vspace{-2mm}
\includegraphics[scale=0.25]{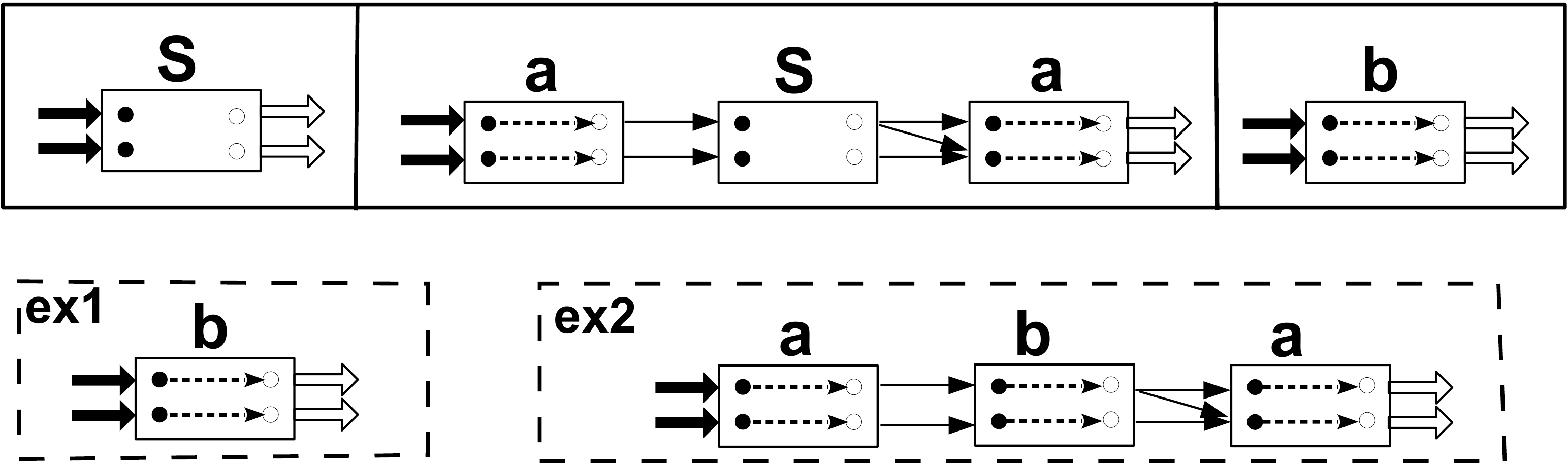}
%\vspace{-4mm}
\caption{\small Unsafe workflow}
\vspace{-4mm}
\label{fig:unsafeSpec}
\end{figure}
\end{example}

{\bf Labeling $\psi_V(u)$.}
The labeling function $\psi_V$ assigns a label to each node $u$ when
the node is derived and will not change the label as the workflow is
executed. The approach is based on a tree representation for a run, called the {\em compressed parse tree}. In contrast to the traditional parse tree used for context-free grammars whose depth may be proportional to the size of the run, the depth of a compressed parse tree is bounded by the size of the specification.   The compressed parse is constructed in a top-down manner, i.e. as productions are fired. A label 
is assigned to each node (module execution) as soon as it is executed, and encodes the sequence of derivation steps that create the module. %Details can be found in \cite{DBLP:conf/sigmod/BaoDM12}.

\begin{figure}
\centering
\vspace{-2mm}
\includegraphics[scale=0.3]{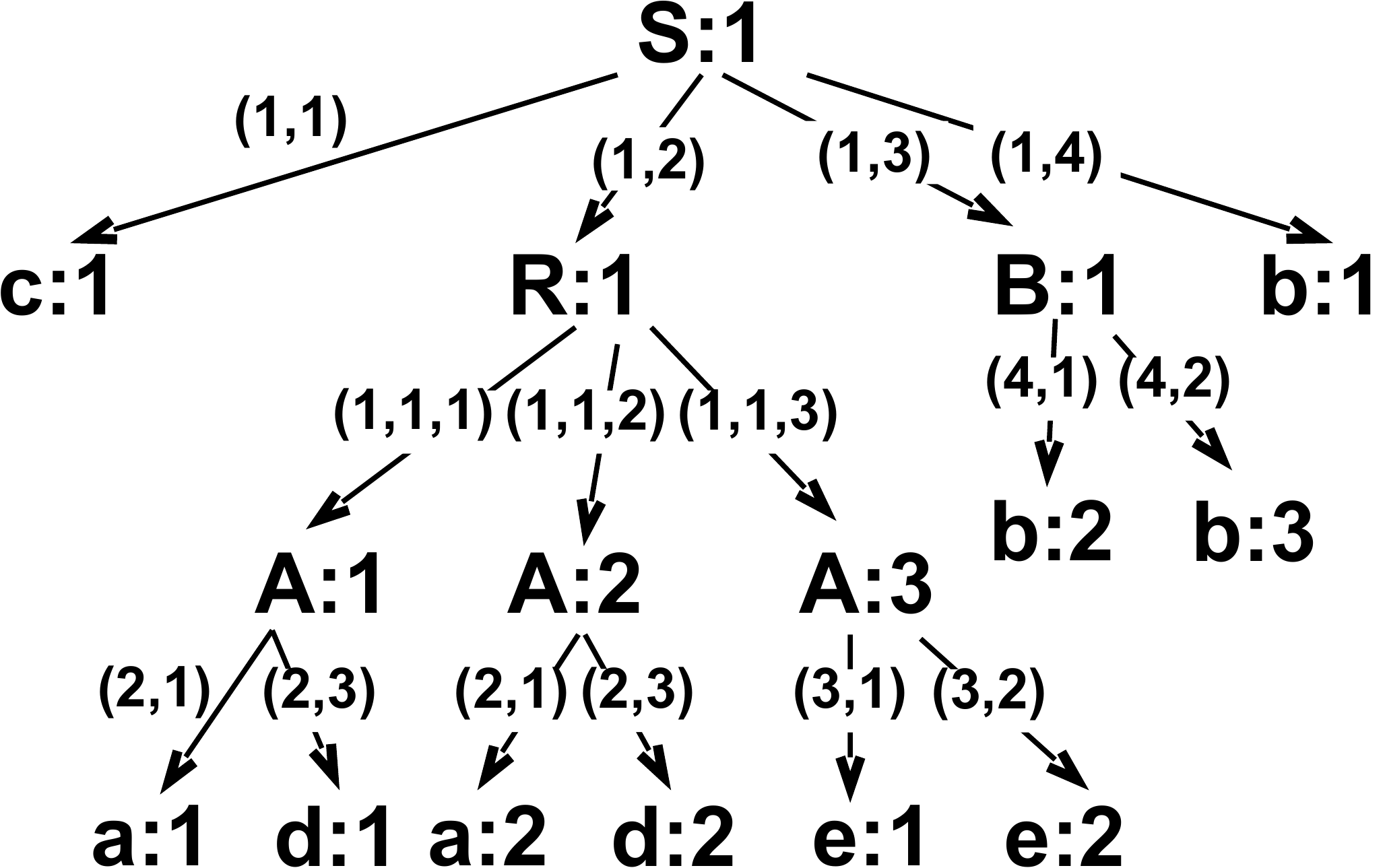}
%\vspace{-4mm}
\caption{\small Compressed Parse Tree}
\vspace{-7mm}
\label{fig:parse}
\end{figure}

%\vspace{-1mm}
\begin{example}
\label{eg:compressed-parse-tree}
The compressed parse tree %$\compressedParseTree{R}$
for the run in Fig.~\ref{fig:run} is shown in Fig.~\ref{fig:parse} (ignore the tree edge labels for now).  
Each leaf node denotes an atomic module, and each non-leaf node denotes either a composite module or a linear recursion, called a {\em recursive node} and labeled $R$. The children of a composite node denote the modules of the simple workflow produced by the production used in its execution; and the children of a recursive node denote a sequence of nested composite modules obtained by unfolding a cycle in the production graph.  
As before, occurrence numbers are used to disambiguate module executions.  In the sample run, 
$A$ is executed three times (denoted $A:1$, $A:2$, $A:3$), twice using $W_2$ and the final time using $W_3$.
\end{example}
%\vspace{-1mm}

Edges in the compressed parse tree $T$ are labeled as follows (we denote by $\psi_T(e)$ the label of an edge $e$):
We begin with assigning labels to edges of the production graph of the
specification, $\productionGraph{G}$
\eat{(recall Section~\ref{sec:cond1})}. First of all, fix an arbitrary ordering among the productions in $P$, and for each production $\production{M}{W}$, fix an arbitrary topological ordering among the modules in $W$. 
Let $p_k = \production{M}{W}$ be the $k$th production in $P$, and
$M_i$ be the $i$th module in $W$, then we assign the edge from $M$ to
$M_i$ in $\productionGraph{G}$ a pair $(k, i)$.  
In addition, fix an arbitrary ordering among all the cycles in
$\productionGraph{G}$, and for each cycle, fix an arbitrary edge as
the first edge of the cycle. We are now ready to label $T$. 
Let $e=(u,v)$ be an edge of $T$.
%We denote by $\cycle{s}$ the $s$th cycle in $\productionGraph{G}$ containing a list of number pairs.
(1) If $u$ is a composite node, then $e$ can be mapped to an edge $e'$ in $\productionGraph{G}$.  
Let $e' = (k, i)$, then $\psi_T(e) = (k, i)$; and (2) otherwise (if $u$ is a recursive node), let $u$ denote the $s$th cycle in $\productionGraph{G}$ starting from the $t$th edge. 
%This can be determined by the first child of $u$.
Let $v$ be the $i$th child of $u$, then $\psi_T(e) = (s, t, i)$.

\eat{
%REPLACED BY THE FOLLOWING EXAMPLE
\vspace{-1mm}
\begin{example}
Consider $\productionGraph{G}$ shown in
Fig.~\ref{fig:prod}. The productions are ordered as shown in the
Fig.~\ref{fig:grammar}. For example, the edge between $S$ and $c$ is
labeled (1,1) since $W_1$ is the first production, and $c$ is chosen
as the first module in its body. Since there is only one cycle in
$\productionGraph{G}$ it is the first cycle, and its first (and only)
edge is (2,2). {\color{red}Thus we label edge $(R:1,A:2)$ in $T$ as (1,1,2),
meaning that $A:2$ is the second child of $R:1$ which corresponds to the first
cycle in $\productionGraph{G}$ starting from the first edge.}
\end{example}
\vspace{-2mm}
}
%\vspace{-1mm}

\begin{example}
Consider $\productionGraph{G}$ shown in
Fig.~\ref{fig:prod}. The productions are ordered as shown in the
Fig.~\ref{fig:grammar}. For example, the edge between $S$ and $c$ is
labeled (1,1) since $W_1$ is the first production, and $c$ is chosen
as the first module in its body. The edge labels for Fig.~\ref{fig:parse} were constructed as follows:  
Labels on edges from the root $S$ of $T$ were taken from the production graph in Fig.~\ref{fig:prod}.  Since there is only one cycle in
$\productionGraph{G}$ it is the first cycle, and its first (and only)
edge is (2,2). Thus the cycle is labeled (1,1). The children of the
recursive node $R$ were ordered by their order of execution. Thus we label edge $(R:1,A:2)$ in $T$ as (1,1,2),
meaning that $A:2$ is the second child of $R:1$ which corresponds to the first
cycle in $\productionGraph{G}$ starting from the first edge.
\end{example}
%\vspace{-2mm}

A module execution (node $v$ in $T$) is labeled using the
concatenation of edge labels from the root of $T$ to $v$, denoted by
$\psi_V(v)$.  For example,  $\psi_V(b:2)=(1,3)(4,1)$.

\eat{
\vspace{-1mm}
\begin{example}
\label{eg:label-run}
The edge labels for Fig.~\ref{fig:parse} were constructed as follows:  
Labels on edges from the root $S$ of $T$ were taken from the production graph in Fig.~\ref{fig:prod}.  The only cycle in $T$ is labeled (1,1), and the children of the recursive node $R$ were ordered by their order of execution.
Module executions (not illustrated) are then labeled by their root to node path.  For example,  $\psi_V(b:2)=\{(1,3),(4,1)\}$.\eat{ and $\psi_V(b:1)=\{(1,4)\}$.}
\end{example}}

{\bf Decoding $\pi(\psi_V(u),\psi_V(v),G)$.}
Given a pair of module executions $u$, $v$ in a run that was generated from the workflow $G$, the predicate $\pi(\psi_V(u),\psi_V(v),G)$ outputs whether $u$ is reachable to $v$
in the run. A constant-time algorithm to evaluate
$\pi(\psi_V(u),\psi_V(v),G)$  is presented
in~\cite{DBLP:conf/sigmod/BaoDM12}. The subtlety is that the workflow
$G$ is taken as a parameter. As a very simple example, consider node
$c:1$ and $b:1$ of the run in Fig.~\ref{fig:stack-join} which was
derived from the (safe) workflow in Fig.~\ref{fig:specQ}. Once we
know $c:1$ and $b:1$ are from the same node replacement i.e. $(S:1,
S\rightarrow W_1')$ (which is determined by identifying their least common ancestor in the compressed parse tree (Fig.~\ref{fig:parse}) using their labels), we know directly from $W_1'$ the connectivity between $c:1$ and $b:1$. This is done in constant time because we access the specification rather than the run. Details of the decoding algorithm are omitted here, since they are not necessary for understanding the new techniques that will be proposed in this work.

\eat{
{\bf Decoding $\pi(\psi_V(u),\psi_V(v),G)$.}
Given a pair of module executions $u$, $v$ in a run that was generated
from the workflow $G$, the predicate
$\pi(\psi_V(u),\psi_V(v),G)$ outputs whether $u$ is reachable to $v$
in the run in constant time. The subtlety is that the workflow $G$ is taken as a
parameter. As a very simple example, consider node $c:1$ and $b:1$
of the run in Fig.~\ref{fig:stack-join} which was derived from the
(safe) workflow in Fig.~\ref{fig:specQ}. Once we know $c:1$ and
$b:1$ are from the same node replacement i.e. $(S:1, S \rightarrow W_1')$
(which is determined by identifying their least common ancestor in the
compressed parse tree (Fig.~\ref{fig:parse}) using their labels), 
we know directly from $W_1'$ the connectivity between $c:1$ and
$b:1$. This is done in constant time because we access the
specification rather than the run.

\eat{$\psi_V(u)$ and  $\psi_V(v)$ are then used together with the fine-grained grammar $G$ from which the run was generated to answer 
the reachability query between $u$, $v$. }

First, the most recent composite module execution $m$ in which $u$ and $v$ co-occur in the compressed parse tree is identified.  
This is done by taking the longest common prefix of $\psi_V(u)$ and  $\psi_V(v)$, which is the label of $u$ and $v$'s least common ancestor in $T$.  The production $p:M\rightarrow W$  which was used 
to implement $m$ can also be identified by the children of $m$ from
which $u$ and $v$ were derived, $m_u$, $m_v$. The connectivity between $u$ and $v$ has three parts: the connectivity
between  $m_u$ and $m_v$, between $u$ and $m_u$,   and between $m_v$
and $v$. Since the workflow is safe, we
can test these three parts in constant time.
Connectivity for the first part,  ($m_u$, $m_v$), has two cases: either $m$ is a composite module or it is a special recursive node.  If $m$ is a composite module, then $m_u$, $m_v$ are from the same simple workflow $W$, and connectivity can be easily tested.  If $m$ is a recursive node, then
since the grammar is strictly linear-recursive (and safe), the recursion being fired has a common connectivity pattern repeated many times. \eat{If we use a $|Q|\times |Q|$ boolean matrix to represent the connectivity of the common pattern (this can be calculated using $\lambda$), then a pattern being repeated only results in a limited number of different connectivities (at most  $|Q|^2$).} Thus by doing {\em fast matrix multiplication}, we can get the connectivity between $m_u$ and $m_v$ in constant time.   %Details can be found in \cite{DBLP:conf/sigmod/BaoDM12}.

\usetikzlibrary{decorations.pathmorphing}

\begin{figure}[h]
\centering
\begin{tikzpicture}
\node(tl) {};
\node(bl) [below of=tl, node distance=60]{};
\node(tr) [right of=tl, node distance=200]{};
\node(br) [below of=tr, node distance=60]{};
\draw[dashed] (tl) rectangle (br);
\node(M) [below of=tl, right of=tl, node distance=10] {$m$};

\node(tl2) [below of=M, node distance=10] {};
\node(bl2) [below of=tl2, node distance=30]{};
\node(tr2) [right of=tl2, node distance=80]{};
\node(br2) [below of=tr2, node distance=30]{};
\draw[dashed] (tl2) rectangle (br2);
\node(Mu) [below of=tl2, right of=tl2, node distance=7] {$m_u$};

\node(u)[below right=0.2 and 0.6 of tl2]{$u$};
\draw[dashed](u) +(-0.2,-0.2) rectangle  +(0.4,0.2);
\node(uo)[right of=u, node distance=7] {};
\draw(uo) circle (0.05);

\node(Muo) [right of=u, node distance=47]{};
\draw(Muo) circle (0.05);
%\draw[->,>=stealth,snake it] (uo)--(Muo);
\draw[-stealth,
decoration={snake, 
    amplitude = .4mm,
    segment length = 2mm,
    post length=0.9mm},decorate] (uo) -- (Muo) node (t1) [pos=0.5, below]{(1)};

\node(tl3) [right of=tl2, node distance=100] {};
\node(bl3) [below of=tl3, node distance=30]{};
\node(tr3) [right of=tl3, node distance=80]{};
\node(br3) [below of=tr3, node distance=30]{};
\draw[dashed] (tl3) rectangle (br3);
\node(Mv) [below of=tl3, right of=tl3, node distance=7] {$m_v$};

\node(v)[below right=0.2 and 1.6 of tl3]{$v$};
\draw[dashed](v) +(-0.4,-0.2) rectangle  +(0.2,0.2);
\node(vo)[left of=v, node distance=7] {};
\draw(vo) circle (0.05);

\node(Mvo) [left of=v, node distance=47]{};
\draw(Mvo) circle (0.05);
%\draw[->,>=stealth,snake it] (uo)--(Muo);
\draw[-stealth,
decoration={snake, 
    amplitude = .4mm,
    segment length = 2mm,
    post length=0.9mm},decorate] (Mvo) -- (vo) node (t3) [pos=0.5, below]{(3)};

\draw[-stealth,
decoration={snake, 
    amplitude = .4mm,
    segment length = 2mm,
    post length=0.9mm},decorate] (Muo) -- (Mvo) node (t2) [pos=0.5, below]{(2)};

\end{tikzpicture}
\vspace{-2mm}
\caption{Connectivity between $u$ and $v$}
\vspace{-2mm}
\label{fig:connectivity}
\end{figure}

For connectivity between $u$ and $m_u$ (or $m_v$ and $v$), we need to consider all the productions fired along the path from $m_u$ to $u$.  Again, since the workflow is safe, \eat{$\lambda(M)$ exists for all $M$, and therefore} for any production $M\rightarrow W$ we know the connectivity between a node in $W$ and the output of $M$. Transitively, we therefore know the connectivity between $u$ and $m_u$. 
}

\eat{
{\bf Labeling $\psi_V(u)$.}
The labeling function $\psi_V$ assigns a label to each node $u$ when
the node is derived and will not change the label as the workflow is
executed. The approach is based on a tree representation for a run, called the {\em compressed parse tree}. In contrast to the traditional parse tree used for context-free grammars whose depth may be proportional to the size of the run, the depth of a compressed parse tree is bounded by the size of the specification.   The compressed parse is constructed in a top-down manner, i.e. as productions are fired. A label 
is assigned to each node (module execution) as soon as it is executed, and encodes the sequence of derivation steps that create the module. %Details can be found in \cite{DBLP:conf/sigmod/BaoDM12}.

\begin{figure}[h]
\centering
\vspace{-2mm}
\includegraphics[scale=0.3]{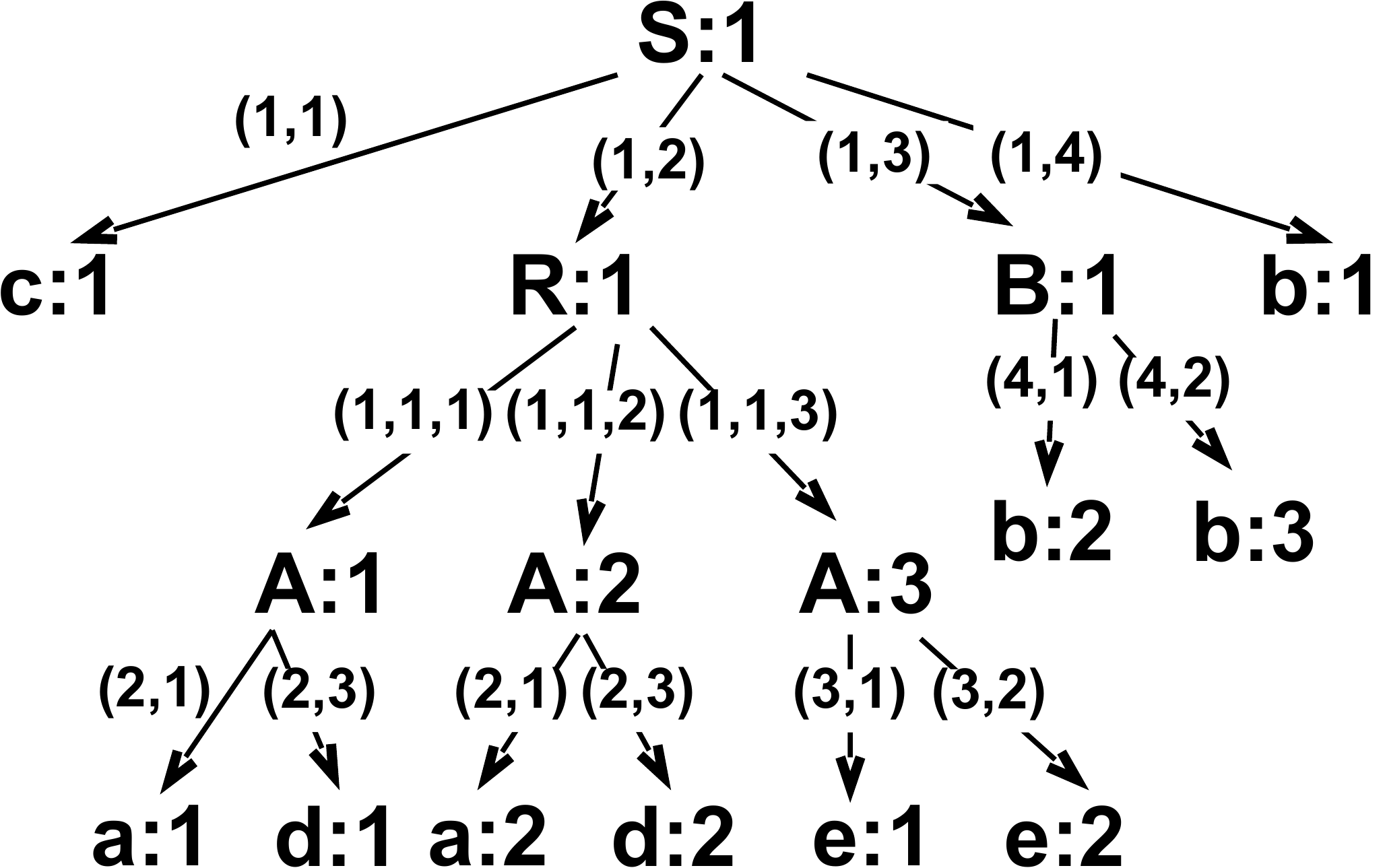}
%\vspace{-4mm}
\caption{\small Compressed Parse Tree}
\vspace{-4mm}
\label{fig:parse}
\end{figure}

\vspace{-1mm}
\begin{example}
\label{eg:compressed-parse-tree}
The compressed parse tree %$\compressedParseTree{R}$
for the run in Figure \ref{fig:run} is shown in Fig.~\ref{fig:parse} (ignore the tree edge labels for now).  
Each leaf node denotes an atomic module, and each non-leaf node denotes either a composite module or a linear recursion, called a {\em recursive node} and labeled $R$. The children of a composite node denote the modules of the simple workflow produced by the production used in its execution; and the children of a recursive node denote a sequence of nested composite modules obtained by unfolding a cycle in the production graph.  
As before, occurrence numbers are used to disambiguate module executions.  In the sample run, 
$A$ is executed three times (denoted $A:1$, $A:2$, $A:3$), twice using $W_2$ and the final time using $W_3$.
\end{example}
\vspace{-1mm}

Edges in the compressed parse tree $T$ are labeled as follows (we denote by $\psi_T(e)$ the label of an edge $e$):
We begin with assigning labels to edges of the production graph of the
specification, $\productionGraph{G}$
\eat{(recall Section~\ref{sec:cond1})}. First of all, fix an arbitrary ordering among the productions in $P$, and for each production $\production{M}{W}$, fix an arbitrary topological ordering among the modules in $W$. 
Let $p_k = \production{M}{W}$ be the $k$th production in $P$, and
$M_i$ be the $i$th module in $W$, then we assign the edge from $M$ to
$M_i$ in $\productionGraph{G}$ a pair $(k, i)$.  
In addition, fix an arbitrary ordering among all the cycles in
$\productionGraph{G}$, and for each cycle, fix an arbitrary edge as
the first edge of the cycle. We are now ready to label $T$. 
Let $e=(u,v)$ be an edge of $T$.
%We denote by $\cycle{s}$ the $s$th cycle in $\productionGraph{G}$ containing a list of number pairs.
(1) If $u$ is a composite node, then $e$ can be mapped to an edge $e'$ in $\productionGraph{G}$.  
Let $e' = (k, i)$, then $\psi_T(e) = (k, i)$; and (2) otherwise (if $u$ is a recursive node), let $u$ denote the $s$th cycle in $\productionGraph{G}$ starting from the $t$th edge. 
%This can be determined by the first child of $u$.
Let $v$ be the $i$th child of $u$, then $\psi_T(e) = (s, t, i)$.
\eat{
\vspace{-2mm}
\begin{example}
Consider again the production graph $\productionGraph{G}$ shown in
Fig.~\ref{fig:prod} and the compressed
parse tree $T$ in Fig.~\ref{fig:parse}. The edges of
$\productionGraph{G}$ are labeled. \eat{The productions in $\productionGraph{G}$ are ordered as shown in the
Fig.~\ref{fig:grammar}.} For example, the edge between $S$ and $c$ is
labeled (1,1) since $W_1$ is the first production, and $c$ is chosen
as the first module in its body. The edge $(S:1, c:1)$ in $T$
corresponds to the edge $(S,C)$ in $\productionGraph{G}$ and thus is
labeled as (1,1).  Since there is only one cycle in
$\productionGraph{G}$ it is the first cycle, and its first (and only)
edge is (2,2). Thus we label edge $(R:1,A:2)$ in $T$ as (1,1,2),
meaning that $A:2$ is the second child of $R:1$ which corresponds to the first
cycle in $\productionGraph{G}$ starting from the first edge.
\end{example}
\vspace{-2mm}}

\vspace{-1mm}
\begin{example}
Consider $\productionGraph{G}$ shown in
Fig.~\ref{fig:prod}. The productions are ordered as shown in the
Fig.~\ref{fig:grammar}. For example, the edge between $S$ and $c$ is
labeled (1,1) since $W_1$ is the first production, and $c$ is chosen
as the first module in its body. The edge labels for Fig.~\ref{fig:parse} were constructed as follows:  
Labels on edges from the root $S$ of $T$ were taken from the production graph in Fig.~\ref{fig:prod}.  Since there is only one cycle in
$\productionGraph{G}$ it is the first cycle, and its first (and only)
edge is (2,2). Thus the cycle is labeled (1,1). The children of the
recursive node $R$ were ordered by their order of execution. Thus we label edge $(R:1,A:2)$ in $T$ as (1,1,2),
meaning that $A:2$ is the second child of $R:1$ which corresponds to the first
cycle in $\productionGraph{G}$ starting from the first edge.
\end{example}
\vspace{-2mm}

A module execution (node $v$ in $T$) is labeled using the
concatenation of edge labels from the root of $T$ to $v$, denoted by
$\psi_V(v)$.  For example, $\psi_V(b:2)=\{(1,3),(4,1)\}$.\eat{ and $\psi_V(b:1)=\{(1,4)\}$.}

\eat{\vspace{-1mm}
\begin{example}
\label{eg:label-run}

Module executions (not illustrated) are then labeled by their root to node path.  For example,  $\psi_V(b:2)=\{(1,3),(4,1)\}$.\eat{ and $\psi_V(b:1)=\{(1,4)\}$.}
\end{example}}

{\bf Decoding $\pi(\psi_V(u),\psi_V(v),G)$.}
Given a pair of module executions $u$, $v$ in a run that was generated
from the workflow $G$, the predicate
$\pi(\psi_V(u),\psi_V(v),G)$ outputs whether $u$ is reachable to $v$
in the run in constant time. The subtlety is that the workflow $G$ is taken as a
parameter. As a very simple example, consider node $c:1$ and $b:1$
of the run in Fig.~\ref{fig:stack-join} which was derived from the
(safe) workflow in Fig.~\ref{fig:specQ}. Once we know $c:1$ and
$b:1$ are from the same node replacement, i.e. $S:1\rightarrow W_1'$
(which is determined by identifying their least common ancestor in the
compressed parse tree (Fig.~\ref{fig:parse}) using their labels), 
we know directly from $W_1'$ the connectivity between $c:1$ and
$b:1$. This is done in constant time because we access the
specification rather than the run.

\eat{$\psi_V(u)$ and  $\psi_V(v)$ are then used together with the fine-grained grammar $G$ from which the run was generated to answer 
the reachability query between $u$, $v$. }

First, the most recent composite module execution $m$ in which $u$ and $v$ co-occur in the compressed parse tree is identified.  
This is done by taking the longest common prefix of $\psi_V(u)$ and  $\psi_V(v)$, which is the label of $u$ and $v$'s least common ancestor in $T$.  The production $p:M\rightarrow W$  which was used 
to implement $m$ can also be identified by the children of $m$ from
which $u$ and $v$ were derived, $m_u$, $m_v$. The connectivity between $u$ and $v$ has three parts: the connectivity
between  $m_u$ and $m_v$, between $u$ and $m_u$,   and between $m_v$
and $v$. Since the workflow is safe, we
can test these three parts in constant time.
Connectivity for the first part,  ($m_u$, $m_v$), has two cases: either $m$ is a composite module or it is a special recursive node.  If $m$ is a composite module, then $m_u$, $m_v$ are from the same simple workflow $W$, and connectivity can be easily tested.  If $m$ is a recursive node, then
since the grammar is strictly linear-recursive (and safe), the recursion being fired has a common connectivity pattern repeated many times. \eat{If we use a $|Q|\times |Q|$ boolean matrix to represent the connectivity of the common pattern (this can be calculated using $\lambda$), then a pattern being repeated only results in a limited number of different connectivities (at most  $|Q|^2$).} Thus by doing {\em fast matrix multiplication}, we can get the connectivity between $m_u$ and $m_v$ in constant time.   %Details can be found in \cite{DBLP:conf/sigmod/BaoDM12}.

\eat{
\begin{figure}[h]
\centering
\includegraphics[scale=0.3]{decoding-crop.pdf}
\vspace{-1mm}
\caption{Connectivity between $u$ and $v$}
\vspace{-3mm}
\label{fig:connectivity}
\end{figure}}

\usetikzlibrary{decorations.pathmorphing}

\begin{figure}[h]
\centering
\begin{tikzpicture}
\node(tl) {};
\node(bl) [below of=tl, node distance=60]{};
\node(tr) [right of=tl, node distance=200]{};
\node(br) [below of=tr, node distance=60]{};
\draw[dashed] (tl) rectangle (br);
\node(M) [below of=tl, right of=tl, node distance=10] {$m$};

\node(tl2) [below of=M, node distance=10] {};
\node(bl2) [below of=tl2, node distance=30]{};
\node(tr2) [right of=tl2, node distance=80]{};
\node(br2) [below of=tr2, node distance=30]{};
\draw[dashed] (tl2) rectangle (br2);
\node(Mu) [below of=tl2, right of=tl2, node distance=7] {$m_u$};

\node(u)[below right=0.2 and 0.6 of tl2]{$u$};
\draw[dashed](u) +(-0.2,-0.2) rectangle  +(0.4,0.2);
\node(uo)[right of=u, node distance=7] {};
\draw(uo) circle (0.05);

\node(Muo) [right of=u, node distance=47]{};
\draw(Muo) circle (0.05);
%\draw[->,>=stealth,snake it] (uo)--(Muo);
\draw[-stealth,
decoration={snake, 
    amplitude = .4mm,
    segment length = 2mm,
    post length=0.9mm},decorate] (uo) -- (Muo) node (t1) [pos=0.5, below]{(1)};

\node(tl3) [right of=tl2, node distance=100] {};
\node(bl3) [below of=tl3, node distance=30]{};
\node(tr3) [right of=tl3, node distance=80]{};
\node(br3) [below of=tr3, node distance=30]{};
\draw[dashed] (tl3) rectangle (br3);
\node(Mv) [below of=tl3, right of=tl3, node distance=7] {$m_v$};

\node(v)[below right=0.2 and 1.6 of tl3]{$v$};
\draw[dashed](v) +(-0.4,-0.2) rectangle  +(0.2,0.2);
\node(vo)[left of=v, node distance=7] {};
\draw(vo) circle (0.05);

\node(Mvo) [left of=v, node distance=47]{};
\draw(Mvo) circle (0.05);
%\draw[->,>=stealth,snake it] (uo)--(Muo);
\draw[-stealth,
decoration={snake, 
    amplitude = .4mm,
    segment length = 2mm,
    post length=0.9mm},decorate] (Mvo) -- (vo) node (t3) [pos=0.5, below]{(3)};

\draw[-stealth,
decoration={snake, 
    amplitude = .4mm,
    segment length = 2mm,
    post length=0.9mm},decorate] (Muo) -- (Mvo) node (t2) [pos=0.5, below]{(2)};

\end{tikzpicture}
\vspace{-2mm}
\caption{Connectivity between $u$ and $v$}
\vspace{-2mm}
\label{fig:connectivity}
\end{figure}

For connectivity between $u$ and $m_u$ (or $m_v$ and $v$), we need to consider all the productions fired along the path from $m_u$ to $u$.  Again, since the workflow is safe, \eat{$\lambda(M)$ exists for all $M$, and therefore} for any production $M\rightarrow W$ we know the connectivity between a node in $W$ and the output of $M$. Transitively, we therefore know the connectivity between $u$ and $m_u$. 

\eat{
%REMOVED by Xiaocheng
\vspace{-1mm}
\begin{example}
\label{eg:pairwiseQ1}
Continuing with the example, suppose we are given\eat{ the regular
  path query $R_3=\; (\_)^*e(\_)^*$ and} node pair  $(b:2,\;b:1)$ in Fig.~\ref{fig:stack-join}. The
connectivity between $b:2$ and $b:1$ is shown below.  
Since $\psi_V(b:2)=\{(1,3),(4,1)\}$ and $\psi_V(b:1)=\{(1,4)\}$, their longest common prefix is the empty string, and therefore their least common ancestor is the root of $T$ in Fig.~\ref{fig:parse}
($S:1$).  The remainder of 
$\psi_V(b:2)$ after the longest common prefix is removed is $\{(1,3),(4,1)\}$, and therefore the module from which $b:2$ was derived is the third module of the first production, module $B:1$.  
The remainder of  $\psi_V(b:1)$ after the longest common prefix is removed is $\{(1,4)\}$, and therefore the module from which $b:1$ was derived is $S:1$. 
 Furthermore, $b:2$ is the first module in the fourth production
($W'_4$).

\begin{figure}[h]
\centering
\begin{tikzpicture}
\node(tl) {};
\node(bl) [below of=tl, node distance=60]{};
\node(tr) [right of=tl, node distance=210]{};
\node(br) [below of=tr, node distance=60]{};
\draw[dashed] (tl) rectangle (br);
\node(m)[below of=tl, node distance=10]{};
\node(M)[right of=m, node distance=21] {$S:1~ (W_1')$};

\node(tl2) [below of=M, node distance=10] {};
\node(bl2) [below of=tl2, node distance=30]{};
\node(tr2) [right of=tl2, node distance=80]{};
\node(br2) [below of=tr2, node distance=30]{};
\draw[dashed] (tl2) rectangle (br2);
\node(mu)[below of = tl2, node distance = 6]{};
\node(Mu) [ right of=mu, node distance=25] {$B:1~(W_4')$};

\node(u)[below right=0.3 and 0.8 of tl2]{$b:2$};
\draw[dashed](u) +(-0.5,-0.2) rectangle  +(0.6,0.2);
\node(uo)[right of=u, node distance=11] {};
\draw(uo) circle (0.05);

\node(Muo) [right of=u, node distance=38]{};
\draw(Muo) circle (0.05);
%\draw[->,>=stealth,snake it] (uo)--(Muo);
\draw[-stealth,
decoration={snake, 
    amplitude = .4mm,
    segment length = 2mm,
    post length=0.9mm},decorate] (uo) -- (Muo) node (t1) [pos=0.5, below]{(1)};

\node(tl3) [right of=tl2, node distance=100] {};
\node(bl3) [below of=tl3, node distance=30]{};
\node(tr3) [right of=tl3, node distance=80]{};
\node(br3) [below of=tr3, node distance=30]{};
\draw[dashed] (tl3) rectangle (br3);
\node(mv)[below of=tl3, node distance=6]{};
\node(Mv) [ right of=mv, node distance=10] {$b:1$};

\eat{
\node(v)[below right=0.2 and 1.6 of tl3]{$b:1$};
\draw[dashed](v) +(-0.6,-0.2) rectangle  +(0.4,0.2);
\node(vo)[left of=v, node distance=11] {};
\draw(vo) circle (0.05);}

\node(Mvo) [right of=Muo, node distance=30]{};
\draw(Mvo) circle (0.05);
\eat{
%\draw[->,>=stealth,snake it] (uo)--(Muo);
\draw[-stealth,
decoration={snake, 
    amplitude = .4mm,
    segment length = 2mm,
    post length=0.9mm},decorate] (Mvo) -- (vo) node (t3) [pos=0.5, below]{(3)};}

\draw[-stealth,
decoration={snake, 
    amplitude = .4mm,
    segment length = 2mm,
    post length=0.9mm},decorate] (Muo) -- (Mvo) node (t2) [pos=0.5, below]{(2)};

\end{tikzpicture}
\vspace{-2mm}
\caption{Connectivity between $b:2$ and $b:1$}
\vspace{-2mm}
\label{fig:connectivity}
\end{figure}

The connectivity between $B:1$ and $b:2$ (marked as $(1)$) can be obtained by
looking at $W_4'$ where paths between $B:1$ and $b:2$ leave the states
unchanged. Similarly, the connectivity between $B:1$ and $b:1$ (marked as $(2)$)
can be obtained from $W_1'$. Combining $(1)$ and $(2)$ we know that the
$q_0$ output port of $b:1$ depends on solely the input port $q_0$ of
$b:2$ and $q_f$ of $b:1$ depends on solely $q_f$ of $b:1$.
\eat{Combining $(1)$ and $(2)$ we know that the
paths between $b:2$ and $b:1$ leave the states of the DFA
unchanged. Thus there is no path from  the $q_0$ port of $b:2$ (which is also the first occurrence of $b$'s port) 
%output port $q_f$ of $b$ to the input port $q_0$ of $B$, 
to the $q_f$ port of $b:1$, and the answer to the pairwise query is
``false".}  This can be confirmed by examining the sample run.
\end{example}
\vspace{-1mm}
}

\eat{
\vspace{-3mm}
\begin{example}
\label{eg:pairwiseQ2}
We are again given the query $R_3$ and node pair $(a:1, e:2)$, where $\psi_V(a:1)=\{(1,2),(1,1,1),(2,1)\}$ and $\psi_V(e:2)=\{(1,2),(1,1,3),(3,2)\}$.  
Their longest common prefix is $\{(1,2)\}$, which identifies $R:1$.  The module from which $e:2$ is derived within $R:1$ is $A$, and for $a:1$ it is also $A$.  We therefore test if there is a path between
the output port $q_f$ of $A$ to the input port $q_0$ in $A$ in $G_Q$.  Since there is, the answer to the pairwise query is  ``true".
This can again be confirmed by examining the sample run.
\end{example}
}

\eat{Note that the labeling of the run, $\psi_V$, is done at run time using
the compressed parse tree.  The fine-grained specification label,
$\lambda$, is calculated at query time as described in the previous
subsection.}
}
\section{Answering Pairwise Safe Queries}
\label{sec:pairwise}

In this section, we show how to answer pairwise safe queries, assuming that the execution has been labeled using the reachability labeling scheme of \cite{DBLP:conf/sigmod/BaoDM12}. 
To achieve this, we must do two things: 1) reduce regular path queries  to equivalent reachability queries;  and
%2) since the reachability labeling is done when query $R$ is unknown, the reduction must not change the labels of runs; 
2) identify constraints on $G$ and $R$ which allow reachability labeling to be used.

Informally, our approach works as follows. We reduce regular path
queries on a coarse-grained workflow $G$ to equivalent reachability queries on
a fine-grained (and query-specific) workflow $G_R$.  This workflow is obtained by
intersecting $G$ with a DFA of query $R$, thereby modeling DFA state transitions within modules of $G_R$ while leaving the sequence of productions unchanged (Section~\ref{sec:transform}).
Since reachability labeling only works for safe workflows, 
we discuss in Section~\ref{sec:constraints} a class of {\em safe
  queries} which guarantee that the query-intersected workflow $G_R$ is safe.
Since the run $g$ is labeled with information about the sequence of productions used as
it was executed, the (pre-existing) reachability labels of a pair of nodes $u$, $v$ can then be combined at query time with DFA state transition information in $G_R$ to answer 
$u\stackrel{R}\leadsto v$ (Section~\ref{sec:decodingReg}).

We start in Section~\ref{sec:queries} by formally defining the class of queries studied in this paper. 

\subsection{Regular path queries}
\label{sec:queries}

To simplify the presentation, in this paper we will consider
regular path queries on  {\em  coarse-grained} workflows in which each module has a single-input and
single-output, and the output is assumed to depend on the input; we
also assume  that simple workflows are {\em acyclic}.
\eat{
We will reduce regular path queries on coarse-grained workflows to reachability queries on fine-grained
workflows in Section~\ref{sec:transform}.
}

Queries are {\em regular expressions} over edge tags, defined using concatenation, alternation and Kleene star:
$$e~:=~c~|~e_1e_2~|~e_1+e_2~|~e_1^*~|~e_1^+$$ 
where $c~:=~\epsilon~|~\_~|~a$ is a constant regular expression
($\epsilon$ is the empty string and $~\_~$ is the wildcard symbol that
matches any single symbol);
$e_1e_2$ denotes the concatenation of two sub-expressions; $e_1+e_2$ denotes alternation; and $e_1^*$ ($e_1^+$) denotes the 
set of all strings that can be obtained by concatenating zero (one) or
more strings chosen from $e_1$. Given a regular expression $R$, we
denote by $L(R)$ the set of strings that conform to $R$.

%\vspace{-1mm}
\begin{definition}
\label{def:regular-query}
{\bf (Regular Path Query)}
Let $G$ be a workflow specification and $g \in L(G)$ be a run.  Given a path 
$p = v_0 \xrightarrow{e_1} v_1 \xrightarrow{e_2} v_2 \ldots v_{n-1} \xrightarrow{e_n} v_n$ in $g$, we define $\tau_P(p)$ to be the concatenation of all edge tags on this
path, that is, $\tau_P(p) = \tau_E(e_1) \tau_E(e_2) \ldots \tau_E(e_n) \in  \Gamma^*$. A {\em regular path query} $R$ over $g$ is a {\em regular expression}  over
$\Gamma$. The result of $R$ on $g$  is defined as the set of node pairs $(u, v)$ in $g$ such that there is a path $p$ in $g$ from $u$ to
$v$  where $\tau_P(p) \in L(R)$.
\end{definition}
%\vspace{-2mm}

In this paper, we study two related sub-problems of answering regular path queries over workflow runs, {\em pairwise} queries and {\em all-pairs} queries.
%\vspace{-1mm}
\begin{definition}{\bf (Pairwise  Query)} 
Given two nodes $u,v$ from an edge-tagged graph $g$,  a {\em pairwise query} $R$ asks if there exists a path $p$ from $u$ to $v$ in $g$
such that $\tau_P(p) \in L(R)$, denoted by $u\stackrel{R}\leadsto v$.  
 \end{definition}
 %\vspace{-1mm}
 
The answer to a pairwise query is either true or false; reachability is a special case ($R= \; \_^*$).

%\vspace{-1mm}
 \begin{definition}{\bf (All-Pairs  Query)} 
Given two lists of nodes $l_1,l_2$ from an edge-tagged graph $g$,  {\em all-pairs query} $R$ asks for all node pairs $(u,v)\in
l_1\times l_2$ such that $u\stackrel{R}\leadsto v$. 
 \end{definition}
%\vspace{-2mm}

%\vspace{-1mm}
\begin{example}
Let $R_1= \;A^+$ and $R_2= \; A$.  
Revisiting the run in Fig.~\ref{fig:run}, the pairwise query result of $R_1$ for $(d:2, b:1)$ is true, but is false for $R_2$.  
The all-pairs query result of $R_1$ for $l_1= \{d:1, d:2, e:2\}$, $l_2=\{b:1, b:2\}$ is $\{(d:1,b:1), (d:2,b:1),(e:2,b:1)\}$.
The all-pairs  query result of $R_2$ for $l_1$, $l_2$ is $\{(d:1,b:1)\}$.
\end{example}

In the next subsection, we reduce regular path queries on coarse-grained workflows to reachability queries on fine-grained
workflows.

\subsection{From regular path queries to reachability queries}
\label{sec:transform}

A simple %linear-time 
algorithm  for answering a pairwise regular path query $R$ over a run $g \in L(G)$ works as follows: 
augment each module in the run with input and output ports representing
the states of a DFA for $R$, and connect the output port of module execution $u$ 
representing state $q$ to the input port of module execution $v$ representing state $q'$ 
iff the tag of edge $e=(u,v)$ causes the DFA to transition from $q$ to
$q'$ ($\delta(q, \tau_E(e))= q'$).  Atomic modules leave states
unchanged. Then for any two nodes $u,v$ in
$g$, $u\stackrel{R}\leadsto v$ iff the input port of $u$ representing
the start state of the DFA reaches an output port of
$v$ representing an accepting state of the DFA. This algorithm is
linear in the run size since it needs to scan the run to perform the
intersection. 

%(The proof is similar to Lemma~\ref{lemma:intersect}).

%\vspace{-1mm}
\begin{example}
The fine-grained run  in Fig.~\ref{fig:stack-join} corresponds to the sample run in Fig.~\ref{fig:run},
augmented with state transition information for query $R_3$ (Fig.~\ref{fig:Q1}).  Since there are two
states in the DFA for $R_3$ , $q_0$ and $q_f$,
each module execution has two input ports and two output ports.
%By ``corresponds", we mean that same sequence of node replacements (derivation graph) was used in the two runs.
Since there is an edge tagged $e$ between $e:1$ and $e:2$ in the sample run, the output
port $q_0$ of $e:1$ connects to the input port $q_f$ of $e:2$.  All other edge tags  leave the DFA in the  same state.

In the sample run, $R_3= (\_)^*e(\_)^*$ evaluates to true for $(c:1,b:1)$, but false for $(c:1,b:3)$.  Correspondingly, in the fine-grained run there is a path from the 
input port $q_0$ of $c:1$ to the output port $q_f$ of $b:1$, but 
there is no path from the input port $q_0$ of $c:1$ to the output port $q_f$ of $b:3$.
\end{example}

\begin{figure}
\centering
\includegraphics[scale=0.25]{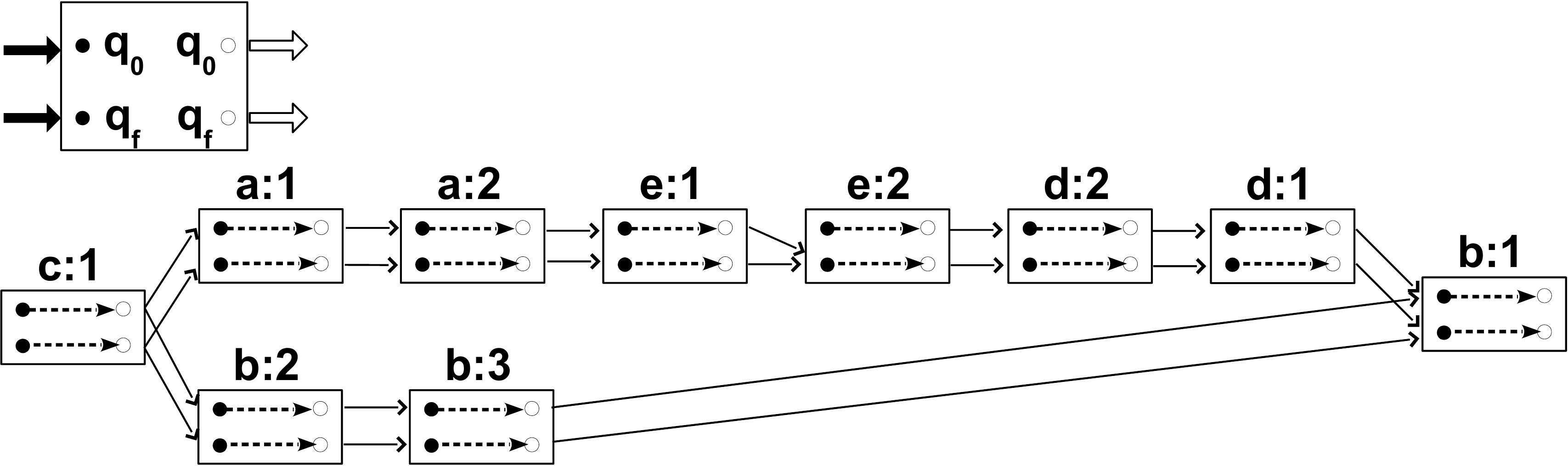}
\vspace{-2mm}
\caption{Fine-grained representation of sample run  for query $R_3$}
\vspace{-7mm}
\label{fig:stack-join}
\end{figure}

However, since the run is very large compared to the specification, we do not actually want to generate the query-augmented run.  Rather, we augment the workflow {\em specification}
with state-transition information from the DFA  for $R$, transforming $G$ into a query-specific, fine-grained workflow\eat{(described in Section~\ref{sec:prelims})}.
We then use the derivation information encoded as labels in the run to answer pairwise queries.  

We now describe how to augment the workflow specification with DFA state-transition information by
{\em intersecting} the workflow with the DFA.

Let ${\cal M}=(Q,\Gamma,\delta,q_0,F)$ be a DFA of query $R$. 
\eat{Intersecting the specification $G=(\Sigma,\Delta,S,P)$ with $\cal
  M$ constructs a new specification $G'=(\Sigma',\Delta',S',P')$ as
  follows:  }
We intersect the specification $G=(\Sigma,\Delta,S,P)$ with $\cal
  M$ to obtain a new specification $G_R=(\Sigma',\Delta',S',P')$ as
  follows:

\begin{enumerate}
\item For each module $M\in\Sigma$, create an {\em augmented} module $f_M(M)$ in $\Sigma'$, 
where $f_M(M)$ has $|Q|$  input ports $I_1\ldots I_{|Q|}$ and $|Q|$ output ports $O_1\ldots O_{|Q|}$, corresponding to the $|Q|$ states of $\cal M$.   
Module names are preserved, i.e. $name(f_M(M))=\; name(M)$.  For each
atomic module $M\in\Sigma/\Delta$, for each input port $q$ of $M$, there is
an edge from $q$ to the output port $q$ of $M$.
%We denote by $[M,q]_I$ ($[M,q] _O$) the input (output) port corresponding to state $q$ of module $M$.  

\item For each $M\in\Delta$, add $f_M(M)$ to $\Delta'$.  $S' = \; f_M(S)$.

\item For each production $p: M\rightarrow W\in P$  where $W=(V,E)$, construct a new production $f_P(p): f_M(M)\rightarrow W'\in P'$, where $W'=(V',E')$ is constructed from $W$
by (i) for each $v \in V$,  $f_M(v) \in V'$; and (ii) for each edge $e=(u,v)\in E$ there 
is an edge from output port $q$ of $f_M(u)$ to input port $q'$ of $f_M(v)$  iff $\delta(q,\tau_E(e))=q'$.
%$([f_M(u),q]_O,[f_M(v),q']_I)\in E'$ iff $\delta(q,\psi(e))=q'$. 
\end{enumerate}

\begin{example}
The intersection of the specification $G$ of Fig.~\ref{fig:grammar}
and the query $R_3=\_^*e\_^*$ in Fig.~\ref{fig:Q1}  is shown in
Fig.~\ref{fig:specQ} (ignore the edges inside composite modules for now). Each
module in $G_R$ has two input ports and two output ports corresponding
to $q_0$ and $q_f$.  
The only occurrence of  the edge tag $e$ in $G$ is on the edge $(e,e)$ in $W_2$. 
Since $\delta(q_0,e)=q_f$, $\delta(q_f,e)=q_f$, all output ports of the first $e$ in $W'_2$ are connected to the $q_f$ input port of the second $e$. 
All other augmented modules in $G_R$ connect $q_0$ to $q_0$ and $q_f$ to 
$q_f$. 
\end{example}
%\vspace{-3mm}

\begin{figure}[h]
\centering
\includegraphics[scale=0.2]{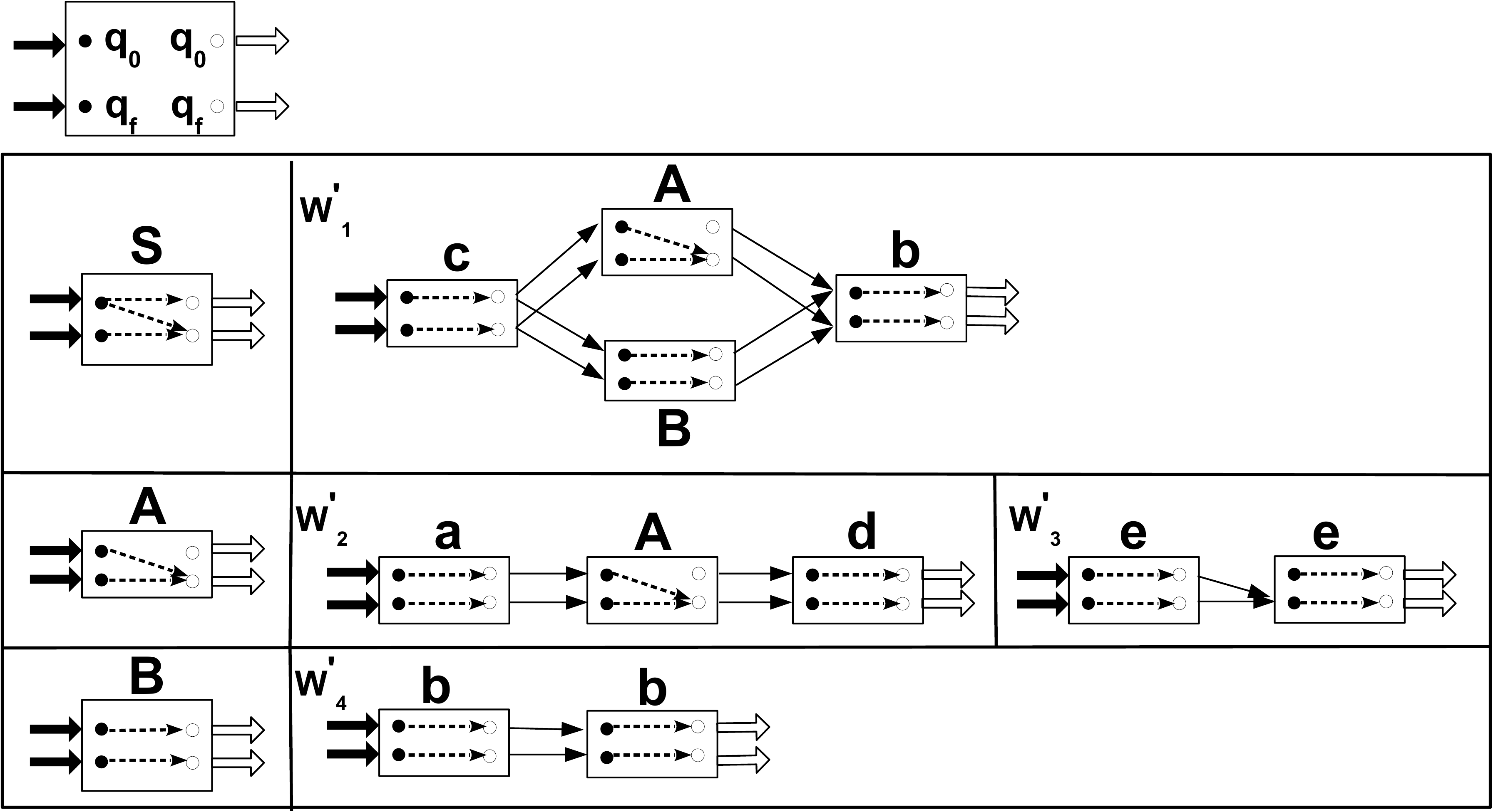}
\vspace{-1mm}
\caption{The fine-grained grammar $G_R$ obtained by intersecting $G$ with query $R_3$}
\vspace{-7mm}
\label{fig:specQ}
\end{figure}

%\scream{Note that till now, we can not draw internal edges for composite modules}

We now reduce a pairwise query $R$ over a run $g$
generated by  $G$ into an equivalent reachability query over a run
$g'$ generated by $G_R$. The correctness  is shown below. 

%\vspace{-2mm}
\begin{lemma}
Let $G$ be a workflow, $\cal M$ be a DFA for query $R$,
and $G_R$ be the intersection of $G$ with $\cal M$.
Given any two vertices $u,v$ of $g \in L(G)$,   $u\stackrel{R}\leadsto v$ iff 
the $q_0$ input port of $u$ reaches an accepting output port $q_f$ of $v$ in $g'$,
where $g' \in L(G_R)$ is the run obtained by the same sequence of node replacements
as $g$.
\label{lemma:intersect}
\end{lemma}

%\eat{%%%%%%%%%%%%%%%%%%%%%%eat by huangxc 7/23/2014
\begin{proof}
Let $G_R=(\Sigma',\Delta', S', P')$ be the intersection of $G= (\Sigma,\Delta, S, P)$ and \M,
and  $g$ be a run of $G$ derived by the sequence of node replacement
$(v_1,p_1)\ldots (v_n,p_n)$. Note that $f_M,f_P$ are
bijective functions. 
The corresponding run $g' \in G_R$ is  derived by the sequence of node replacements
$(v_1',f_P(p_1))\ldots (v_n',f_P(p_n))$.  Observe that since $g$ and $g'$ are derived by corresponding sequences of node replacements, the only difference between
$g$ and $g'$ is the addition of input/output ports and connections between modules corresponding
to edge tag induced state transitions in the DFA.  

If $u\stackrel{R}\leadsto v$ in $g$, then there is a path $p= e_1...e_k$
from $u$ to $v$ in $g$ with  $e_i= (u_{i-1}, u_i)$ ($u=u_0$ and $v=u_k$) such that
$\delta(q_0, \tau_E(e_1))=q_1$, $\delta(q_1, \tau_E(e_2))=q_2$, ...,
$\delta(q_{k-1}, \tau_E(e_k))=q_f$, and $q_f$ is a final state in
${\cal M}$.   Thus there are corresponding edges $e'_i$
in $g'$ from output port $q_{i-1}$ of $u'_{i-1}$ to input port $q_i$ of $u'_i$.  Since modules in 
$g'$ are atomic, the $i$'th input port of each module is connected to its $i$'th output port. Hence
there is a path $p'$ in $g'$ from input port $q_0$ in $u'$ to output port $q_f$ in $v'$. 
The only-if direction proceeds analogously.
 \end{proof}
%}

The benefit of this approach is that intersecting the workflow $G$ with a DFA $\cal M$ to obtain $G_R$ models state transitions in $\cal M$ but  does not change the
sequence of productions in $G$.   Thus the reachability labels of \cite{DBLP:conf/sigmod/BaoDM12},
which rely solely on the production sequence that arrive at a node, can be used.
However, in order to use the results of \cite{DBLP:conf/sigmod/BaoDM12} $G_R$ must be 1)  strictly-linear recursive, a condition which is guaranteed by $G$ being strictly-linear recursive; and 2) safe, a condition which is  guaranteed by $R$ being a {\em safe query} for $G$.  
\eat{In the next subsection w}We next discuss the general problem of using (dynamic) labels for answering regular path queries, and then define safe queries.

\subsection{Safe Queries}
\label{sec:constraints}
%%%%%%%%%%%%%%%%%%%%%%%%%%%%%%%%%%%%%%%%%%%%%%%All this is deleted!
\eat{In this section, we define safe queries corresponding to safe
workflows  in \cite{DBLP:conf/sigmod/BaoDM12}.}

\eat{Similar to \cite{DBLP:conf/sigmod/BaoDM12}, for compact and efficient  labeling to be achievable for regular path queries, two constraints must be met:  1) the workflow $G$ must be strictly-linear recursive; and
2) the query $R$ must be {\em safe} with respect to $G$.
The first constraint arises (even when we consider coarse-grained
workflow model) since the (more general) problem of regular path queries over coarse-grained
workflows can be reduced to one of reachability in fine-grained workflows, as discussed in Section~\ref{sec:pairwise}, whereas
the second constraint is due to the particular added challenge introduced by regular path queries.}

\eat{
\subsubsection{Strictly-linear recursive workflows}  
\label{sec:cond1}
To define the first constraint, we use the notion of a {\em production graph}.

\vspace{-2mm}
\begin{definition}
\label{def:production-graph}
{\bf (Production Graph~\cite{DBLP:conf/sigmod/BaoDM12})}
Given a workflow  $G = (\Sigma, \Delta, S, P)$, the {\em production graph} of $G$ is a directed multigraph $\productionGraph{G}$ in which each vertex denotes a unique module in $\Sigma$. For each production $\production{M}{W}$ in $P$ and each module $M'$ in $W$, there is an edge from $M$ to $M'$ in $\productionGraph{G}$. Note that if $W$ has multiple instances of a module $M'$, then $\productionGraph{G}$ has multiple parallel edges from $M$ to $M'$.
\end{definition}
\vspace{-2mm}

\eat{
\vspace{-2mm}
\begin{example}
The production graph for the grammar  in Fig.~\ref{fig:grammar} is shown in Fig.~\ref{fig:prod} (ignore for now the pair of numbers on edges).  $G$ is recursive since there is a cycle in $\productionGraph{G}$.  The only recursive node is $A$.  
\end{example}
\vspace{-2mm}
}

\eat{MOVED
\begin{figure}[th!]
%\begin{minipage}{0.5\linewidth}
%\begin{center}
\centering
\subfloat[Workflow specification $G$]{\label{fig:grammar}\includegraphics[scale=0.2]{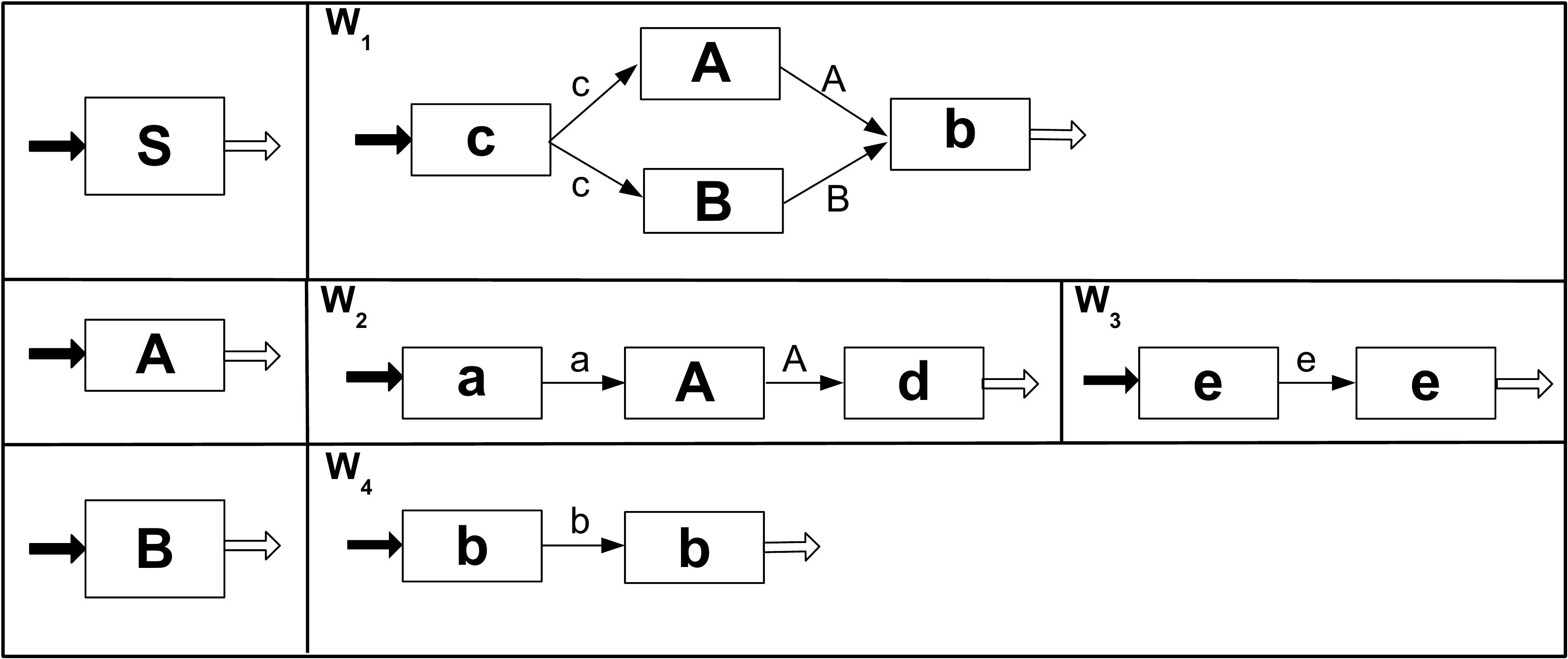}}
\\
\subfloat[A run of $G$]{\label{fig:run}\includegraphics[scale=0.2]{docs/run-crop.pdf}}
\\
\subfloat[Partial derivation graph]{\label{fig:derv}\includegraphics[scale=0.25]{docs/derivation-crop.pdf}}
\\
%\subfloat[Production graph $\productionGraph{G}$]{\label{fig:prod}\includegraphics[scale=0.6]{docs/ProdGraph.pdf}}
\subfloat[Production graph $\productionGraph{G}$]{\label{fig:prod}\includegraphics[scale=0.3] {docs/production-crop.pdf}}
\\
%\end{center}
%\end{minipage}
\caption{Sample Workflow}
\end{figure}
}

\vspace{-1mm}
\begin{definition}
\label{def:strictly-linear-grammar}
{\bf (Strictly Linear-Recursive Workflow~\cite{DBLP:conf/sigmod/BaoDM12})} 
A workflow $G$ is said to be {\em recursive} if $\productionGraph{G}$ is cyclic, and a module in $G$ is  recursive if it belongs to a cycle in $\productionGraph{G}$.
$G$ is  {\em strictly linear-recursive} iff all cycles in $\productionGraph{G}$ are vertex-disjoint.
\end{definition}
\vspace{-1mm}

\vspace{-2mm}
\begin{example}
The production graph for the grammar $G$  in Fig.~\ref{fig:grammar} is shown in Fig.~\ref{fig:prod} (ignore for now the pair of numbers on edges).  
$G$ is recursive since there is a cycle in $\productionGraph{G}$; it is strictly linear-recursive since $\productionGraph{G}$ contains only one cycle. 
The only recursive node in $G$ is $A$.  

 In contrast, the hypothetical (and unlabeled) production graph shown in 
Fig.~\ref{fig:nonSLR} contains two cycles which share a node, $S$, and therefore the workflow that it represents is not  strictly linear-recursive. 
\end{example}
\vspace{-2mm}

\eat{
\begin{figure}[h]
\centering
\includegraphics[scale=0.3]{docs/NonLRProdGraph-crop.pdf}
\caption{Another production graph}
\label{fig:nonSLR}
\end{figure}
\vspace{-2mm}}

%\subfloat[Production graph $\productionGraph{G}$]{\label{fig:prod}\includegraphics[scale=0.3] {docs/production-crop.pdf}}
%\vspace{-2mm}
\begin{figure}[ht]
\centering
\begin{minipage}{0.6\linewidth}
\includegraphics[scale=0.3]{docs/production-crop.pdf}
\caption{Production graph $\productionGraph{G}$}
\vspace{-2mm}
\label{fig:prod}
\end{minipage}
\begin{minipage}{0.3\linewidth}
\includegraphics[scale=0.3]{docs/NonLRProdGraph-crop.pdf}
\caption{Synthetic production graph}
\vspace{-2mm}
\label{fig:nonSLR}
\end{minipage}
\end{figure}

As argued in \cite{DBLP:conf/sigmod/BaoDM12}, strict linear recursion is able to capture
common recursive patterns found in repositories of scientific workflows, 
\eat{e.g. myExperiment~\cite{DBLP:journals/fgcs/RoureGS09},}
in particular looping and forked executions. }

\eat{\subsubsection{Safe queries}  
%%%%%%%%%%%%%%%%%%%%%%%%%%%%%%%%%%%%%%%%%%%%%%%%%%%%%%%%%%%%%%%%%%%%%%%%%%%%%%%%%To here.
\label{sec:safety}}
Due to the complexity of regular expressions, there are specifications for which dynamic labeling
is not possible, even if arbitrarily large labels
are allowed. For example, consider the grammar in Fig.~\ref{fig:grammar} and  query $(\_)^*a(\_)^*$. At the second step of the derivation of the run in Fig.~\ref{fig:run}, the graph
$g_1$ is $W_1$. Observe that we cannot tell if the query will be satisfied for ($c:1$, $b:1$), and therefore it is impossible to label them as they arise:   If $A\rightarrow
W_2$ is applied, the answer would be ``yes''.  However, if $A\rightarrow W_3$ is applied, the answer would be ``no''. We therefore say that query $(\_)^*a(\_)^*$
is not {\em safe} with respect to this workflow. In contrast, the query
$(\_)^*e(\_)^*$ is  safe, since $A$ must eventually terminate with an execution of $W_3$.
It is also easy to see that the reachability query $(\_)^*$ is safe with respect to any workflow, since every module in a coarse-grained model has a single input and a single output. \eat{For the fine-grained model with multiple inputs and outputs, a similar notion of safety is introduced in \cite{DBLP:conf/sigmod/BaoDM12} to ensure the feasibility of dynamic labeling.}

We formally define safety of a query with respect to a workflow in terms of a finite state automata (DFA) for the query.  Since there may be several equivalent DFAs for a query, we then show that it is sufficient to consider the minimal DFA for the query, and give an efficient algorithm for checking safety using the minimal DFA.

Recall the standard definition of a DFA \cite{hopcroft2007introduction}:

%\vspace{-1mm}
\begin{definition}
({\bf DFA})  A {\em DFA} ${\cal M}$ is a 5-tuple $(Q,\Gamma,\delta,q_0,F)$, where $Q$ is a set of states, $\Gamma$ is a set of input symbols (i.e., edge tags), $\delta:Q\times\Gamma\rightarrow Q$ is a transition function, $q_0$ is the start state and $F\subseteq Q$ is a set of accept states. We extend $\delta$ to $\delta^*:Q\times \Gamma^*\rightarrow Q$ such that  $\delta^*(q,w_1\ldots w_n)=\delta^*(\delta(q,w_1),w_2\ldots w_n)$ where $w_i\in \Gamma$.
%We denote by $L^{\cal M}_{q_1q_2}$ the set of strings that transit state $q_1\in Q$ to $q_2\in Q$.
\end{definition}
%\vspace{-2mm}

%\vspace{-2mm}
\begin{definition}
{\bf (Safe DFA)}
A DFA ${\cal M}=(Q,\Gamma,\delta,q_0,F)$ is {\em safe} with respect to  a workflow specification $G=(\Sigma,\Delta,S,P)$ iff any state pair $(q_1,q_2)\in Q\times Q$ is safe.  
A state pair $(q_1,q_2)$ is {\em safe} iff  $ \forall M\in \Sigma$ and any two executions $ex_1,ex_2$
%\in (\Sigma/\Delta)^*$ 
of $M$, if there exists a path $p$ connecting an input $i$ of $M$ and an output $o$ of $M$ in $ex_1$ such that $\delta^*(q_1,\tau_P(p))=q_2$, then there exists a path $p'$ connecting $i$ and $o$ in $ex_2$ such that $\delta^*(q_1,\tau_P(p'))=q_2$.
% (otherwise, we say $(q_1,q_2)$ is unsafe with respect to $(i,o)$ of $M$).
\end{definition}
%\vspace{-2mm}

%\vspace{-2mm}
\begin{definition}\label{def:safeQuery}
%{\bf (Strictly Safe Query)} A regular path query $R$ is said to be {\em strictly safe} 
{\bf (Safe Query)} A regular path query $R$ is said to be {\em safe} 
with respect to a workflow specification $G$ iff there exists a DFA that accepts
 $R$ and is safe with respect to $G$.
\end{definition}
%\vspace{-1mm}

%\vspace{-2mm}
\begin{example}\label{ex:safe_unsafeQ}
\eat{Returning to the  sample workflow $G$ in
  Fig.~\ref{fig:grammar}, consider the composite module $A$.}
Consider the composite module $A$ in the sample workflow $G$ in Fig.~\ref{fig:grammar}.
%, the only composite module in $G$ that contains alternation. 
Two of $A$'s executions are shown in Fig.~\ref{fig:A_ex}; all other executions of $A$ will represent $k$ recursions, with $k$ modules named $a$, followed by two named $e$, followed by $k$ named $d$.

\begin{figure}[h]
\centering
\includegraphics[scale=0.3]{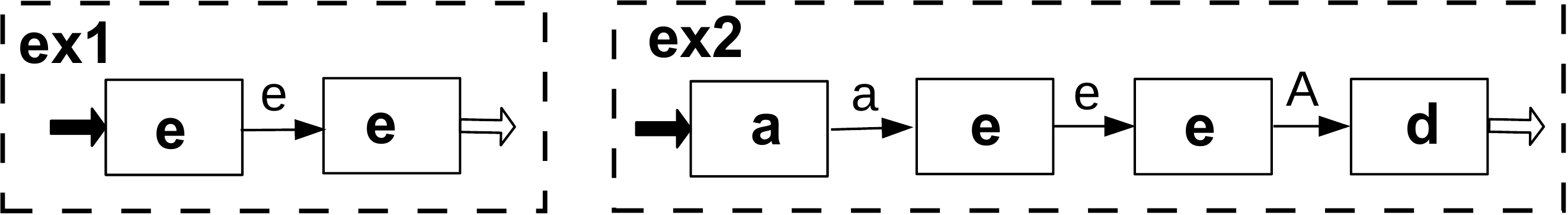}
\vspace{-1mm}
\caption{Two of $A$'s executions}
\vspace{-3mm}
\label{fig:A_ex}
\end{figure}

\noindent
\vspace{-0.6cm}
\begin{figure}[hb]
%\centering
\begin{minipage}[b]{0.45\linewidth}
\subfloat[$R_3: (\_)^*e(\_)^*$ (safe)]{\label{fig:Q1}
\begin{tikzpicture}[shorten >=1pt,node distance=2cm,on grid,auto] 
   \node[state,initial] (q_0)   {$q_0$};  
   \node[state,accepting](q_f) [ right of= q_0, node distance=1.5cm] {$q_f$};
    \path[->] 
    (q_0) edge  node {$e$} (q_f)
          edge  [loop above] node {$a,b,c,A,B$} ()
    (q_f) edge [loop above] node {$\_$} ();
\end{tikzpicture}
}
\end{minipage}
\begin{minipage}[b]{0.45\linewidth}
\subfloat[$R_4: e$ (not safe)]{
\label{fig:Q2}
\begin{tikzpicture}[shorten >=1pt,node distance=2cm,on grid,auto] 
   \node[state,initial] (q_0)   {$q_0$};  
   \node[state,accepting](q_f) [ right of= q_0, node distance=1.5cm] {$q_f$};
    \path[->] 
    (q_0) edge  node {$e$} (q_f);
          %edge  [loop above] node {$a,b,c,A,B$} ()
    %(q_f) edge [loop above] node {$a,b,c,e,A,B$} ();
\end{tikzpicture}}
\end{minipage}
\vspace{-1mm}
\caption{Two regular path queries}
\vspace{-5mm}
\label{fig:Q1Q2}
\end{figure}
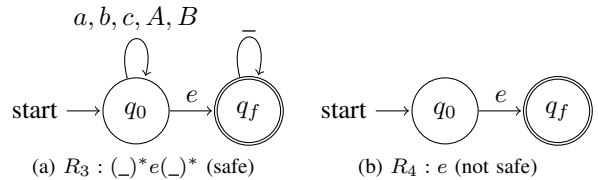

The DFAs of two queries are shown in Figures~\ref{fig:Q1Q2}.  $R_3$ is
safe with respect to $G$ since all state pairs (i.e., $(q_0,q_0)$,
$(q_f,q_f)$, $(q_0,q_f)$) are safe. 
For example, $(q_0,q_f)$ is safe because {\em all} executions of $S$ and $A$ contain a path whose tag transitions the DFA from $q_0$ to $q_f$, whereas {\em none} of $B$'s executions contains a path whose tag transitions the DFA from $q_0$ to $q_f$. 

In contrast, $R_4$ is not safe with respect to $G$ because $(q_0,q_f)$ is not safe. In particular, the two executions of $A$ 
shown in Fig.~\ref{fig:A_ex}, $ex_1$, $ex_2$,  behave differently:  There exists a path in $ex_1$ that transitions the DFA from $q_0$ to $q_f$, but there does not exist such a path in $ex_2$.
\end{example}

Note that the number of states of the DFA determines the size of the
fine-grained workflow. We now show that it is sufficient to check the
{\em minimal} DFA for $R$ in order to determine whether or not $R$ is
safe for a given workflow.

\begin{lemma}
Given a workflow $G$ and a regular expression $R$, $R$ is safe with respect to $G$ iff the minimal DFA of $R$ is safe with respect to $G$.
\end{lemma}
%\vspace{-2mm}

\begin{proof}
By definition~\ref{def:safeQuery}, it suffices to show that  if the minimal DFA of $R$ is  unsafe with respect to $G$ then all DFAs of $R$ are unsafe with respect to $G$. Specifically, for any unsafe state pair of the
minimal DFA, we can always find in any DFA of $R$ a corresponding
state pair that is unsafe.

Let ${\cal M}=(Q,\delta,q_0,F)$ be any DFA of $R$,  and 
$\\{\cal   M}_{min}=(Q_{min},\delta_{min},q^{min}_0,F_{min})$ be the minimal
DFA. By the definition of minimal DFA, for any state
$q\in Q$ there is an {\em equivalent} state $q_{min}\in Q_{min}$,
denoted by $q\equiv q_{min}$ (which we will define later) and
vice versa. Suppose $(q_1^{min},q_2^{min})\in Q_{min}\times Q_{min}$ is unsafe
with respect to some composite module $M$, i.e. there are two executions
$W_1,W_2$ of $M$ where there exists some path $p$ connecting the input
$i$ and output $o$ of $M$ in $W_1$
such that 
\begin{equation}\label{eq1}
\delta^{*}_{min}(q_1^{min},\tau_P(p))=q_2^{min}
\end{equation}
 while for every
path $p'$ connecting $i$ and $o$ in $W_2$,
$\\ \delta^{*}_{min}(q_1^{min},\tau_P(p'))\not=q_2^{min}$. Now choose a
state $q_1\in Q$ such that $q_1\equiv
q_1^{min}$. Let 
\begin{equation}\label{eq2}
q_2=\delta^*(q_1,\tau_P(p))
\end{equation}
 We prove that (1) $q_2\equiv q^{min}_2$ and thus $q_2$ is not a rejecting state; and (2) $(q_1,q_2)$ is
unsafe w.r.t. $M$. 

We first prove $q_2\equiv
q^{min}_2$. Proof is by contradiction.  Suppose
$q_2\not\equiv q^{min}_2$. Recall from \cite{hopcroft2007introduction}, a state $q$ of $\cal M$  is equivalent to a
state $q_{min}$ of ${\cal M}_{min}$ iff any string that transits $\cal
M$ from $q$ to some accepting state transits ${\cal M}_{min}$ from
$q_{min}$ to some accepting state and vice versa. Formally $q\equiv
q_{min}$ iff $\{w|\delta^*(q,w)\in
F\}=\{w|\delta_{min}^*(q_{min},w)\in F_{min}\}$. Take any string $w$
that transits ${\cal M}_{min}$ from $q^{min}_2$ to some accepting
state while cannot transit ${\cal M}$ from $q_2$ to some accepting state.
Because of Equation~\ref{eq1}, string $\tau_P(p)w$ transits
${\cal M}_{min}$ from $q^{min}_1$ to some accepting state while
because of Equation~\ref{eq2}, $\tau_P(p)w$
cannot transit $\cal M$ from $q_1$ to some accepting state, which
contradicts $q_1\equiv q^{min}_1$.

We prove the second one by first showing that for any string $w$, if
$q^{min}_1\equiv
q_1,q^{min}_2\equiv q_2$ and
$\delta_{min}^*(q^{min}_1,w)\not=q^{min}_2$, then
$\delta^*(q_1,w)\not=q_2$. Proof is similar to the first part.
\eat{Proof is
 by contradiction. Let 
$\delta_{min}^*(q^{min}_1,w)=q^{min'}_2$.  Suppose
$\delta^*(q_1,w)=q_2$. Similar to the proof for the first part, $q_2\equiv
q^{min'}_2$. Since $q_2\equiv q^{min}_2$, $q^{min'}_2\equiv
q^{min}_2$, which contradicts that no two states in a minimal DFA that
are equivalent \cite{hopcroft2007introduction}.}  Recall that no path in  $W_2$ can transit 
from $q^{min}_1$ to $q^{min}_2$. Thus no path in $W_2$ can transit
from $q_1$ to $q_2$. Revisiting Equation~\ref{eq2}, the path $p$ in $W_1$
can transit from $q_1$ to $q_2$. Thus by Definition~\ref{def:safeQuery},
$(q_1,q_2)$ is unsafe.
\end{proof}

In the rest of the paper, we refer to the minimal DFA of a
query as its DFA.

%\vspace{2mm}
{\bf Checking safety of DFA:}  
Intuitively, a DFA $\cal M$ is safe w.r.t. a workflow $G$ if the
fine-grained workflow $G_R$ obtained by intersecting $G$ with $\cal M$
is safe.
We must therefore understand the {\em internal edges} in modules of $G_R$, since a safe workflow is one in which the dependency between the input and output ports  of each composite module is deterministic w.r.t. all its executions (recall Section~\ref{sec:labeling}). 
\eat{For an {\em atomic module} $M$, if the execution of   $M$ causes the DFA to
transition from state $q$ to state $q'$, then we draw an internal edge from the input port representing
$q$ to the output port representing $q'$.}An {\em
  atomic module} leaves states of a DFA unchanged; thus we draw an
internal edge from the input port representing $q$ to the output port
representing $q$.
If the DFA  $\cal M$ is safe, then internal edges can be ``lifted" deterministically to composite modules.  

%\vspace{-1mm}
\begin{example}
Consider  the internal module edges in Fig.~\ref{fig:specQ}.  All atomic modules leave the 
state unchanged.  The execution of composite module $B$ leaves the states unchanged, whereas any execution of composite module $A$ causes a transition from $q_0$
to $q_f$, and from $q_f$ to $q_f$.   Thus $R_3$ is safe for $G$.
\end{example}
%\vspace{-1mm}

An algorithm for checking safety is as follows:  We denote by $\lambda(M,ex)$ for each execution $ex$ of each module
$M\in \Sigma$ a $|Q|\times |Q|$ 
%$0\mbox{-}1$ 
boolean matrix, such that
$\lambda(M,ex)[q_1,q_2]=1$ iff there exists a path in $ex$ whose tag can
transition the DFA from $q_1$ to $q_2$. Note that if the DFA is safe
with respect to the specification, then executions of the same module provide
the same matrix.  In this case, we simply use $\lambda(M)$, where $ex$ is understood from the context.
%($M\in \Sigma$) when query is safe. 
To check safety, we start by defining $\lambda(M)$ as the identity matrix for
any atomic module $M$, and then compute $\lambda(M)$ for composite
modules by verifying all the productions. A production $M\rightarrow
W$ is said to be {\em verifiable} if $\lambda$ is already defined for
all modules in $W$, so that $\lambda(M)$ can be computed. The
algorithm reports that the DFA is safe  if $\lambda$ is consistently
defined for all composite modules, and outputs $\lambda$ as a
by-product. To visit each production once, we 
 adapt the algorithm in
\cite{hopcroft2007introduction} of checking whether the language of a
given context-free string grammar is empty. The time complexity of
checking safety is then $O(|Q|^2*G)$.

\eat{
{\bf Time Complexity:} Given a query $R$, checking safety with respect to $G$ consists
  of 1) creating a DFA for $R$; 2) minimizing the DFA; and 3) checking
  whether the minimum DFA is safe. The third step dominates the first
  two since the size of the DFA  is the main cost.
While in general the first step produces a DFA state space of which is exponential in $|R|$ in time  $O(c^{|R|}*|\Sigma|)$~\cite{hopcroft2007introduction}, 
for the subclass of {\em deterministic regular expressions}~\cite{DBLP:journals/iandc/Bruggemann-KleinW98}, which are widely used in XML processing,
the DFA state space size and construction time is only $O(|\Sigma|*|R|)$. 
Furthermore,  whether or not $R$ is deterministic can be tested in  time which is linear in $|R|$~\cite{Groz:2012:DRE:2213556.2213566}.
\eat{The second step can be solved using the well-known algorithm
in~\cite{hopcroft2007introduction}, which runs in  time
$O(|Q|*|\Sigma|*log(|Q|))$.}% (where $n$ is the number of states). 
 Note that in our case, $|\Sigma|$ is bounded by the grammar size, $|G|$. 
The third step can be solved in $O(|Q|^2*|G|)$ as discussed above, and dominates the first two steps.
Thus the overall time complexity for general grammars is
$O(c^{|R|}*|G|^3)$, and for deterministic regular expressions is $O(
(|\Sigma|*|R|)^2*|G|)$= $O((|G|*|R|)^2*|G|)=O(|R|^2*|G|^3)$.
}

{\bf Time Complexity:} 
Given a query $R$, checking safety with respect to $G$ consists of 1)
creating a DFA for $R$; 2) minimizing the DFA; and 3) checking
whether the minimum DFA is safe.
The third step dominates the first two, since the DFA size is the main factor.
While in general the DFA may have a state space which is exponential in
$|R|$, precisely  $O(c^{|R|}*|\Gamma|)$,  %~\cite{hopcroft2007introduction}, 
for the subclass of {\em deterministic regular expressions}~\cite{DBLP:journals/iandc/Bruggemann-KleinW98}, which are widely used in XML processing,
the DFA state space size  is only $O(|\Gamma|*|R|)$. 
Note that in our case, $|\Gamma|$ is bounded by the grammar size, $|G|$. 
Thus the overall time complexity for general grammars is  $O(c^{|R|}*|G|^3)$, and for deterministic regular expressions is $O(|R|^2*|G|^3)$.

\subsection{Decoding \eat{Regular Path} Labels}
\label{sec:decodingReg}
%%%%%%%%%%%%%%%%%%%%%%%%%%%%%%%%%%%%%%%%%%%%%%%%%%%%%%%%%Deleted lots of stuff
\eat{To complete the transformation of $G$ to the  fine-grained  grammar $G_R$, we must connect the
input ports of a module to its output ports (internal module edges).  
Intuitively, if the execution of  module $M$ causes the DFA to
transition from state $q$ to state $q'$, then an internal edge connects the input port representing
$q$ to the output port representing $q'$.
More precisely, 
%we use the $\lambda$ function from the previous section:  For each $M \in \Sigma$,
if $\lambda(M)[q,q']=1$ then there is an internal edge $(q,q')$.  
This can only be done for safe queries. 

\vspace{-1mm}
\begin{example}
Now consider  the internal module edges in Fig.~\ref{fig:specQ}.  All atomic modules leave the 
state unchanged.  The execution of composite module $A$ causes a transition from $q_0$
to $q_f$, and from $q_f$ to $q_f$.  
\end{example}
\vspace{-1mm}

Since $G_R$ is a   fine-grained workflow, runs that are generated from
$G_R$ are also fine-grained.   For example, the fine-grained run in
Fig.~\ref{fig:stack-join} is an element of $L(G_R)$.}
%%%%%%%%%%%%%%%%%%%%%%%%%%%%%%%%%%%%%%%%%%%%%%%%%%%%%%%%%To here

We conclude this section by summarizing a constant-time algorithm for
answering safe regular path queries. 
%Since we use reachability  labels to answer regular path queries, we refer to them as {\em   regular path labels} in the rest of the paper.

%vspace{-1mm}
\begin{theorem}
Given a workflow $G$ and the labels of two nodes, a safe pairwise
query $R$ can be answered in constant time.
\end{theorem}
\begin{proof}
We prove this by presenting
Algorithm~\ref{algo:answerPairwiseRegularQuery}. We override the
decoding function $\pi$ in \cite{DBLP:conf/sigmod/BaoDM12}  by adding query
$R$ as a parameter. The subtlety is that the run is labeled when it is created (offline) rather that at query time.  At query
time, we first compute $G_R$, which runs in $O(1)$ time
w.r.t. the run size. We then decode labels, which also runs in $O(1)$ time
w.r.t. the run size assuming that any operation on two
words ($log\;n$ bits) can be done in $O(1)$ time \cite{DBLP:conf/sigmod/BaoDM12}. 
Pairwise safe queries can therefore be answered in constant
time.
\eat{
Given a fine-grained workflow $G'$, let $(\psi',\pi')$ be the reachability labeling scheme for
fine-grained workflows proposed in \cite{DBLP:conf/sigmod/BaoDM12},
where $G'$ is the fine-grained workflow; $\psi'$ assigns for each
module execution a label and $\pi'(\psi'(u),\psi'(v),G')$ outputs whether
$u\leadsto v$. 

By
Lemma~\ref{lemma:intersect}, given $R$, $(\psi_V,\pi,G)$ is equivalent
to $(\psi',\pi',G_R)$. The subtlety is, as we will describe in the next subsection, we
label the run query-independently, specifically $\psi_V=\psi'$. $\pi(\psi_V(u),\psi_V(v)
,G,R)$ firstly obtains $G_R$, and then outputs
$\pi'(\psi_V(u),\psi_V(v),G_R)$. $\pi(\psi_V(u),\psi_V(v),G,R)$ iff
$\pi'(\psi'(u),\psi'(v),G_R)$. 

Obtaining $G_R$ takes $O(1)$ time (w.r.t. run size), and $\pi'(\psi'(u),\psi'(v),G_R)$ can be done
in $O(1)$ time \cite{DBLP:conf/sigmod/BaoDM12}. The overall time complexity is $O(1)$.}
\end{proof}

\vspace{-2mm}
\setlength{\algomargin}{0.3em}
\SetAlgoSkip{smallskip}
\SetInd{0.1em}{1em}
\begin{algorithm}[h]
\small
% \SetKwRepeat{doWhile}{do}{while}
\DontPrintSemicolon
$\pi(\psi_V(u),\psi_V(v), G, R)$\\
\KwIn{The regular path labels of node $u$ and $v$, $\psi_V(u)$,
  $\psi_V(v)$, query $R$ and
  specification $G$\\}
\KwOut{whether $u\stackrel{R}\leadsto v$}
%\textbf{Initialization }\\
%Let $T_1$ ($T_2$) be the tree representation of $l_1$ ($l_2$)\\
\Begin{
Transform $R$ to its minimal DFA $\cal M$\\
\If{$\cal M$ is safe w.r.t. $G$}{
Intersect $G$ with $\cal M$ resulting in $G_R$\\
%\tcp{$\pi(\psi_V(u),\psi_V(v), G_R)$ is the predicate for reachability
%queries in Section~\ref{sec:labeling}}
\Return{$\pi(\psi_V(u),\psi_V(v), G_R)$ \tcp{Section~\ref{sec:labeling}}}
}}
\caption{{\em { answerPairwiseSafeQuery}}}
\label{algo:answerPairwiseRegularQuery}
\end{algorithm}
\vspace{-5mm}

\section{Answering All-Pairs Queries}
\label{sec:allpairs}

We now turn to all-pairs regular path queries. Recall that in this problem, we are
given a graph $g$, two list  $l_1,l_2$ of nodes of $g$, and a regular expression $R$, and return the set of all node pairs $(u, v)$ in $l_1\times l_2$ such that $u \stackrel{R}\leadsto v$.
We start by describing how to efficiently answer all-pairs {\em safe} regular
path queries before moving to {\em unsafe} (general) queries.

\subsection{Safe queries}
\label{sec:our-approach}
We assume that each node in the run $g$ is labeled using the algorithm presented in
Section~\ref{sec:labeling}, so that pairwise regular path queries $u\stackrel{R}\leadsto v$ can be answered in constant
time.  
We then do {\em structural joins} over the two lists $l_1$ and $l_2$ to find all pairs of nodes $(u, v)$ such that $(u,v)\in l_1 \times l_2$ and $u\stackrel{R}\leadsto v$.
In particular, we consider two types of structural joins.

\smallskip \noindent \textbf{Option S1: Nested-loop join.} A straightforward algorithm is to use nested loops to perform structural joins:  For each node $u$ in $l_1$,  iterate over the list $l_2$ and
return $(u,v)$ if $u \stackrel{R}\leadsto v$. Given that testing $u \stackrel{R}\leadsto v$ can be done by comparing the labels of $u$ and $v$ in constant time, the overall time complexity  
is $\Theta(|l_1| \times |l_2|)$. However, it turns out
  that we can do better.

\smallskip 
\noindent 
\textbf{Option S2: Reachable node pairs as a filtering step.}
In this option, we first find reachable node pairs in $l_1\times l_2$ (i.e. pairs
$(u,v)$ such that $u\leadsto v$) and then check if
$u\stackrel{R}\leadsto v$. 
We will show in Section~\ref{sec:experiments} that the filtering step significantly improves query performance
although the overall time complexity is proportional to the size of
the input, which is $O(|l_1| \times |l_2|)$.  

\begin{lemma}
The all-pairs safe regular query $R$ over input lists $l_1$ and $l_2$ from an edge-labeled run of workflow $G$ 
can be answered in time linear in $|l_1|$, $|l_2|$, $N$ and polynomial in $|G|$, 
where $N$ is the number of reachable node pairs in $l_1$,  $l_2$.
\end{lemma}

We prove this by presenting a corresponding algorithm.  We show that
%is shown in Algorithm~\ref{algo:answerAllPairsRegularQuery}.
all-pairs {\em reachability} queries can be answered in time
$O(|l_1|+|l_2|+N)$, which is optimal.
This can be done 
by using the fact  that if $u\leadsto v$,
then all nodes derived by $u$ can reach all nodes derived by
$v$.   The trick is to represent a list of nodes as a tree.

\textbf{Tree representation of a list of nodes.} Given a list $l$ of
node labels $\psi_V$, we  transform $l$ into an edge-labeled tree $T$ which is a projection of the compressed parse tree for $g$ and whose leaves correspond to the list $l$.  
Since a node label  is a list of entries $(k,i)$ or $(s,t,i)$, if $l$ is already sorted then the 
tree can be constructed in linear time.

Returning to the sample run in Fig.~\ref{fig:run} whose compressed
parse tree is in Fig.~\ref{fig:parse}\eat{of Section~\ref{sec:labeling}},  the list $l_1=\;\{a:1,d:1,b:3\}$ represented by their labels would be $\\\{(1,2)(1,1,1)(2,1), \; (1,2)(1,1,1)(2,3), \; (1,3)(4,2)\}$, and has the tree representation shown as $T_1$ in Fig.~\ref{fig:adlists}.  

Next, we {\em color} an edge $e=(u,v)$ in $T$ if the incoming edge of $u$ is from a recursive node, i.e. has label of form $(s,t,j)$ (see Section~\ref{sec:labeling}). 
Suppose $e$ is labeled as $(k,i)$, and $M \rightarrow W$ is the $k^{th}$ production; then $v$ is the $i^{th}$ node of $W$. Let $v'$ be the recursive node of $W$ (if there is one). 
Then $e$ is colored red if $v\leadsto v'$ in $W$ or blue if $v'\leadsto v$ in $W$.    
For example, in Fig.~\ref{fig:adlists} the edge 
$(A:1,a:1)$ of $T_1$ is colored red (indicated by the letter ``r") since node $a$ can reach the recursive node $A$ in $W_2$,
and edge $(A:1, d:1)$ is colored blue (indicated by the letter ``b") since the recursive node $A$ can reach node $d$.

\begin{figure}
\centering
\includegraphics[scale=0.25]{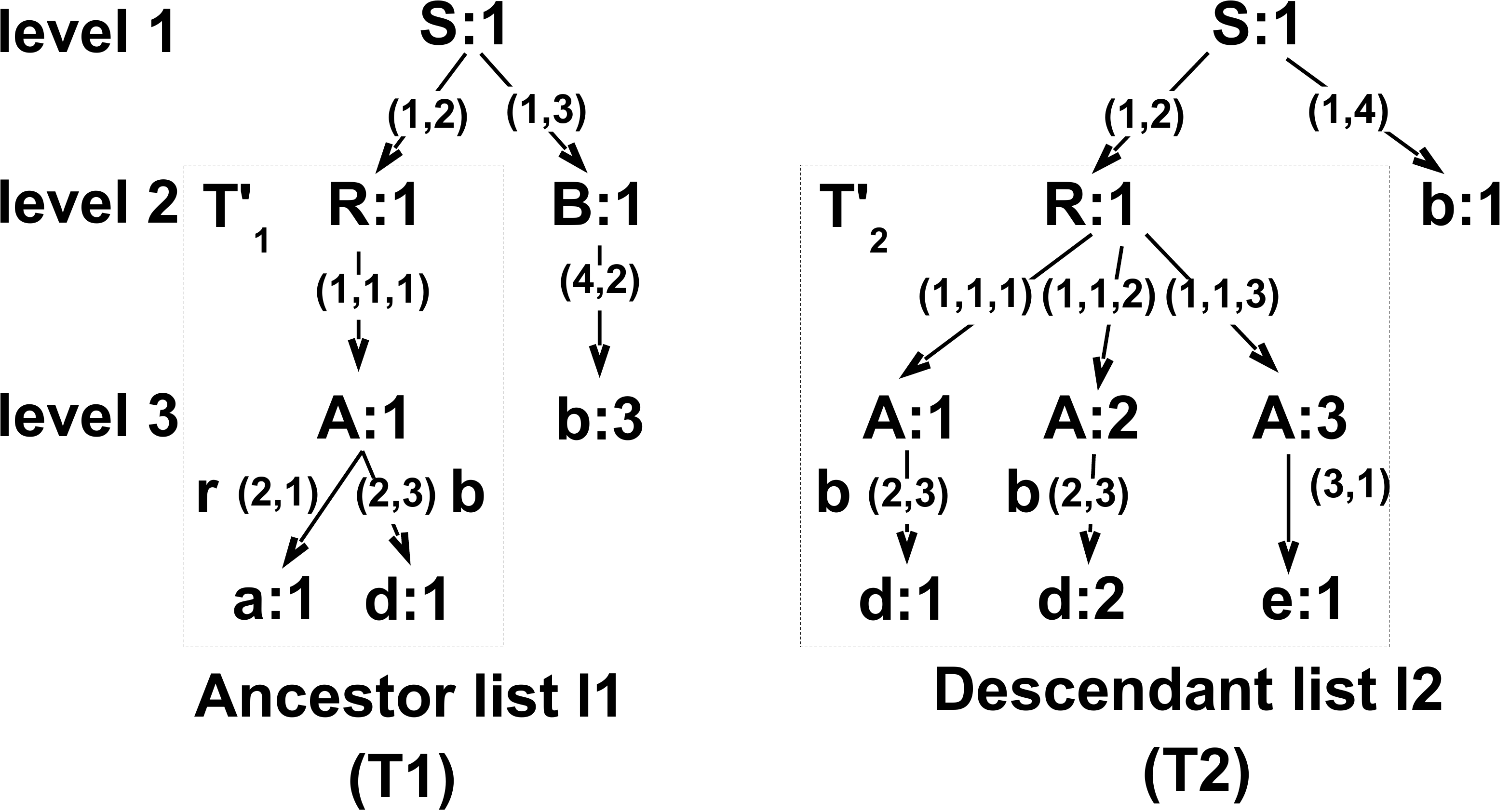}
\vspace{-1mm}
\caption{Two lists of nodes}
\vspace{-1mm}
\label{fig:adlists}
\end{figure}

\eat{The tree is built as follows.  Given a list of nodes $l$,  we assume nodes in
the list have been sorted in post-order as they appear in the
compressed parse tree. Recall that that the label of a node $v$, $\psi_V(v)$ is a
list of edge labels along the path from root in the compressed parse
tree. For each node $v$, we have a pointer to the first element of
$\psi_V(v)$. %Let $l[i]$ be the list of the $i^{th}$ element of labels
             %of nodes in $l$.
We start with a virtual root. We now identify the distinct
outgoing edges of the root, i.e. the distinct first elements of
labels of nodes in $l$, which can by obtained by reading the
pointers. We then for each child of the root  }

\textbf{Answering all-pairs reachability queries.} To answer all-pairs
reachability query, we traverse the
two trees $T_1$ and $T_2$ representing $l_1$ and $l_2$
using Algorithm~\ref{algo:answerAllPairsRegularQuery}.  This is done
top-down, level by level. Note that children of  the ``same" node in
$T_1$ and $T_2$ are either from the same simple workflow or from recursion (Cases 1 and 2, respectively, in Algorithm~\ref{algo:answerAllPairsRegularQuery}).
For example, the children of $S:1$ in Fig.~\ref{fig:adlists} are from the same simple workflow, whereas
the children of $R:1$ are from recursion.  Two nodes $v_1, v_2$ in $T_1,T_2$
are said to be same if they have the same labels, $\psi_V(v_1)=\;
\psi_V(v_2)$. To disambiguate, we denote by $\psi_T(u)$ the label of the incoming edge
of $u$ in $T$, e.g. $\psi_T(B:1)=(1,3)$.

\eat{
\setlength{\algomargin}{0.3em}
\setlength{\textfloatsep}{1mm}
\SetAlgoSkip{smallskip}
\SetInd{0em}{0.5em}
\begin{algorithm}[h]
\small
% \SetKwRepeat{doWhile}{do}{while}
\DontPrintSemicolon
\KwIn{Tree lists, $T_1,T_2$, 
  specification $G$ and safe query $R$\\}
%\KwOut{boolean}
\textbf{Initialization }\\
%Let $T_1$ ($T_2$) be the tree representation of $l_1$ ($l_2$)\\
Let $E_1 ~(E_2)$ be the outgoing  edges of  $root(T_1)$~($root(T_2)$)\\
Given an edge $e$ of $T_1$ (or $T_2$), let $\psi_T(e)$ be the label of
$e$, let $T(e)$ be the tree rooted at  the tail node of $e$\\ 
Let $Z(k,i,j)$ be a boolean representing the reachability between the $i^{th}$ and
$j^{th}$ node on the right-hand side of the $k^{th}$ production\\
\Begin{
%\tcp{There are only two cases}
\nl \If{$\forall e_1\in E_1, e_2\in E_2$, $\psi_T(e_1)=(k,i)$ and $\psi_T(e_2)=(k,j)$}
{// \textbf{Case 1:} Same simple workflow\\
%\tcp{$\forall e_1\in outE_1, e_2\in outE_2$, $\psi_T(e_1)=(k,i)$ and $\psi_T(e_2)=(k,j)$}
\ForEach{$e_1\in E_1, e_2\in E_2$}{
let $\psi_T(e_1)=(k,i),\psi_T(e_2)=(k,j)$\\
\nl\lIf{$i=j$}{$answerAllPairsSafeQuery(T(e_1),T(e_2))$}
\nl\lElseIf{$Z(k,i,j)$}{$output(leaves(T(e_1),leaves(T(e_2))))$}}}
\nl\ElseIf{ $\forall e_1\in E_1, e_2\in E_2$,$\psi_T(e_1)=(s,t,i)$ and
  $\psi_T(e_2)=(s,t,j)$} {
// \textbf{Case 2:} Recursion\\
%\tcp{$\forall e_1\in outE_1, e_2\in outE_2$, $\psi_T(e_1)=(s,t,i)$ and $\psi_T(e_2)=(s,t,j)$}
%\tcp{Note that $E_1$ and $E_2$ are sorted by the labels in ascending
%order of $i$}
// $E_1$ and $E_2$ are sorted in ascending order of
  $i$. We only list
the body of merge joining $E_1$, $E_2$\\
$mergeJoin(E_1,E_2)$\{ \\
let $p_1,p_2$ be the respective pointer to the current element of
$E_1$ and $E_2$; 
let $\psi_T(p_1)=(s,t,i)$, $\psi_T(p_2)=(s,t,j)$\\
\nl \lIf{$i=j$} { $answerAllPairsSafeQuery(T(p_1),T(p_2))$}
\nl \ElseIf{$i<j$}{
\nl\If{The tail node of $p_1$ has red outgoing edges}{
\ForEach{$e'\in E_2$ that is ordered after $p_2$}{
// Tail of $p_1$ derives tail of $e'$\\
let $e_r$ be a red outgoing edge of the tail node of $p_1$
\nl$output(leaves(T(e_r)),leaves(T(e')))$}}}
\nl \lElse{Similarly process $p_1$ and edges ordered before $p_2$\}}}}
%\SetAlgoNoLine
{\em output}(a list of node labels $l_1$, a list of node labels $l_2$ ) \{\\
\quad \lForEach{ $u\in l_1$, $v\in l_2$}{\\
\nl\quad \lIf{$\pi(\psi_V(u),\psi_V(v),G,R)$}{output($u$,$v$)\}}}
\caption{\em {answerAllPairsSafeQuery}}
\label{algo:answerAllPairsRegularQuery}
\vspace{-3mm}
\end{algorithm}}
%\vspace{-2mm}

\setlength{\algomargin}{0.3em}
\setlength{\textfloatsep}{1mm}
\SetAlgoSkip{smallskip}
\SetInd{0em}{0.5em}
\begin{algorithm}[h]
\small
% \SetKwRepeat{doWhile}{do}{while}
\DontPrintSemicolon
\KwIn{The tree presentations of two lists, $T_1,T_2$, 
  specification $G$ and safe query $R$\\}
%\KwOut{boolean}
\textbf{Initialization }\\
%Let $E_1 ~(E_2)$ be the outgoing  edges of  $root(T_1)$~($root(T_2)$)\\
%Given an edge $e$ of $T_1$ (or $T_2$), let $\psi_T(e)$ be the label of
%$e$, let $T(e)$ be the tree rooted at  the tail node of $e$\\ 
Let $root(T)$, $children(T)$, $leaves(T)$ be the root of tree $T$, the
set of children of
the root, the set of leaf-descendants of the root; let $T(u)$ be the
subtree rooted at node $u$\\
%Denote by $\psi_T(u)$ the label of the incoming edge of $u$ \\
%Let $Z(k,i,j)$ be a boolean representing the reachability between the $i^{th}$ and
%$j^{th}$ node on the right-hand side of the $k^{th}$ production\\
\Begin{
%\tcp{There are only two cases}
\nl \If{$root(T_1), root(T_2)$ are the same non-recursive node}
{// \textbf{Case 1:} Same simple workflow\\
%\tcp{$\forall e_1\in outE_1, e_2\in outE_2$, $\psi_T(e_1)=(k,i)$ and $\psi_T(e_2)=(k,j)$}
\ForEach{$u\in children(T_1), v\in children(T_2)$}{
let $\psi_T(u)=(k,i),\psi_T(v)=(k,j)$\\
\nl\lIf{$i=j$}{$answerAllPairsSafeQuery(T(u),T(v))$}
\nl\lElseIf{The $i^{th}$ node reaches the $j^{th}$ node on the right-hand side of the $k^{th}$ production }{$output(leaves(T(u),leaves(T(v))))$}}}
\nl\ElseIf{$root(T_1), root(T_2)$ are the same recursive node} {
// \textbf{Case 2:} Recursion\\
%\tcp{$\forall e_1\in outE_1, e_2\in outE_2$, $\psi_T(e_1)=(s,t,i)$ and $\psi_T(e_2)=(s,t,j)$}
%\tcp{Note that $E_1$ and $E_2$ are sorted by the labels in ascending
%order of $i$}
Let $Set_{op}$ ($op \in \{<, =, >\}$) be the set of all node pairs $(u,v)\in$ $children(T_1)$
$\times children(T_2)$ such
that if let $\psi_T(u)=(s,t,i)$, $\psi_T(v)=(s,t,j)$ then $i~op~j$, and if $op$ equals $<$ then $u$ has red
children, and if $op$ equals $>$ then $v$ has blue children\\
// $Set_{op}$ can be computed by merge join; details are omitted \\
\nl\ForEach{$(u,v)\in Set_=$} {$answerAllPairsSafeQuery(T(u),T(v))$}
\nl\ForEach{$(u,v)\in Set_<$, red child $u_r$ of $u$}{$output(leaves(T(u_r)),leaves(T(v)))$}
\nl \ForEach{$(u,v)\in Set_>$, blue child $v_b$ of $v$}{$output(leaves(T(u)),leaves(T(v_b)))$}}}
%\SetAlgoNoLine
{\em output}(a list of node labels $l_1$, a list of node labels $l_2$ ) \{\\
\quad \lForEach{ $u\in l_1$, $v\in l_2$}{\\
\nl\quad \lIf{$\pi(\psi_V(u),\psi_V(v),G,R)$}{output($u$,$v$)\}}}
\caption{\em {answerAllPairsSafeQuery}}
\label{algo:answerAllPairsRegularQuery}
\vspace{-3mm}
\end{algorithm}

For Case 1, let  node $u$ of $T_1$ and  node $v$ of $T_2$ be 
children of their respective root, and then $\psi_T(u)$ and $\psi_T(v)$
are of the form $(k,i)$. If $u$ can reach $v$ in the simple workflow for that production\footnote{Since $T_1$ and $T_2$ are both projections of the same compressed parse tree, the same productions must be used.}
(line 3), then all leaf-descendant nodes of $u$ in $T_1$ can reach 
all leaf-descendant nodes of $v$ in $T_2$. However, if $u$ and $v$
are the same node, then we need to move to the next level to process
the subtrees rooted at $u$ and $v$  (line 2).  E.g.,  for $T_1$, $T_2$
in Fig.~\ref{fig:adlists}, we first process level 1. The outgoing
edges of the root of $T_1$ are labeled $\{(1,2),(1,3)\}$, and
$\{(1,2),(1,4)\}$ for $T_2$, and fall into case 1 (line
1). Note that $(1,2),(1,3),(1,4)$ corresponds to $A,B,b$ in $W_1$,
respectively.  It is clear that $A$ can reach $b$
%i.e. $Z(1,2,4)\not= 0$ 
and hence all leaf-descendants of $(1,2)$ in
$T_1$ (i.e. $a:1,d:1$) are ancestors of all leaf-descendants of $(1,4)$
in $T_2$ ( i.e. $d:1$) (line 3). However, root of $T_1$ and $T_2$ have
the same outgoing edge $(1,2)$. We then recursively process $T_1'$ and $T_2'$
(shown in the rectangles). 

For Case 2, let a recursive node $u$ of $T_1$ and a recursive node $v$ of $T_2$ be a
child of the root, and then $\psi_T(u)$ ($\psi_T(v)$) is of the form
$(s,t,i)$ ($(s,t,j)$). Let $u_r$ be a child of $u$.
If $u$ derives $v$, i.e. $i<j$ (line 6), and $u_r\leadsto v$ (i.e. $(u,u_r)$ is a red edge) then leaf-descendants of $u_r$ can reach leaf-descendants of $v$. Similarly,
if $v$ derives $u$, i.e. $i>j$ (line 7), and $u\leadsto v_b$ (i.e. $(v,v_b)$ is a
blue edge) then leaf-descendants of $u$ can reach leaf-descendants of $v_b$.
If $u$ and $v$ are the same, we then move to the next level (line
 5). E.g., $A:1~(1,1,1)$ of $T_1'$ derives $A:2~(1,1,2)$ of $T_2'$, 
 hence all leaf-descendants of $a:1~(2,1)$ (red child) of $T_1'$
  can reach all leaf-descendants of $A:2$. 

\textbf{Answering all-pairs safe regular path queries.} For each
reachable node pair $(u,v)$ i.e. $u\leadsto v$, we invoke
Algorithm~\ref{algo:answerPairwiseRegularQuery} to check if
$u\stackrel{R}\leadsto v$ (Line 8).

\textbf{Time complexity.}\eat{The preprocessing step of Algorithm~\ref{algo:answerAllPairsRegularQuery} takes a list 
of node labels and builds its tree representation.
Assuming that the list is sorted by post-order of nodes in the
compressed parse tree, this can be done in time $O(|l|)$ and the
leaf-descendants of any node can be found in constant time. } Building
the tree representation of a list can be done in time linear in the
list size, assuming the list is sorted by post-order of nodes in the
compressed parse tree.
Case 1 %of Algorithm~\ref{algo:answerAllPairsRegularQuery} 
does  a nested-loop join. Since
the size of the outer and inner loop of the join is bounded by the
grammar size,  this  takes $O(|G|^2)$. 
Case 2 %of Algorithm~\ref{algo:answerAllPairsRegularQuery} 
does  a merge join and hence
takes $O(|l_1|+|l_2|)$. 
The body of the merge join runs in time linear in
output size of the function $output$. For each level of the two trees, Algorithm~\ref{algo:answerAllPairsRegularQuery} therefore takes
$O(|l_1|+|l_2|+max(|l_1|,|l_2|)*|G|^2 + N_l)$ where $N_l$ is the
number of reachable node pairs from this level. Since the height of the two trees
is bounded by the grammar size, letting $l=\; max(|l_1|,|l_2|)$, Algorithm~\ref{algo:answerAllPairsRegularQuery} takes
\eat{$O(|G|*(l*|G|^2)+ N)=$}$O(|G|^3 l + N)$  where $N$
is the number of reachable node pairs in $l_1, l_2$.

\subsection{General queries}
\label{sec:general}

\eat{ As observed in Section~\ref{sec:pairwise}, derivation-based dynamic labeling 
does not work for unsafe queries since the input-output port connections are not ``stable" between
executions, i.e. $\lambda$ cannot be computed.   
The approach that we propose, therefore, is to decompose a regular expression $R$ into a set of safe sub-expressions, and stitch the results together. }

We now discuss answering general (unsafe) all-pairs regular path queries.  We first describe three previous approaches, which will be used to provide a baseline for the experiments  in the next section, and then
describe our approach.

\eat{
The first two come from XML query processing, while the third uses the results of \cite{DBLP:conf/sigmod/BaoDM11}
for a special form of query.}

%\medskip \noindent \textbf{Previous approaches}

\eat{\noindent\textbf{Option G1:  Traverse the run.}
The brute-force approach is to traverse the run graph $g$ at query time. 
Starting with each node $v$ in turn, we perform a DFS over $g$ while simulating the state transitions in the DFA accepting $L(R)$. Since each invocation of this algorithm takes $O(|g|)$ time, the time complexity is $O(|g| * n)$, where $n$ is the number of nodes in $g$.}

\smallskip \noindent
 {\bf Option G1:  Represent $R$ as a tree and evaluate bottom-up using joins \cite{DBLP:conf/vldb/LiM01}.}
This approach treats a regular expression as a (binary/unary) tree
(parse tree),
where leaves are single symbols, and internal nodes are union,
concatenation, or Kleene star. We then evaluate the tree
 bottom-up. \cite{DBLP:conf/vldb/MendelzonW89} is too slow we omit
it.
 
\eat{For each leaf $\lambda$, we retrieve (using an
index) a list of node pairs
$(u,v)$ where $u,v$ are connected by an edge tagged $\lambda$. We
then perform joins up to the root.

For example, given two regular expressions $e_1,e_2$ and
the list of node pairs $l(e_1)$, $l(e_2)$ satisfying $e_1$, $e_2$, respectively.
The result of $e_1+ e_2$ is  $l(e_1) \cup l(e_2)$, and
the result of $e_1e_2$ is $\{(u,v)|(u,w)\in l(e_1),(w,v)\in l(e_2)\}$. 
Kleene star repeats concatenation until a fixpoint is reached.
The complexity of this approach is quadratic in the number of nodes in the run graph $g$
due to the Kleene star operator.
}

\smallskip \noindent 
\textbf{Option G2: Use rare edge labels  \cite{DBLP:conf/ssdbm/KoschmiederL12}.} ``Rare'' edge labels are ones which match very few node pairs.  The approach decomposes a query to a series of smaller subqueries using rare labels, then performs a breadth-first search on the graph.

\eat{What [20] does is to Use rare edge labels to perform multiway search where rare edge labels are the ones matching few node pairs. E.g. Suppose query R= R1 a R2 where edge label a matches only a few node pairs (less than a threshold); then [20] will split the query to R1 and R2 and then perform breadth-first search starting from a, to R1 and R2 respectively, (so-called multiway search). I simply don
	t think it's necessary to spend much space on explaining how other approaches work.}

\smallskip \noindent
 \textbf{Option G3: Use reachability labeling \cite{DBLP:conf/sigmod/BaoDM11} combined with indexing for queries of a special form.}
Regular expressions of the form $R = (\_)^*a_1(\_)^*a_2(\_)^*\ldots(\_)^*a_k(\_)^*$
% which we will refer to as IFQs ((I)ndex (F)riendly (Q)ueries) in experiments, 
can be decomposed into $k$
sub-expressions of the form $R_i = a_i$.  The set $l_i$ of nodes pairs $(u_i, v_i)$ matching $a_i$ can be found using indexing, and reachability tested between $v_i$ and $u_{i+1}$ using dynamic labeling.
\eat{The benefit is that, for such expressions, allpairs reachability
  can be answered in time linear in input and output size ($|l_i| +
  |l_{i+1}|$).} 

\medskip \noindent \textbf{Our approach:} %  Decompose $R$ into a set of safe sub-expressions.}
We first represent the regular expression $R$ as a (binary/unary) parse tree 
\cite{DBLP:conf/vldb/LiM01} as described in Option G1, and 
find its {\em largest safe subtree}, which then can be evaluated using the approach described in Section~\ref{sec:our-approach}.   The remainder of the query can then be evaluated using 
Option G1.

Given a tree representing $R$, to find the largest
safe subtree, we traverse the tree top-down from the root.  For each
subtree, we check the safety of the regular expression it represents
using the algorithm given in Section~\ref{sec:constraints}.
If the
subtree is unsafe, we move to its child subtrees until we find a subtree
that is safe. We will show in experiments that this simple heuristic
 yields significant performance improvements over previous
 approaches.  

It is worth noting that there may be  many
 trees equivalent to $R$ due to query rewriting \cite{DBLP:conf/pods/CalvaneseGLV99}. 
 Finding  the {\em largest safe subexpression} is therefore an interesting optimization problem, which we leave for
 future work. 
\eat{ an interesting optimization problem, which we have not yet studied.  The approach we currently use is to work bottom up from the leaves of the tree and until
the subtree becomes unsafe.  Note that safe expressions are closed under
union, but not under concatenation. However, it is straightforward (albeit potentially time-consuming when more than a few symbols are concatenated) to enumerate all possible concatenations of
subexpressions to find the biggest safe one. Safe expressions are also not closed under Kleene star, which is more problematic.    However, our experimental results show
that this simple heuristic yields significant performance improvements over the previous approaches.}
\eat{
\medskip \noindent \textbf{Option 5: Decompose the DFA for $R$ into a set of safe DFAs.}
The benefit of this approach is that we can reuse the safety result
computed in Section~\ref{sec:prelims}. The optimization problem is to
find the biggest safe sub-DFA of the given DFA. Particularly the sub-DFA should
be a single-source-single-sink graph and all incoming edges from the
outside of the sub-DFA go to the source node and all outgoing edges to
the outside are from  the sink.
}
\eat{It is also possible to decompose the DFA for $R$ rather than the tree for $R$.  How this is done, and whether it yields significant improvements over our current approach, also remains as future research.}

%%% Local Variables:
%%% TeX-master: "rpl"
%%% End:
%\input{example}
\section{Experiments}\label{sec:experiments}

To evaluate the performance of all-pairs general regular path queries, we evaluate each component  separately.   
We start in Section~\ref{subsec:ex_overhead} by evaluating
overhead, and then evaluate pairwise safe queries (Algorithm~\ref{algo:answerPairwiseRegularQuery}) in Section~\ref{subsec:ex_pairwise} and all-pairs safe queries (Algorithm~\ref{algo:answerAllPairsRegularQuery}) in Section~\ref{subsec:allpairssafe}.  
Finally, we evaluate the performance of all-pairs general queries in Section~\ref{subsec:exp_general}.

\begin{figure*}
\centering
\subfloat[Time overhead (vary grammar size)]{\label{fig:overhead_syn}\includegraphics[scale=0.45]{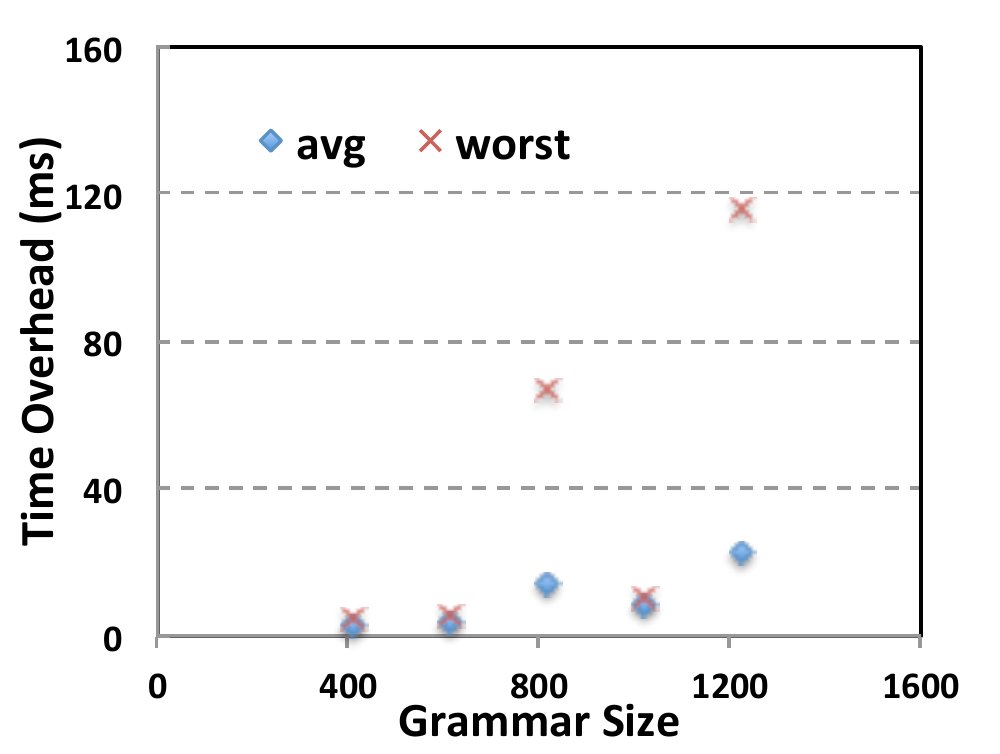}}
\subfloat[Time overhead (vary query size)]{\label{fig:overhead_real}\includegraphics[scale=0.45]{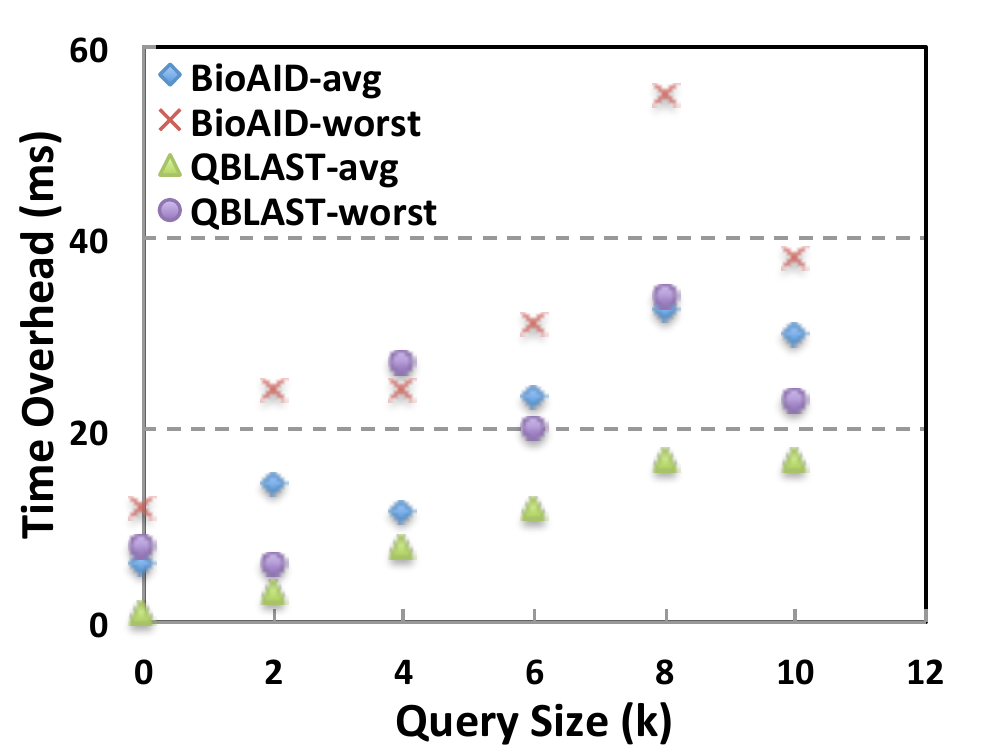}}
\subfloat[pairwise query time (vary run size)]{\label{fig:pairwise_varyRun}\includegraphics[scale=0.45]{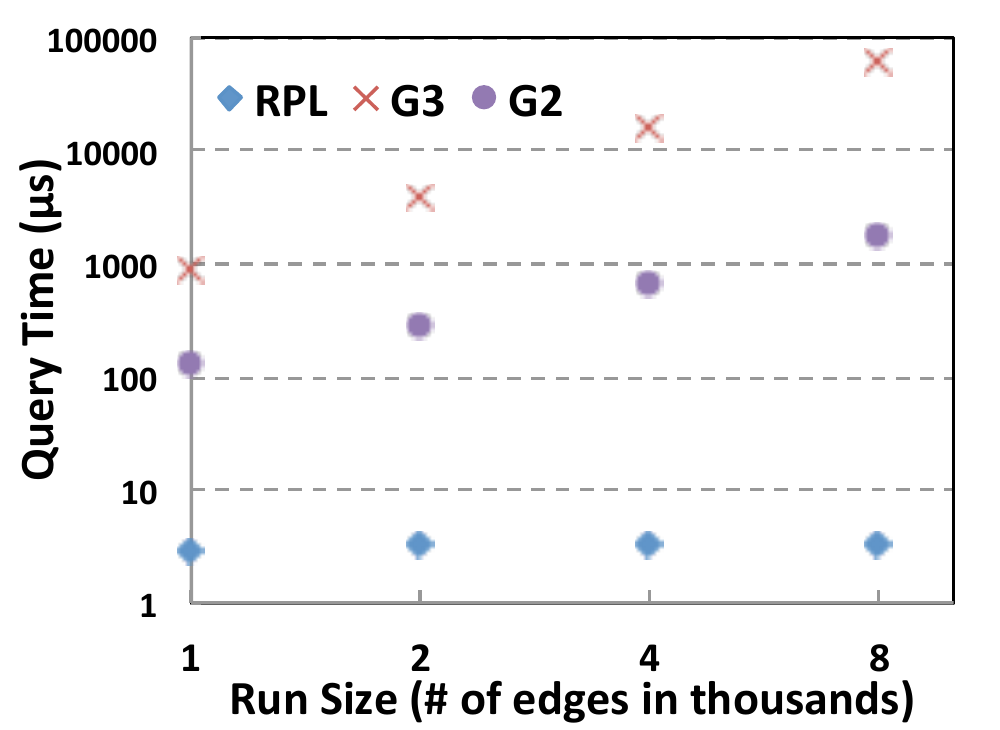}}
\subfloat[pairwise query time (vary query size)]{\label{fig:pairwise_varyQ}\includegraphics[scale=0.45]{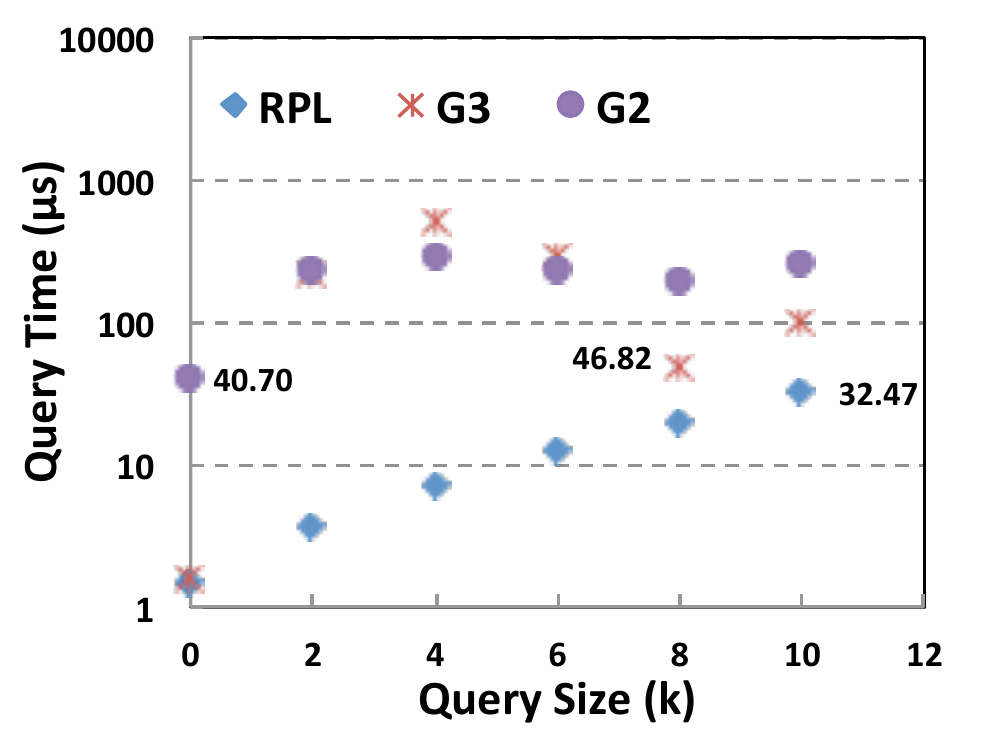}}
\\
\vspace{-3mm}
\subfloat[all-pairs query time (BioAID)\newline  $8$ IFQs ($k=3$), run size=$2K$]{\label{fig:allpairs_selectivity_BioAID}\includegraphics[scale=0.45]{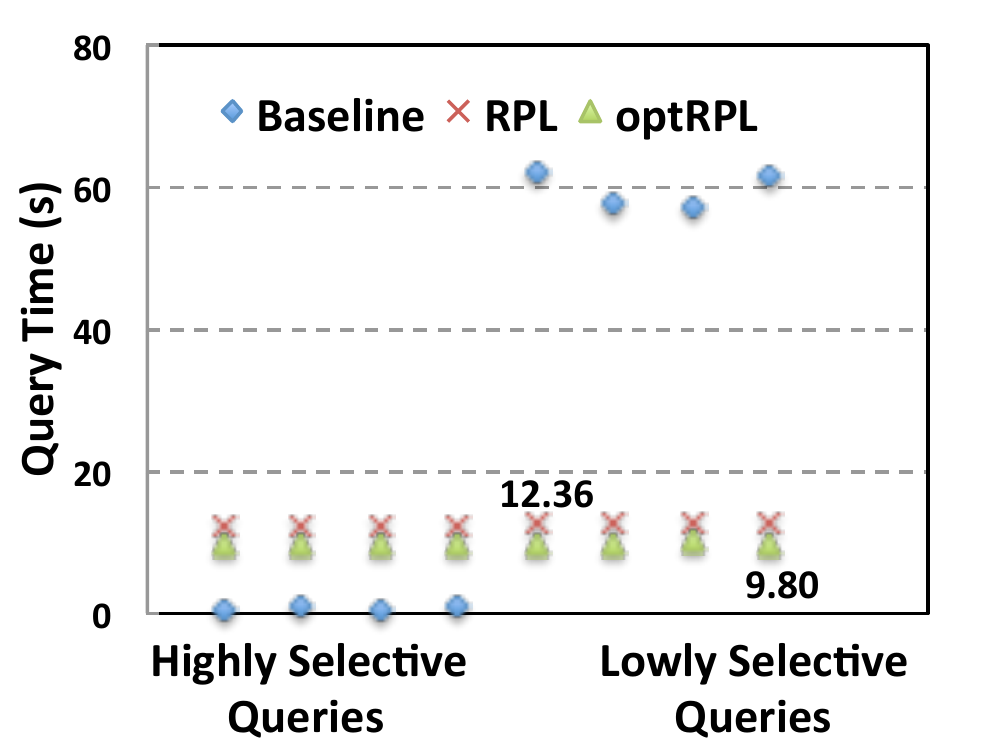}}
\subfloat[all-pairs query time (QBLast) \newline $8$ IFQs ($k=3$), run size=$2K$]{\label{fig:allpairs_selectivity_QBLast}\includegraphics[scale=0.45]{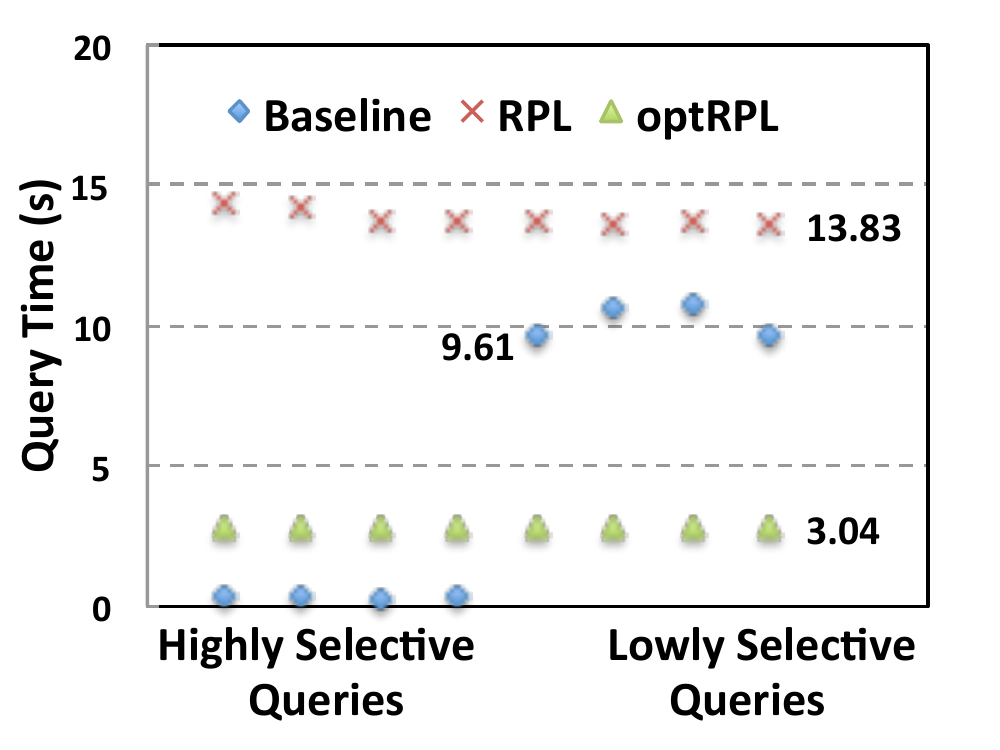}}
\subfloat[all-pairs  query time (BioAID)\newline $R=a^*$, run size=$2K$]{\label{fig:allpairs_astar_BioAID}\includegraphics[scale=0.45]{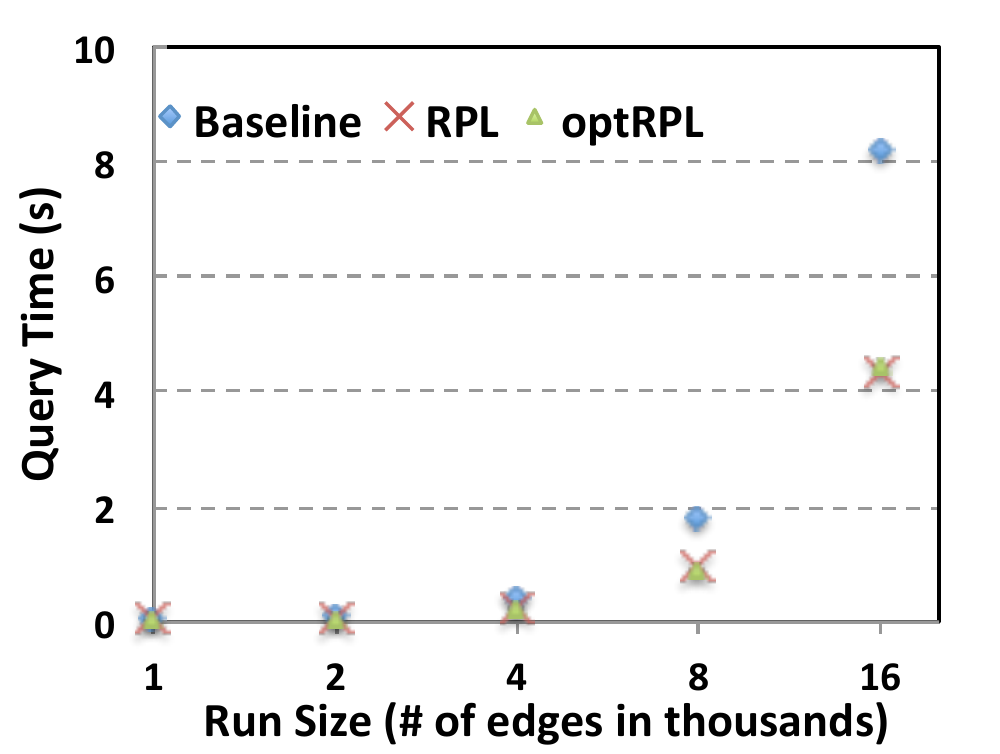}}
\subfloat[all-pairs query time (QBLast)\newline $R=a^*$, run size=$2K$]{\label{fig:allpairs_astar_QBLast}\includegraphics[scale=0.45]{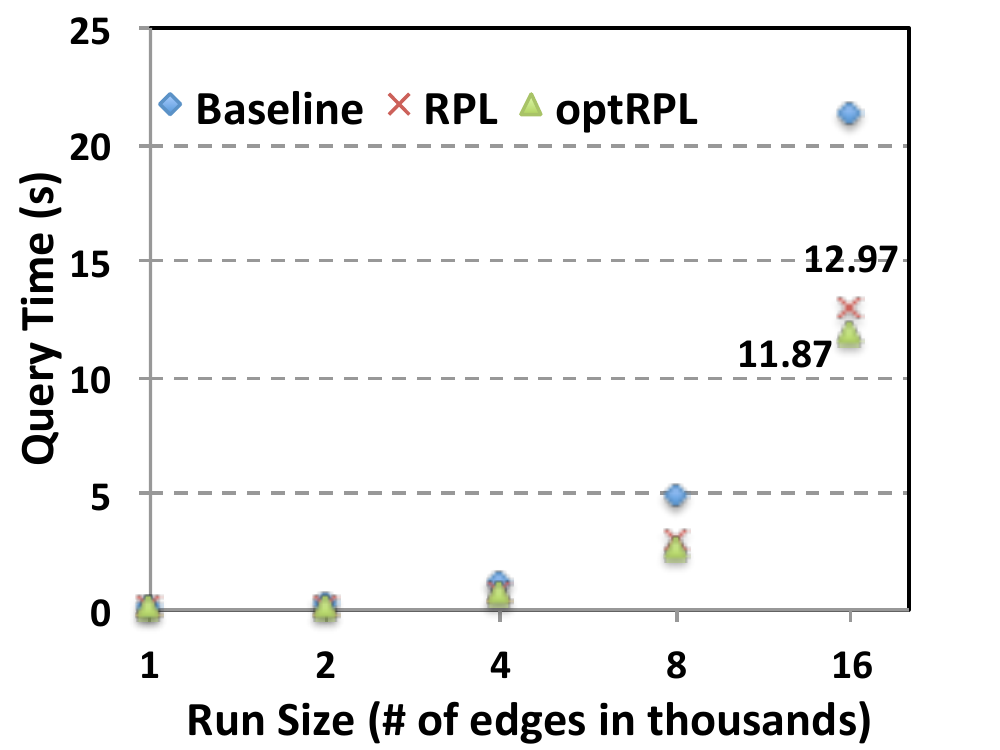}}
\vspace{-0.2cm}
\caption{Performance of our approaches (RPL and optRPL) for safe queries}
\vspace{-5mm}
\label{fig:ex_safe}
\end{figure*}

\subsection{Experimental Setup}
Experiments were performed on a
Mac Pro with Intel Core i5 2.3GHz CPU and 4G memory. We use the library \cite{automaton} to parse regular
expressions and minimize DFAs. 
\eat{We measure the query time,
% for safe and general regular path queries. Note that 
which for our algorithm includes parsing a regular expression, transforming to a DFA,
minimizing a DFA, and decoding  labels. } All reported running times are averages of 5 sample runs per setting. 

{\bf Realistic and Synthetic Datasets.}  Realistic scientific workflows were collected from  {\em myExperiment} \cite{DBLP:journals/fgcs/RoureGS09}. We report on results for two representative, recursive workflows, BioAID and QBLast.  BioAID is deep while QBLast is more ``branchy''. BioAID, of size 166, has 112 modules (16 of which are composite) and 23 productions (7 of which are
recursive)\footnote{\small The size of a workflow is the sum of the size of all productions where the size of a production equals one plus the number of modules on the right-hand side.}; QBLast, of size 105, has 77 modules (11 of which are composite) and 15 productions (5 of which are recursive). 
%More characteristics of BioAID and QBLast can be found in \cite{DBLP:conf/sigmod/BaoDKR10}.\eat{ The simple workflow on the right-hand side of each production contains at most 19 modules. } 
%We assign each edge a tag which is the module name of the head of the edge. 
To evaluate the overhead of
our approach, we create a set of synthetic workflows while varying workflow
parameters (e.g. size, recursion depth, node degree).
%and etc. \cite{DBLP:conf/sigmod/BaoDKR10}. 
Due to lack of space, we only report the overhead while varying workflow size.

Since  {\em myExperiment}  does not record executions,  we simulate
runs. If not specified, we apply a random sequence
of productions, varying run sizes (i.e. the number of edges) from 1K
to 8K by a factor of 2, and labeling the nodes as they are
generated. \eat{We stop at 8K because Option G2  (Section~\ref{sec:general})
overflows the memory at 16K due to the size of  intermediate result, which can be exponential in the run size.}
All executions are stored as Java serializable objects on disk. The loading time is omitted\eat{ in our experiments}. 

In addition, we build indices to support comparisons (Option G3).
\eat{We build inverted indices for each run to help identify adjacent node pairs.} 
For each run, an index maps an edge tag $\gamma\in\Gamma$  to a list of node pairs 
%$(u,v)$ where $u,v$ 
that are connected by an edge tagged $\gamma$. We
store indices as Java serializable objects and materialize them on
disk. The running time %of the comparing approaches
for all-pairs queries thus includes disk access for indices. Although we could
further reduce the index access time by using more sophisticated indices,
the inverted index is fast enough (below $10ms$) in light of the query
time being more than $1s$.

\eat{we will show in
  Section~\ref{subsec:ex_pairwise} that the disk access for indices
  are fast enough (compared with query time) that  it will not hurt the overall performance of existing
approaches and thus the comparison still makes sense.}

\smallskip
{\bf Queries.}  We test using two classes of queries known to be
expensive: 
%We generate safe and unsafe queries of various sizes. 
\eat{To illustrate the performance of our regular path labels (for safe
queries),  }

(1) {\bf  IFQs}  ($R=(\_)^*a_1(\_)^*a_2\ldots a_k(\_)^*$)  ask for
  node pairs that are processed by a sequence of modules. IFQs can be handled
  by indexing and reachability labels (Option G3), 
%for which we use the state-of-the-art reachability labeling scheme \cite{DBLP:conf/sigmod/BaoDM12},
  and are therefore challenging to improve upon. {\em We will show that our approach beats the baseline
when queries are not too selective.}

(2)  {\bf Kleene Star ($R=a^*$)},  used to query provenance for recursions,
%for forks and loops \cite{DBLP:conf/icde/BaoBDEK09} (recursion in grammars)
is at the other extreme. The baseline for Kleene star is Option G2, which performs
 self-joins on node pairs that are connected by an edge labeled $a$ until a fixpoint is reached. Since it is unknown how many rounds it takes to
 reach a fixpoint, the performance\eat{ of this approach} can be
 very bad. {\em We will show that our approach achieves a major gain in this case.}

\eat{It is worth noting that Kleene star is important since workflows may
be recursive, and thus subworkflows may be executed repeatedly either
in series (looping) or in parallel (forking). {\color{red} Thus to evaluate Kleene
star, we generate runs by executing one cursion repeatedly}.}

%\eat{
It is worth noting that Kleene star is important for querying
workflow provenance because of fork and loop operations. We illustrate
by the simplest Kleene star expression $a^*$. For example, one fork operation in BioAID is captured by two productions shown
in Fig.~\ref{fig:fork}, where module $A$ is  composite, whose
productions are omitted and
module $M$ is to call $A$ in forks; atomic module $a$ is a fork distributor and
atomic module $b$ is a fork aggregator. When executing, the first production may be
fired many times, the result run of which is shown in
Fig.~\ref{fig:forkrun}. When a user queries for data that are
processed by forks, he may issue query $a^*$. To evaluate the
performance of $a^*$, we generate runs by firing the specified fork
recursion many times and other recursions only once, while varying run size from 1K to 16K by a factor of 2.

\usetikzlibrary{arrows}
\begin{figure}
\centering
\subfloat[Fork]{\label{fig:fork}\begin{tikzpicture}
\node[draw,circle,inner sep=1pt] (head) {M};
\node[draw,circle,inner sep=1pt, right of =head] (v0) {a};
\draw[double,->] (head.east) to (v0.west);

\node[right of =v0, node distance = 25] (center) {};
\node[draw,circle, inner sep=1pt, above of=center, node distance=15] (v1) {M};
\node[draw,circle,inner sep=1pt, below of=center, node distance=15] (v2) {A};
\node[draw,circle, inner sep=1pt, right of =center, node distance =30](v3) {b};
 \path[->,double]
            (v0)    edge  node [above]{$a$}     (v1)
 (v0)    edge  node[below] {$a$}     (v2)
 (v1)    edge  node[above] {M}     (v3)
 (v2)    edge  node[below] {A}     (v3)
;

\node[right of = v1, node distance = 40] (top){};
\node[right of = v2, node distance = 40] (bottom){};
\draw[-] (top) -- (bottom);
\node[draw, circle,inner sep=1pt, right of = center, node distance = 50]{A};

%\draw[double,->] (v0) -- (v1) \node[above]{d};

\end{tikzpicture}
}
\subfloat[Run of Fork]{\label{fig:forkrun}\tikzstyle{level 1}=[level distance=0.8cm, sibling distance=0.8cm]
\tikzstyle{level 2}=[level distance=0.8cm, sibling distance=0.8cm]

% Define styles for bags and leafs
\tikzstyle{bag} = [circle, draw=black,text centered,inner sep=1pt]

% The sloped option gives rotated edge labels. Personally
% I find sloped labels a bit difficult to read. Remove the sloped options
% to get horizontal labels. 
\begin{tikzpicture}[grow=right, sloped]
\node[bag] {a}
    child {
        node[bag] {A}    
            child {
             node[right] {...}
                edge from parent node[below]  {A}
            }
edge from parent node[below]  {$a$} 
    }
    child {
        node[bag] {a}         
        child {
                 node[bag] {A}    
                 child {
             node[right] {...}
                edge from parent node[below]  {A}
            } edge from parent node[below]  {$a$}
                } 
        child {
                node[bag] {a}
                child {
                  node[bag] {A}    
  child {
             node[right] {...}
                edge from parent node[below]  {A}
            }
                 edge from parent node[below]  {$a$}
               } 
                child {
                  node[right] {...}
                  edge from parent
                  node[above]  {$a$}
                }
                edge from parent node[above]  {$a$}
              }
                edge from parent
                node[above]  {$a$}
            };
\end{tikzpicture}
}
\vspace{-1mm}
\caption{A fork example of BioAID}
\vspace{-1mm}
\end{figure}
%}

\eat{To show how labels can be used to answer general
queries,} We also generate queries by randomly combining edge tags using
concatenation, union, and Kleene star. Due to space
  constraints, we do not explicitly list these queries. 
    \eat{Since we use labels to answer regular path queries, we will refer
    to the labels as {\em regular path labels} in our experiments.}

\eat{
\begin{figure*}
\centering
\subfloat[FRPL vs IRL ($8$ IFQs) \newline (Input list size
$|l_1|*|l_2|=
2M$)]{\label{fig:ex_IFQs}\includegraphics[scale=0.45]{figs/INRLvsFRPL.pdf}}
~
\subfloat[FRPL vs Non-RPL ($R=a^*$)]{\label{fig:ex_FRPLvsNonRPL}\includegraphics[scale=0.45]{figs/FRPLvsNonRPL.pdf}}
~
\subfloat[FRPL vs
NRPL]{\label{fig:ex_FRPL}\includegraphics[scale=0.45]{figs/FRPLvsNRPL.pdf}}
\vspace{-0.5cm}
\caption{Performance of All-pairs Safe Queries}
\vspace{-3mm}
\label{fig:ex_safe}
\end{figure*}}

\subsection{Overhead of Our Approach}
\label{subsec:ex_overhead}
In this section, we use {\bf RPL} ({\bf R}egular {\bf P}ath {\bf L}abels) for our approach to pairwise queries (Algorithm~\ref{algo:answerPairwiseRegularQuery}) and Option S1 for all-pairs queries.  We use {\bf optRPL} for Option S2 (Algorithm~\ref{algo:answerAllPairsRegularQuery}).

The overhead of our approach comes from checking the safety of a query w.r.t. a workflow (Section~\ref{sec:constraints}). We thus evaluate
overhead while varying the grammar and  query size. Fig.~\ref{fig:overhead_syn} shows the average and worst time overhead of $20$ IFQs with $k=3$ on synthetic workflows of size varying from $400$ to $1200$ ($10$ workflows per size).  Fig.~\ref{fig:overhead_real} shows the average and worst time overhead of IFQs varying $k$ from $0$ to $10$ on both BioAID and QBLast. We can see that
time overhead increases as the grammar or query grows. Nonetheless the time overhead is  $<200ms$, which is
acceptable compared to the query time in seconds (see Section~\ref{subsec:allpairssafe}).

\subsection{Performance of Pairwise Safe Queries}
\label{subsec:ex_pairwise}
%%%%%%%%%%
We test the performance of pairwise safe IFQs on real datasets while comparing Option G2 and G3 and %our Algorithm~\ref{algo:answerPairwiseRegularQuery} referred as
RPL. Option G1 is clearly worse than Option G3 for IFQs so we omit it here.
\eat{(1) (D)epth-(F)irst (S)earch ({\bf DFS}), which traverses the
provenance graph using depth-first search while simulating DFA
transitions; (2) (I)ndex + (R)eachability (L)abels ({\bf IRL}), which
works only
for IFQs ($R=(\_)^*a_1(\_)^*a_2\ldots
a_k(\_)^*$) and for which we use the reachability labeling
scheme in \cite{DBLP:conf/sigmod/BaoDM12}; and (3) our new approach in Section~\ref{sec:pairwise},
(R)egular (P)ath (L)abels ({\bf RPL}).}
We vary run size and query size. QBLast has similar trends to BioAID, so we only report on BioAID.  The pairwise query time  is in
microseconds, so we use $10K$ node pairs  and report the average time of $5$ queries per setting. For RPL, the query time thus includes time overhead amortized over $10K$ node pairs. Fig.~\ref{fig:pairwise_varyRun} reports the query time of  IFQs of size $3$ on runs varying in size from $1K$ to $8K$. We can see that  RPL runs in almost constant time (below $10\mu s$) while the run times of the other approaches grow sharply as the run size grows, from $100\mu s$ to $1 ms$ and $1ms$ to $100ms$ respectively. Fig.~\ref{fig:pairwise_varyQ} reports the  query time of IFQs of size $k=0$ to $k=10$ on runs of size $2K$. We can see that the query time of RPL grows as the query size grows, but remains below $40 \mu s$. In contrast, the query time of Option G2 and G3 goes above $40 \mu s$.  Note that when $k=0$, the query degrades to a reachability query and thus Option G3 has constant
query time. It is worth noting that Option G2 and G3 do not show a clear trend; this is because they rely on query selectivity, which we will discuss in the next section. {\em In summary, RPL significantly outperforms Options G2 and G3 for pairwise safe queries.}
\eat{Note that our approach RPL involves {\em preparation},   i.e. it checks query safety and precomputes matrices (see
Section~\ref{sec:pairwise}). The preparation is done once per
query and can be amortized over all node pairs. \eat{ Recall that the preparation time
of RPL depends on the grammar size and query size.} In our experiments,
the preparation time of RPL is around 4ms. The query time reported in
Fig.~\ref{fig:ex_pairwise} includes the preparation time amortized
over 10K node pairs.}

\eat{IRL accesses the disk to retrieve indices; and RPL  checks
{\color{red} For DFS, we assume the graph is in the
memory.} Table~\ref{ex:pre_time} shows the
preparation time for all three methods. The preparation time for
\eat{DFS and} IRL depends on the run size, while the time for RPL depends on the
 grammar and query size. Although RPL's preparation time is
 sometimes worse, it does not hurt the overall query performance since
 it is amortized over all nodes pairs, which can be in the thousands.}
\eat{\color{red}Although the preparation time for IRL (below $10ms$) could
  be further reduced by using complex
index techniques, we can see that it does not affect the overall query performance since
 it is amortized over all nodes pairs, which can be in the thousands.}

\eat{
\begin{table}
\centering
\begin{tabular}{|l|l|l|c|}
\hline
Run Size & DFS & IRL & RPL ($|Q|\leq 5$)\\
\hline
1K&{\color{red} 0}&1.8&4.0\\
\hline
2K&{\color{red} 0}&2.4&4.0\\
\hline
4K&{\color{red}0}&5.2&3.9\\
\hline
8K&{\color{red}0}&9.3&4.1\\
\hline
\end{tabular}
\vspace{-1mm}
\caption{Preparation time (ms)}
\vspace{-0.2cm}
\label{ex:pre_time}
\end{table}
}

\eat{
\begin{figure}[h]
\centering
\includegraphics[scale=0.5]{figs/pairwise_decoding.pdf}
\caption{Query time of pairwise RPQ (Run Size=2K)
 $(R=(\_)^*a_1(\_)^*\ldots a_k(\_)^*)$}
\label{fig:ex_pairwise}
\end{figure} }

\eat{Fig.~\ref{fig:ex_pairwise} reports the time to answer pairwise
 queries while varying the query size\eat{, which includes the
preparation time amortized over 10K node pairs}. Note that when $k=0$,
the query degrades to a reachability query and thus IRL has constant
query time.
%IRL and DFS have linear time  complexity in the run size, and
The  time for IRL and DFS
grows sharply as the query grows, from $1 \mu s$ to $650 \mu s$ and $368 \mu s$ to $768 \mu s$, 
respectively. In contrast, RPL has almost constant query time, below $4 \mu$s,
and significantly outperforms IRL and DFS. Even though
the RPL query time appears to be constant, a closer look shows that it increases slightly
as the query size increases. As $k$ increases from $0$ to $4$, the query
time of RPL increases from $1\mu s$ to $4 \mu s$.}
\eat{
\begin{figure*}
\centering
\eat{
\subfloat[Query
time]{\label{fig:ex_RPL}\includegraphics[scale=0.55]{figs/RPLvsBaseline.pdf}}}
\subfloat[Query Time of Pairwise RPQ (Run Size=2K)  $(R=(\_)^*a_1(\_)^*\ldots a_k(\_)^*)$]{\label{fig:ex_pairwise}\includegraphics[scale=0.55]{figs/pairwise_decoding.pdf}}
~
\subfloat[Improvement]{\label{fig:ex_imprvRPL}\includegraphics[scale=0.55]{figs/impvRPL.pdf}}
\caption{Performance of General Queries (Run size = 4K)}
\end{figure*}}

\subsection{Performance of All-pairs Safe Queries}
\label{subsec:allpairssafe}
We now show how labels can be used to help
answer all-pairs safe queries. \eat{ Our approaches are named by RPL (Option
S1) and optRPL (Option S2).} Since Option G2 is designed to run in parallel for all-pairs queries, it would be unfair to include it. So for IFQs, which favor existing approaches, the baseline is\eat{ the best
existing approach i.e. } Option G3. For Kleene
star, which favor our approach, the baseline is\eat{ the best
existing approach i.e.} Option G1.

Fig.~\ref{fig:allpairs_selectivity_BioAID}  reports the query time for
$8$ IFQs of size 3 $k=3$
over runs of size $2K$ (the
input lists $l_1$, $l_2$ are the list of all nodes) on BioAID. As expected, the baseline performance  varies
widely since it depends on the selectivity of the query. The first set of 4 queries are highly selective 
( $<10$ node pairs), while the second set of 4 queries are not (around 100 node pairs). 
The baseline query time  increases from $0.4s$ for the first set  to $60s$ for the second set.
In contrast, the performance of  RPL and optRPL depend  solely on the
input list size (and query size), and  have respective query times of
$12s$ and $10s$ for all queries of this size, regardless of their
selectivity.\eat{ The same trend can also be observed from
  Fig.~\ref{fig:allpairs_selectivity_QBLast}.} Compared with
Fig.~\ref{fig:allpairs_selectivity_QBLast}, RPL is stable when query
size and input list size are fixed while optRPL achieves a major gain
over RPL. Note that query times of (opt)RPL for lowly selective queries
on BioAID and QBLast vary, around 60s and 14s respectively; that is
because BioAID is deeper so that IFQs match more node pairs.

Fig.~\ref{fig:allpairs_astar_BioAID} reports the query time for $a^*$ on BioAID. 
The query time of the baseline increases dramatically 
from $44ms$ to $7.8s$  as the run size grows from $1K$ to
$16K$\eat{, since it takes up to **** rounds to reach the fixpoint}.  
In contrast, RPL and optRPL increase slowly from $22ms$ to
$4.7s$, reducing the query time by an order of magnitude. The same trend can be observed in Fig.~\ref{fig:allpairs_astar_QBLast} for QBLast. Note that optRPL shows limited improvement over RPL because a relatively small number of unreachable node pairs are filtered out.

\eat{ eat for now
We now show how regular path labels can be used to help
answer all-pairs safe queries. 
We start with a comparison between using and
not using regular path labels, illustrated using
the two query classes,  IFQs  (which favors existing approaches)  and Kleene star (which favors our
approach). 
The best existing approach for IFQs  (Option G3 in Section~\ref{sec:general}) uses (I)ndex and (R)eachability
(L)abels ({\bf IRL} in Fig.~\ref{fig:ex_IFQs}), for which we use the reachability labeling scheme of
\cite{DBLP:conf/sigmod/BaoDM12}. Note that we do not precompute
transitive closures.
The best existing approach for
Kleene star uses either depth-first search (Option G1 in Section~\ref{sec:general}) or
structural joins (Option G2 in Section~\ref{sec:general}).  We omit results of Option G1 since it is
very slow, and refer to Option G2 as {\bf Non-RPL}. We compare with our 
Algorithm~\ref{algo:answerAllPairsRegularQuery}, i.e. {\bf FRPL} in Fig.~\ref{fig:ex_safe} (pre(F)iltering +
(R)egular (P)ath (L)abels). We then show that part of the performance gain of FRPL comes from the
prefilering step.\eat{(Section~\ref{subsubsec:FRPL}).}  In
Fig.~\ref{fig:ex_FRPL},  {\bf NRPL} stands for 
(N)ested-loop joins over regular path labels (Option S1 in Section~\ref{sec:our-approach}).

%\subsubsection{FRPL vs IRL and Non-RPL}\label{subsubsec:FRPLvsNonRPL}
{\bf FRPL vs IRL and Non-RPL.}
Fig.~\ref{fig:ex_IFQs} reports the query time using IRL and FRPL for
$8$ IFQs of size 3  \\ ($R=(\_)^*a_1(\_)^*a_2(\_)^*a_3(\_)^*$)
over runs of size $2K$ (the
input lists $l_1$, $l_2$ are the list of all nodes). As expected, the performance of IRL varies
widely since it depends on the selectivity of the query. The first set of 4 queries are highly selective 
( $<10$ node pairs), while the second set of 4 queries are not (around 100 node pairs). 
The query time of IRL increases from $0.2s$ for the first set  to $42s$ for the second set.
In contrast, the performance of  FRPL depends solely on the input list size (and query size). 
As shown in Fig.~\ref{fig:ex_IFQs}, FRPL has a query time of $5s$ for all queries of this size, regardless of their selectivity. 
\eat{\scream{Xiaocheng:  Perhaps move this to the summary -- By using
statistics of data, we can probably decide which approach to choose ahead of
time. We leave it as future work. }}
It is worth noting that our linear-time algorithm for all-pairs
reachability queries could also be used for IFQs, and that it outperforms IRL by
more than 20\%. (Details are omitted.)

%\subsubsection{FRPL vs Non\text{-}RPL}\label{subsubsec:FRPLvsNonRPL}
\eat{In this subsection, we show the power of our regular path
labels. Since all existing work cannot avoid generating huge amount of
intermediate
results whereas our approach solely depends on the input list size. 
It is not hard to imagine that our approach is extremely good for
Kleene star.}

Fig.~\ref{fig:ex_FRPLvsNonRPL} compares FRPL and Non-RPL for the query $a^*$. 
%As seen from Fig.~\ref{fig:ex_FRPLvsNonRPL}, 
The query time of Non-RPL increases dramatically (exponentially)
from $2ms$ to $47ms$  as the result size of $a$ grows from $20$ to
$200$, since it takes up to 10 rounds to reach the fixpoint.  In contrast, FRPL increases slowly from $2ms$ to
$4ms$, reducing the query time by an order of magnitude.

%\subsubsection{FRPL vs NRPL}\label{subsubsec:FRPL}
{\bf FRPL vs NRPL.}
We now show that the performance gain of FRPL partly results from the
prefiltering step. Recall that both FRPL and NRPL use regular path
labels. Given two lists of nodes $l_1,l_2$, NRPL performs a nested-loop
join over $l_1$ and $l_2$  and checks  regular path labels for each node pair. 
FRPL filters out unreachable node pairs first, and then uses nested-loops to check
regular path labels for the remaining node pairs.

 Fig.~\ref{fig:ex_FRPL} reports the query time for NRPL and FRPL. 
First, both approaches are almost quadratic in the input list size (note that the x-axis is
$|l_1|*|l_2|$), independent of the run size. Second, both approaches are very
fast, less than $200ms$ when $|l_1|*|l_2|<10K$, since answering
pairwise queries is fast.  % (Section~\ref{subsec:ex_pairwise}). 
Moreover, FRPL outperforms NRPL by almost $20\%$. 
}

\subsection{Performance of General Queries}
\label{subsec:exp_general}
Finally, we evaluate the performance of regular path labels
when used in general queries. We compare our approach
optRPL with the baseline Option G1 on queries randomly generated by
combining IFQs, Kleene stars and edge labels. We observed that most of
the queries are safe. In this subsection, we only report the
performance of $40$ unsafe queries. 
As expected, our technique significantly speeds up the performance of queries
that generate massive intermediate results due to lowly selective components.
For BioAID, 75\% (31/40) of the unsafe queries were of this form; the improvement over the baseline of these queries is shown in 
Fig.~\ref{fig:ex_general_BioAID}.  Note that over 60\% (19/31) of these queries show significant improvement  ($>40\%$). 
For QBLast, over 25\% (13/40) of the unsafe queries were of this form, and the improvement is shown in 
Fig.~\ref{fig:ex_general_QBLAST}. 
We therefore conclude that our technique could be a very useful component in a cost-based query optimizer that uses statistical
information to choose the right query plan and would significantly reduce the evaluation cost of lowly selective subqueries.

%Our approach works well for queries of the form $R=R_1|R_2|\ldots |R_k$ where $R_i$ is a safe subqueries.

%\vspace{-5mm}
\begin{figure}
\centering
\subfloat[BioAID]{\label{fig:ex_general_BioAID}\includegraphics[scale=0.4]{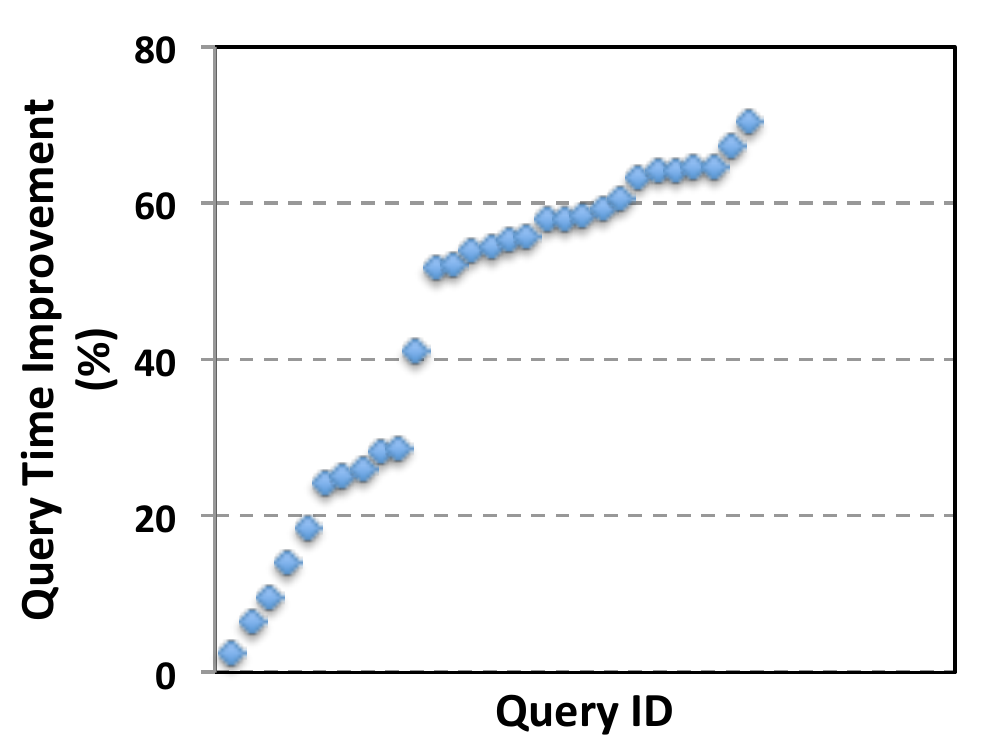}}
~
\subfloat[QBLast]{\label{fig:ex_general_QBLAST}\includegraphics[scale=0.4]{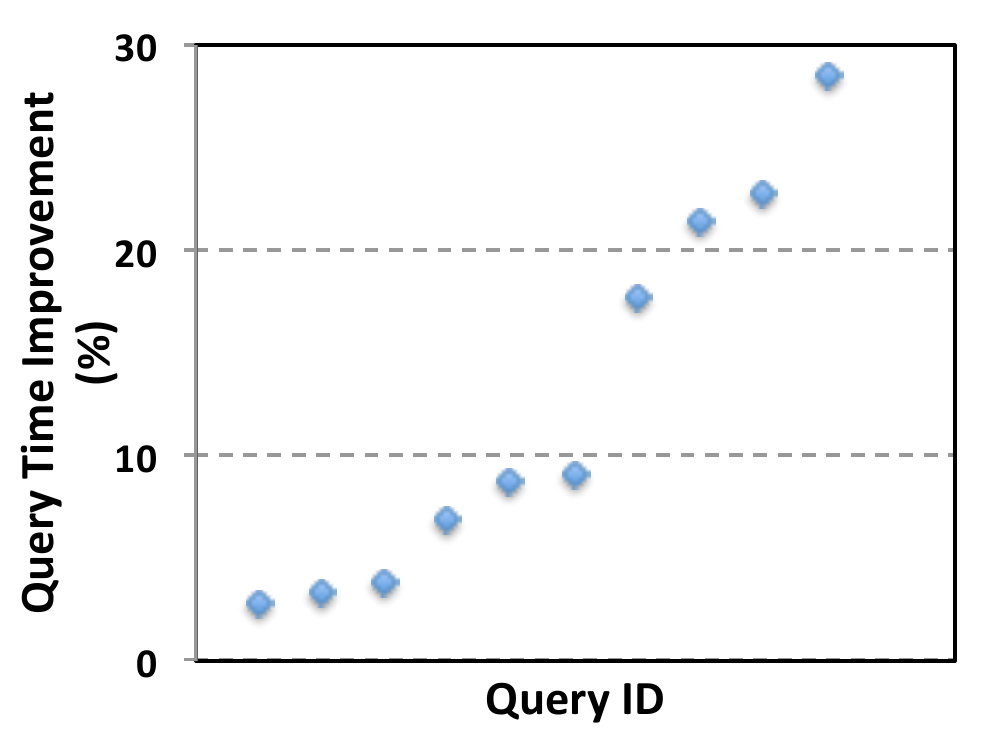}}
\vspace{-0.2cm}
\caption{Improvement of optRPL for general queries}
\vspace{-1mm}
\label{fig:ex_safe}
\end{figure}

%%%%%%%%%%%%%
\eat{ eat for now
Finally, we evaluate the performance of regular path labels
when used in general queries. We compare our approach
(RPL, Section~\ref{sec:general}) with the baseline (Option G2 in Section~\ref{sec:general},
Option G1 is too slow and is omitted).  
\eat{We generate queries by randomly combining edge tags using concatenation, union, and Kleene
star, varying the number of safe subqueries.  Each query has
$10$ subqueries in total while the number of safe subqueries varies from
$2$ to $8$. }
We generate a query with $2 \leq n \leq 8$ safe subqueries by randomly choosing $n$ of a set of pre-constructed safe queries, and combining them with $10-n$ random edge labels and operators (concatenation, alternation, Kleene star).  
\eat{Let n be the x-axis of Figure 16 (# of safe subqueries), i.e n=2,4,6,8. Take n=2 for example.
I have a set of safe queries the size of which vary. I randomly generate 30 queries each of which is constructed by combing n=2 random safe queries and 10-n=8 random edge labels using random operators (concatenation, alternation, Kleene star).  The performance of queries when n=2 is the average of these 30 queries, i.e. the leftmost bar of Figure 16.
Similarly, when n=4, I generate each query by combining n=4 safe queries and 10-n=6 random edge labels using random operators. 
So neither n=1 nor n=10 can guarantee a safe query.}
\eat{
 \begin{figure}
\centering
\includegraphics[scale=0.45]{figs/impvRPL.pdf}
\vspace{-1mm}
\caption{Improvement of RPL for General Queries}
\vspace{-5mm}
\label{fig:ex_imprvRPL}
\end{figure}}

\eat{
{\bf Implementation and Observation}  
Our main observation is that when dealing with general queries, it may
not be better to use FRPL all the time, which we discuss
next. Recall that we consider three operators in this paper,
concatenation, union and Kleene Star. Concatenation and union are cheap
operators.
If the two input node lists are sorted,
concatenation can be done using merge join in time linear in input
size. Union can be done using hash set structure in Java and therefore can be
done almost in constant time.  In our experiments, queries without Kleene Star can be
answered in less than 20ms (when queries have <10 operators in total)
by the baseline. In contrast, FRPL is
quadratic in input size. Seen from Fig.~\ref{fig:ex_FRPL}, when each
of the two input list has 200 nodes, the running time of FRPL is already
50ms. Note that intermediate results of a general query are usually
hundreds even thousands. Thus, in our experiment, for concatenation and
union, we choose not to apply FRPL. However, Kleene star is an expensive
operator as has shown in Fig.~\ref{fig:ex_FRPLvsNonRPL}. Thus we
apply FRPL for safe subqueries that are Kleene star expressions.

Fig.~\ref{fig:ex_RPL} reports the query time of all-pairs general
queries for RPL and baseline. Each data point here is average of 30
randomly generated queries each of which has 10 subqueries in total.
 The queries are evaluated  over 10 runs of size 4K with two input
lists of 200 nodes. The x-axis is the number of {\bf safe} subqueries, which are {\bf all
Kleene star} expressions. The first two points show that RPL
is no better than baseline. We observe that most Kleene star expressions are
safe. That's why the first two points have little query time. Thus the
overhead of checking safety becomes the main cost. As queries become
complex (and have more Kleene star subexpressions), RPL starts to
win. }

Fig.~\ref{fig:ex_imprvRPL} reports the improvement of RPL over
the baseline. % among $30$ queries each. 
The leftmost bar is the number of
queries (out of $30$) for which RPL wins. As the number of safe subqueries
grows beyond $4$, more than $25$ ($83\%$) queries show an
improvement. The middle bar shows the number of queries for which RPL wins
by at least $30\%$. Nearly $67\%$ of queries show a $30\%$ improvement. The
rightmost bar shows that for some queries, RPL achieves significant
improvement. We also observe in our experiments that when the query
contain few safe subqueries ($<2$), RPL is a bit slower, $18ms$ on
average, since RPL spends
time trying to find safe subqueries. When the query contain $8$ safe
subqueries (out of $10$), RPL brings down the average query time from
$126ms$ to $55ms$.}

\eat{\subsection{Summary}
To summarize, our approach (opt)RPL is best for queries that generate
massive intermediate results, e.g. low selectivity IFQs and
Kleene star. In contrast, for queries that are
highly selective or have few intermediate results, it is better to use the
baseline. An interesting future direction is to build a cost
model using data statistics to pick the right approach.}

\section{Conclusions}\label{sec:conclu}
This paper considers the problem of answering regular path queries over
workflow provenance graphs (executions). 
The approach assumes that the execution has been labeled with derivation-based {\em reachability} labels~\cite{DBLP:conf/sigmod/BaoDM12},
and shows how to use them to answer {\em regular path queries}.
For this we identify a core property of a query w.r.t. a specification,
{\em safe query}, which allows reachability labels 
to be used in conjunction with the query-intersected specification to answer a pairwise regular path query in constant time.
The reason that this works is that the reachability  labels of \cite{DBLP:conf/sigmod/BaoDM12}, unlike other labeling techniques,
are parameterized by the specification. 
Building on this, we develop efficient algorithms to answer all-pairs safe/general queries. Experimental
results demonstrate the advantage of our approach, especially for
queries which  generate large intermediate results, e.g. Kleene
star. Future work includes 1) building a cost model to predict the
intermediate result size so as to optimize the query process; and 2) query
rewriting, taking the  workflow specification into account.

\section{Acknowledgement}
We thank the authors of \cite{DBLP:conf/ssdbm/KoschmiederL12} for
  providing their code.

\small
\balance
\bibliographystyle{abbrv}
\bibliography{rpl}
\end{document}